\documentclass[a4paper,12pt]{article}
\usepackage[bbgreekl]{mathbbol}
\usepackage[small]{titlesec}
\usepackage[utf8]{inputenc}
\usepackage[T1]{fontenc} 
\usepackage[]{mdframed}
\usepackage{lineno}
\usepackage{pgfplots}
\pgfplotsset{compat=1.18}
\usepackage[backend=biber,style=alphabetic,natbib=false,giveninits=true,doi=true,isbn=false,url=false,date=year,maxbibnames=99,sorting=nyt,defernumbers=true]{biblatex}
\DeclareNameAlias{default}{family-given}

\DeclareFieldFormat[article]{title}{#1}
\renewbibmacro{in:}{}
\DeclareFieldFormat[article]{pages}{#1}
\DeclareFieldFormat{journal}{#1}
\renewcommand\bf\bfseries

\renewbibmacro*{volume+number+eid}{%
	\printfield{volume}%
	\setunit*{\addnbspace}% NEW (optional); there's also \addnbthinspace
	\printfield{number}%
	\setunit{\addcomma\space}%
	\printfield{eid}}
\DeclareFieldFormat[article]{number}{\mkbibparens{#1}}
\DeclareFieldFormat[article]{volume}{\textbf{#1}}
\DeclareFieldFormat{year}{\mkbibparens{#1}}
\DeclareBibliographyDriver{article}{%
	\printnames{author}:%
	\newunit\newblock
	\newunit\newblock
    \printfield[titlecase]{title}
	\printfield{journaltitle}
	\newunit
	\iffieldundef{number}{\printfield{volume}}{\printfield{volume}\addspace\printfield{number}}
	\addcomma\addspace\printfield{pages}\addspace
	\printfield{year}
}
\AtEveryBibitem{\clearfield{month}}
\AtEveryBibitem{\clearfield{day}}
\usepackage[english]{babel}
\usepackage[T1]{fontenc}
\usepackage{csquotes}
\usepackage{bbm}
\usepackage[leqno]{amsmath}
\usepackage{esint}
\usepackage{amsfonts,amsthm,amsbsy,amssymb,dsfont,stmaryrd}
\usepackage{aliascnt}
\usepackage[dvipsnames]{xcolor}
\usepackage{braket}

\makeatletter
\newcommand{\leqnomode}{\tagsleft@true\let\veqno\@@leqno}
\newcommand{\reqnomode}{\tagsleft@false\let\veqno\@@eqno}
\makeatother

\usepackage{mathtools}
\usepackage[makeroom]{cancel}
\usepackage[normalem]{ulem}
\numberwithin{equation}{section}

\newcommand\myshade{85}
\colorlet{mylinkcolor}{violet}
\colorlet{mycitecolor}{YellowOrange}
\colorlet{myurlcolor}{Aquamarine}

\usepackage[left=2.5cm,right=2.5cm,top=2.5cm,bottom=2.5cm]{geometry}
\usepackage[unicode=true,pdfusetitle,
bookmarks=true,bookmarksnumbered=false,bookmarksopen=false,
breaklinks=false,pdfborder={0 0 1},backref=false,
linkcolor = mylinkcolor!\myshade!black,
citecolor = mycitecolor!\myshade!black,
urlcolor  = myurlcolor!\myshade!black,
colorlinks = true,
]
{hyperref}
\usepackage[nameinlink]{cleveref}

\usepackage{tikz}
\usetikzlibrary{decorations.pathreplacing}
\usetikzlibrary{arrows}
\usetikzlibrary{arrows.meta}
\usetikzlibrary{intersections}
\usetikzlibrary{calc}
\usepackage{graphicx}
\usepackage{caption}
\usepackage{slashed}
\definecolor{ct_black}{HTML}{000000}
\definecolor{ct_orange}{HTML}{ED872D}
\definecolor{ct_purple}{HTML}{7A68A6}
\definecolor{ct_blue}{HTML}{348ABD}
\definecolor{ct_turquoise}{HTML}{188487}
\definecolor{ct_red}{HTML}{E32636}
\definecolor{ct_pink}{HTML}{CF4457}
\definecolor{ct_green}{HTML}{467821}

\definecolor{ct2_green}{HTML}{9FF781}
\definecolor{ct2_green_dark}{HTML}{088A08}
\newif\ifshowchanges
\showchangesfalse
\DeclareRobustCommand{\added}[1]{\ifshowchanges{\textcolor{ct_green}{#1}}\else#1\fi}
\DeclareRobustCommand{\deleted}[1]{\ifshowchanges{\textcolor{red}{\sout{#1}}}\else\fi}
\newcommand{\addedm}[1]{\ifshowchanges{\color{ct_green}#1}\else#1\fi}
\newcommand{\deletedm}[1]{\ifshowchanges{\color{red}\cancel{#1}}\else\fi}
\DeclareRobustCommand{\replaced}[2]{\deleted{#1}\added{#2}}
\newcommand{\replacedm}[2]{\deletedm{#1}\addedm{#2}}

\theoremstyle{plain}
\newtheorem{thm}{\protect\theoremname}[section]

\newaliascnt{lem}{thm}
\newtheorem{lem}[lem]{\protect\lemmaname}
\aliascntresetthe{lem}

\newaliascnt{cor}{thm}
\newtheorem{cor}[cor]{\protect\corollaryname}
\aliascntresetthe{cor}

\newaliascnt{prop}{thm}
\newtheorem{prop}[prop]{\protect\propositionname}
\aliascntresetthe{prop}

\newaliascnt{claim}{thm}

\aliascntresetthe{claim}

\newaliascnt{assumption}{thm}

\aliascntresetthe{assumption}

\newaliascnt{convention}{thm}

\aliascntresetthe{convention}

\theoremstyle{remark}
\newaliascnt{rem}{thm}
\newtheorem{rem}[rem]{\protect\remarkname}
\aliascntresetthe{rem}

\theoremstyle{definition}
\newaliascnt{defn}{thm}
\newtheorem{defn}[defn]{\protect\definitionname}
\aliascntresetthe{defn}

\theoremstyle{plain}
\newaliascnt{example}{thm}

\aliascntresetthe{example}

\providecommand{\assumptionname}{Assumption}
\providecommand{\conventionname}{Convention}
\providecommand{\claimname}{Claim}
\providecommand{\corollaryname}{Corollary}
\providecommand{\definitionname}{Definition}
\providecommand{\lemmaname}{Lemma}
\providecommand{\propositionname}{Proposition}
\providecommand{\remarkname}{Remark}
\providecommand{\theoremname}{Theorem}
\providecommand{\examplename}{Example}

\crefname{section}{Section}{Sections}
\crefname{appendix}{Appendix}{Appendices}
\crefname{figure}{Figure}{Figures}
\crefname{assumption}{Assumption}{Assumptions}
\crefname{thm}{Theorem}{Theorems}
\crefname{lem}{Lemma}{Lemmas}
\crefname{rem}{Remark}{Remarks}
\crefformat{equation}{(#2#1#3)}
\crefname{table}{Table}{Tables}
\crefname{prop}{Proposition}{Propositions}
\Crefname{prop}{Proposition}{Propositions}

\crefrangelabelformat{equation}{(#3#1#4--#5#2#6)}

\crefmultiformat{equation}{(#2#1#3}{, #2#1#3)}{#2#1#3}{#2#1#3}
\Crefmultiformat{equation}{(#2#1#3}{, #2#1#3)}{#2#1#3}{#2#1#3}

\newtheorem*{lem*}{\protect\lemmaname}

\newcommand{\ee}{\operatorname{e}}
\newcommand{\ii}{\operatorname{i}}

\newcommand{\ZZ}{\mathbb{Z}}

\newcommand{\NN}{\mathbb{N}}
\newcommand{\RR}{\mathbb{R}}
\newcommand{\CC}{\mathbb{C}}

\newcommand{\calB}{\mathcal{B}}

\newcommand{\calE}{\mathcal{E}}

\newcommand{\calO}{\mathcal{O}}

\newcommand{\calV}{\mathcal{V}}

\newcommand{\calL}{\mathcal{L}}
\newcommand{\Ord}[1]{\calO\left(#1\right)}

\newcommand{\Op}[1]{\operatorname{Op}\br{#1}}
\newcommand\norm[1]{\left\lVert#1\right\rVert}

\newcommand\abs[1]{\left|#1\right|}
\newcommand\mLog[1]{-\log\left(\abs{#1}\right)}

\newcommand{\ip}[2]{\langle #1, #2 \rangle}

\newcommand{\dif}{\operatorname{d}}
\newcommand{\tr}{\operatorname{tr}}
\newcommand{\Hessian}[1]{\operatorname{Hessian}#1}
\usepackage{graphicx}

\newcommand{\CharFun}{%
  \mathchoice
    {\raisebox{0.18ex}{\scalebox{1.15}{$\displaystyle\chi$}}}%
    {\raisebox{0.16ex}{\scalebox{1.15}{$\textstyle\chi$}}}%
    {\raisebox{0.10ex}{\scalebox{1.05}{$\scriptstyle\chi$}}}%
    {\raisebox{0.06ex}{\scalebox{1.00}{$\scriptscriptstyle\chi$}}}%
}

\usepackage{bm}
\renewcommand{\Im}[1]{\operatorname{\mathbb{I}\mathbbm{m}}\left\{#1\right\}}
\renewcommand{\Re}[1]{\operatorname{\mathbb{R}\mathbbm{e}}\left\{#1\right\}}
\newcommand{\ve}{\varepsilon}
\newcommand{\vf}{\varphi}
\newcommand{\Id}{\mathds{1}}
\newcommand{\HH}{\mathbb{H}}

\newcommand{\Open}[1]{\mathrm{Open}(#1)}

\newcommand{\dist}{\mathrm{dist}}

\newcommand{\sgn}{\operatorname{sgn}}

\newcommand{\supp}{\operatorname{supp}}
\renewcommand{\vec}[1]{\mathbf{#1}}

\newcommand{\nc}{\normalcolor}

\usepackage{environ}

\NewEnviron{malign}{%
	\begin{align}\begin{split}
			\BODY
	\end{split}\end{align}
}

\newcommand{\eq}[1]{\begin{align*}#1\end{align*}}
\newcommand{\eql}[1]{\begin{align}#1\end{align}}

\newcommand{\DD}{\mathbb{D}}

\newcommand{\br}[1]{\left(#1\right)}

\usepackage{lmodern} % scalable fonts
\usepackage{titling}
\title{Magnetic Double-Wells: \\Lower Bounds on Tunneling\\$\,$\\ {\fontsize{14.4pt}{16pt}\selectfont with an Appendix by Tal Shpigel}}
\author{\href{mailto:cf@math.princeton.edu}{Charles L. Fefferman}, \href{mailto:jacobshapiro@princeton.edu}{Jacob Shapiro}\\
	{\footnotesize Department of Mathematics, Princeton University}\\
	 \href{mailto:miw2103@columbia.edu}{Michael I. Weinstein}\\
		\footnotesize{Department of Applied Physics and Applied Mathematics,}\\\footnotesize{and Department of Mathematics, Columbia University}
}

\begin{document}

	\reqnomode
	
	\maketitle
	\begin{abstract}
		
        We study double-well systems with strong magnetic fields and deep potential wells. We present lower bounds on tunneling rates for generic values of the coupling constant. This result was recently announced in \cite{FSW24} and complements our recent counter-example construction \cite{fefferman2025magneticdoublewellsabsencetunneling} which exhibits vanishing tunneling for specially-constructed double-well potentials.\\ {\ }\\
	\end{abstract}
	\tableofcontents

\section{Introduction}

Quantum tunneling in double-well systems is a cornerstone phenomenon in mathematical physics.
In the absence of an external magnetic field, a particle which is initially  localized in one of two sufficiently deep or well-separated potential wells has a nonzero probability of tunneling to the other well. The tunneling rate is proportional to the difference between the lowest two eigenvalues of the double-well Hamiltonian, a quantity
which is (typically) exponentially small if  the well-depth or well-separation \replaced{are}{is} large. The tunneling time is its reciprocal.
This paradigm, developed both heuristically and rigorously, goes back to the classical literature and underlies a great variety of effects in physics and chemistry; see, for example, \cite{Landau_Lifshitz_vol_3} in physics or \cite{Simon_1984_10.2307/2007072,Helffer_Sjostrand_1984} in mathematics. The latter two works derive a lower bound on the eigenvalue splitting (or tunneling rate) assuming the single-well potential has a single non-degenerate minimum. In \cite{FLW17_doi:10.1002/cpa.21735} lower bounds are derived assuming rather that the single-well is compactly supported. 

In contrast, adding a \emph{constant magnetic field} fundamentally alters the picture.
 Early works on the magnetic case considered \deleted{the case of }weak magnetic fields, and confirmed a small correction to the non-magnetic case \cite{Helffer_Sjostrand_1987_magnetic_ASNSP_1987_4_14_4_625_0}. For strong magnetic fields, magnetic translations \cite{Zak_1964_PhysRev.134.A1602}  endow localized states with nontrivial complex phases, yielding interference among tunneling paths. For example in \cite{FSW_22_doi:10.1137/21M1429412} we showed that for double-well magnetic systems whose single-well is radial, the magnetic field effect leads to an exponentially smaller tunneling amplitude. However, such systems with radial single-well potential were shown to still possess a strictly positive lower bound on the tunneling rate \cite{FSW_22_doi:10.1137/21M1429412,HelfferKachmar2024,Morin2024}. 
 The tunneling rate lower bound in this case was used in a derivation of tight binding models in the regime of strong binding and strong magnetic fields in models of  (translation invariant or disordered) crystals \cite{ShapWein22}. 

Recently, we announced in \cite{FSW24} and then provided a complete proof in  \cite{fefferman2025magneticdoublewellsabsencetunneling},  that the magnetic setting admits a phenomenon that has no analogue in the non-magnetic case: \emph{exact vanishing of tunneling} in a double-well system built from a suitable non-radial single-well potential.  
We constructed a family of (non-radial) single-well potentials such that, in a strong constant magnetic field and for large well depth, the associated symmetric double-well has zero eigenvalue splitting between its two lowest energy levels, so that tunneling is entirely eliminated. Moreover, by varying parameters within this family one can flip the parity of the ground state (from even to odd), 
via a coalescence and re-emergence of two distinct eigenvalues. A discussion of implications for potential applications to quantum materials is also discussed in \cite{FSW24}.

In this paper, we prove a second result announced in \cite{FSW24}:\\
{\it The vanishing of the tunneling is a non-generic phenomenon;  for couplings, $\lambda>0$, outside of a set of density zero, 
there is a lower bound on the eigenvalue splitting (upper bound on the tunneling time)}.

\subsection{Framework and statement of the main result}
Let us begin  by introducing a precise framework. We work with the Hilbert space $L^2(\RR^2)$.  On $L^2(\RR^2)$, we consider magnetic Schr\"odinger operators of the form
\eql{\label{eq:general form of magnetic schroedinger operators we consider}
  \br{P-\frac12b\lambda X^\perp}^2+\lambda^2 V(X)\ ,
} 
a perturbation by an electric potential, $V(X)$, of the 
 Landau Hamiltonian
 \eql{\label{eq:HLandau}
 H^{\rm Landau}_\lambda := \br{P-\frac12b\lambda X^\perp}^2.
 }
 Here, $P\equiv-\ii\nabla$ is the momentum operator, $X$ is the position operator, with \eq{X^\perp\equiv(-X_2,X_1)\, .} 
 The parameter $\lambda>0$ is a sufficiently large coupling constant which controls simultaneously the scaling of the magnetic field strength (which grows as $\lambda$), and the depth of the potential $\lambda^2V:\RR^2\to(-\infty,0]$.
   The parameter $b>0$ controls the relative strength of the magnetic field $b\lambda$ to the well-depth $\lambda^2\min V$.

  Such Hamiltonians describe the dynamics of an electron bound to a two-dimensional plane and subject to a constant perpendicular magnetic field.
 The choice of $\lambda$-dependence of the electric and magnetic potentials in \cref{eq:general form of magnetic schroedinger operators we consider} ensures a balance
of electric and magnetic effects for one-well potentials with a single nondegenerate minimum. 
Specifically,  for large $\lambda$, the Landau energy-level spacing, associated with $(P-\frac12b\lambda X^\perp)^2$,  is $O(\lambda)$
 and the harmonic oscillator energy-level spacing, associated with the harmonic approximation  of $v(X)$, $P^2 + \frac12\lambda^2 \left\langle X,\Hessian{v}(0) X\right\rangle$,  is  also $O(\lambda)$.
Further, the confinement within the wells of eigenstates in our double-well system is governed by a magnetic oscillator with comparable magnetic and harmonic contributions; see \cite{Nakamura_doi:10.1080/03605309608821214} for the harmonic approximation in this scaling. In fact, a scaling where the magnetic field tends to infinity is necessary to obtain topologically nontrivial tight-binding models \cite{Nakamura1990LowEnergyBands}.

 % We note that the particular asymptotic scaling of the magnetic and electric fields we choose is essential in order to get non-trivial results: otherwise one of the effects gets washed out by the other. as well as 

 We consider simultaneously two main operators of this type: the \emph{single-well} Hamiltonian $h_\lambda$ and \emph{double-well} Hamiltonian $H_\lambda$. They are defined as 
\eql{\label{eq:single-well Hamiltonian}
  h_{\lambda,b} := \br{P-\frac12b\lambda X^\perp}^2+\lambda^2 v(X)
} and 
\eql{\label{eq:double-well Hamiltonian}
  H_{\lambda,b} := \br{P-\frac12b\lambda X^\perp}^2+\lambda^2 v(X+d)+ \lambda^2 v(X-d)
} Here $v:\RR^2\to[-1,0]$ is a smooth compactly supported (say within $B_a(0)$ for some $a>0$) single-well potential and $d\equiv(d_1,0)\in\RR^2\setminus\Set{0}$ is the displacement of each well taken, without loss of generality, to lie along the $1$-axis.

We assume that $v$ is chosen so that when $\lambda$ is sufficiently large, $h_{\lambda,b}$ has ground state energy $e_{0,\lambda}=-\lambda^2+\mathrm{o}(\lambda^2)$ which has distance to the remainder of the spectrum with at least order of magnitude $c_{\mathrm{gap}}$, a constant which is independent of $\lambda$. It follows that $E_{0,\lambda}\leq E_{1,\lambda}$, the two lowest eigenvalues of $H_\lambda$, are also at a distance of at least order $c_{\rm gap}$ from the remainder of the spectrum of $H_\lambda$. We are chiefly concerned with the \emph{double-well eigenvalue splitting} defined by \eql{\label{eq:splitting}
  \Delta_0(\lambda,b) \equiv \Delta_0(\lambda) := E_{1,\lambda} - E_{0,\lambda}\,.
} 

An emergent quantity related to $\Delta_0(\lambda)$ is the \emph{hopping coefficient}, $\rho_0(\lambda,b)=\rho_0(\lambda)$, which  plays a role in the present analysis. It is given by
\eql{\label{eq:hopping coefficient}
   \rho_0(\lambda) &:= \ip{\widehat{R}^{-d}\vf_{0,\lambda}}{\br{H_\lambda-e_{0,\lambda}\Id}\widehat{R}^{d}\vf_{0,\lambda}}\\
  &= \lambda^2\int_{x\in\RR^2} \overline{\vf_{0,\lambda}(x+d)}v(x+d) \exp\br{\ii b\lambda d_1 x_2}\vf_{0,\lambda}(x-d)\dif{x}\nonumber
} Here $\vf_{0,\lambda}$ is the $L^2$-normalized ground state of $h_\lambda$, and  $\widehat{R}^{ d }$ denotes Zak's magnetic translation operator \cite{Zak_1964_PhysRev.134.A1602}:
 \eql{\label{eq:magnetic translations}\widehat{R}^z := \exp\br{-\ii z\cdot \br{P+\frac12b\lambda X^\perp}}\qquad(z\in\RR^2)\,.} Since $z\cdot x^\perp = -z^\perp\cdot x$,   \eq{
 (\widehat{R}^{ z }f)(x) \equiv \exp\br{\ii\frac{b\lambda}{2} x \cdot  z^\perp} f(x-z)\qquad(x,z\in\RR^2\,,f:\RR^2\to\CC)\,.
}
For a discussion of how the hopping coefficient arises  in magnetic 
double-well and crystalline systems, see, for example, \cite{FSW_22_doi:10.1137/21M1429412} and \cite{ShapWein22}.

Below we use the notion of \emph{a set of density zero}; we say that $S\subseteq(0,\infty)$ has \emph{zero density} iff \eq{
\lim_{M\to\infty}\frac{\mu\br{S\cap[0,M]}}{M} = 0
} where $\mu$ is the Lebesgue measure.

Our main theorem is

\begin{thm}\label{thm:main theorem}
   Assume that the real-valued potential $v$ satisfies the following conditions:
   (a) $v\in C^3(\mathbb R^2\to[-1,0])$.\\
   (b) $\supp(v)\subset B_a(0)$ for some $a>0$. \\ 
   (c) $v(0)=-1$, the origin is the unique global minimum of $v$, and $\Hessian{v}(0)>0$.
   
 \noindent   Let $\ve>0$. Assume that $2d_1>Ca$ for some $C$ sufficiently large but independent of $\lambda$, where $d=(d_1,0)\in\RR^2$ is the displacement parameter of the double-well; see \cref{eq:double-well Hamiltonian}.
    
    Then, there is a constant $b_\ve>0$ such that for all $0<b<b_\ve$, the following holds:
    \begin{enumerate}\item 
    There exists a set of density zero $Z(\ve,v)\subseteq(0,\infty)$ such that 
\eql{\label{eq:pointwise lower bound on rho}
        \abs{\rho_0(\lambda)} \geq \exp\br{-\lambda^{1+\ve}}\quad\textrm{for all}\quad \lambda\in(0,\infty)\setminus Z(\ve,v)\ ,
    } and 
    \eql{\label{eq:pointwise lower bound on Delta}
        \Delta_0(\lambda) \geq \exp\br{-\lambda^{1+\ve}}\quad\textrm{for all}\quad \lambda\in(0,\infty)\setminus Z(\ve,v)\ .
    }
    \item {\it Sparsity:} There exists an $\varepsilon$-dependent $\Lambda_\varepsilon<\infty$ and a discrete set $\Set{\lambda_k}_{k\ge1}\subset Z(\ve,v)$  such that
    \eql{\label{eq:density of zeros}
        \sum_{k=1}^\infty \lambda_k^{-\br{1+\ve}}<\infty\qquad(\ve>0)\,.
    }
     and such that $\rho_0(\lambda)$ and $\Delta_0(\lambda)$ are non-zero for  $\lambda\in \left[\Lambda_\varepsilon,\infty\right)\setminus\Set{\lambda_k}_{k\ge1}$.
   \item \emph{Averaged counterpart to item 1}: \\ For any $\ve>0$ there exist $\Lambda_\varepsilon$, $b_\ve>0$ and $C(\ve,d_1)<\infty$ such that for all $b\in(0,b_\ve)$  and for all $\lambda\geq\Lambda_\varepsilon$, 
    \eql{\label{eq:avg lower bound on rho}
        \frac{1}{\lambda}\int_{\lambda}^{2\lambda} -\log\br{\abs{\rho_0(\widetilde{\lambda})}}\dif{\widetilde{\lambda}} \leq C(\ve,d_1)\  \lambda^{1+\ve}
    } and 
    \eql{\label{eq:avg lower bound on Delta}
        \frac{1}{\lambda}\int_{\lambda}^{2\lambda}-\log\br{\abs{\Delta_0(\widetilde{\lambda})}}\dif{\widetilde{\lambda}} \leq C(\ve,d_1)\  \lambda^{1+\ve}\,.
    }
    \deleted{CF: check constant independence.}
    \end{enumerate}
\end{thm}

\subsection{Rough sketch of the proof}\label{sec:sketch}
\begin{itemize}
\item[(1)] The key is to prove that, up to a $\lambda$-dependent normalization factor which can be controlled for large real $\lambda$, the functions $\rho_0(\lambda,b)$ and  $\Delta_0(\lambda,b)$ are analytic in a region \eq{\Omega_{\Lambda,\varepsilon}:=\Set{z\in\mathbb{C}|\Lambda<\left|z\right|\quad{\rm and}\quad\left|\arg\left(z\right)\right|<\frac{\pi}{2}-\varepsilon}\,,}  which is contained within the open right half plane. Here $\ve>0$ is a fixed arbitrarily small parameter and $\Lambda>0$ is chosen sufficiently large (depending on $\varepsilon$), subject to a finite number of conditions to be specified below. Further, we show that the above functions are bounded from above
 on  $\Omega_{\Lambda,\ve}$ by $\lambda\mapsto\abs{\lambda}^{k}\exp(C b \abs{\lambda})$ for some universal $k,C$.
%\ni\lambda\mapsto\rho_0(\lambda,b)$ and  %$\Omega_{\Lambda,\ve}\ni\lambda\mapsto\Delta_0(\lambda,b)$ 
% are bounded above by .
%
\item[(2)]  Fix some $\lambda_\star\in \Omega_{\Lambda,\ve}\cap\RR$. Then, by lower bounds of hopping coefficients and the splitting in {\it non-magnetic systems}, \cite{Simon_1984_10.2307/2007072,Helffer_Sjostrand_1984,FLW17_doi:10.1002/cpa.21735}, we have strictly positive  lower 
bounds on $|\rho_0(\lambda_\star,b=0)|$ and 
 $|\Delta_0(\lambda_\star,b=0)|$.  Therefore, by continuity, we have strictly positive lower bounds on  
  $|\rho_0(\lambda_\star,b)|$ and 
 $|\Delta_0(\lambda_\star,b)|$ for all $0<b<b_\varepsilon$ with $b_\varepsilon$ sufficiently small (dependent on the choice of $\lambda_\star$). 
\item[(3)]  Given (1) and (2) we apply the complex analytic \cref{lem:complex-analysis} (below), together with bounds on the normalization factor for real $\lambda$, to establish that for fixed $b\in(0,b_\varepsilon)$, the hopping $|\rho_0(\lambda,b)|$ and  splitting $|\Delta_0(\lambda,b)|$  satisfy the desired lower bounds for all real $\lambda\in\Omega_{\Lambda,\ve}\cap\RR$ except possibly for $\lambda$ in an exceptional set of density zero. 
 \item[(4)] Hence the proof of \cref{thm:main theorem} boils down to establishing analytic continuation properties of  $\rho_0(\lambda)$ and  $\Delta_0(\lambda)$, as well as their upper bounds on $\Omega_{\Lambda,\ve}$.
%d
\item[(5)] {\it An obstruction to analyticity?} We first observe that for real $\lambda$ large, the ground state of $h_\lambda$, $\varphi_\lambda$, is well-approximated by the normalized ground state, $\varphi_\lambda^{\rm MHO}$, of the magnetic harmonic oscillator Hamiltonian. Therefore,  
 for all $\lambda$ real and sufficiently large, $\varphi_\lambda$, the ground state of $h_\lambda$  is given by a Riesz projection: a contour integral of  $\zeta\mapsto \big(h_\lambda-\zeta\Id\big)^{-1} \varphi_\lambda^{\rm MHO}$ (up to normalization). In the contour integral,  $\zeta$ varies over a circle centered at $e_\lambda$ of radius, say, half the distance to the second eigenvalue of $h_\lambda$. 
 
 Now, while $\lambda \mapsto \varphi_\lambda^{\rm MHO}$ has an analytic continuation for $\lambda\in\Omega_{\Lambda,\ve}$, the mapping 
  $\lambda\mapsto \big(h_\lambda-\zeta\Id\big)^{-1} \in\calB(L^2(\RR^2)) $,  does \emph{not}. Indeed, $h_\lambda$ is a relatively compact perturbation of $H^{\rm Landau}_\lambda$, and the resolvent of 
    the latter does not have an analytic continuation off the real axis; see \cite{krejcirik2024abruptchangesspectralaplacian}  and the discussion in \cref{prop:Landau resolvent does NOT continue}.
 \item[(6)] {\it The remedy:} We note that the above Riesz projection representation of $\varphi_\lambda$, up to normalization, works just as well if we replace $\varphi^{\rm MHO}_\lambda$ by $ \CharFun_B(X)\ \varphi^{\rm MHO}_\lambda$, where $\CharFun_B(X)$ denotes the projection to functions supported on a compact $B\subseteq\RR^2$, just as long as the resulting vector is not perpendicular to $\vf_\lambda$. Due to the explicit Gaussian form of $\vf_\lambda^{\rm MHO}$, non-orthogonality is guaranteed as long as $B$ contains any small disc centered at the origin. In this representation, the resolvent of $h_\lambda$ acts on a function with compact support. Analyticity of $\varphi_\lambda$ (up to normalization) then follows from a lemma in which we prove that the mapping \eq{
 \lambda\in \Omega_{\Lambda,\ve}\cap\RR\ \longmapsto\  \big(h_\lambda-\zeta\Id\big)^{-1}\CharFun_B(X)\in\calB\big(\CharFun_B(X) L^2(\RR^2)\to L^2(\RR^2)\big)
 } continues to an analytic function on $\Omega_{\Lambda,\ve}$ with exponential upper bounds.
\end{itemize}

We conclude this subsection with our key complex analysis lemma, which we use to produce lower bounds on $\rho_0(\lambda,b)$ and $\Delta_0(\lambda,b)$. Its statement as \cref{cor:main complex analytic lower bound} and proof, based on the Blaschke factorization of analytic functions,   are given in \cref{sec:bounds on analytic functions}. 
\begin{lem}[Complex analytic lemma]\label{lem:complex-analysis} Let $\alpha\in(0,1]$ and \eq{
        \Gamma_\alpha := \Set{z\in\CC:-\alpha\pi/2<\operatorname{Arg}(z)<\alpha\pi/2}\,.
    } Here $\operatorname{Arg}(z)$ refers to the principal value, taking values $-\pi <\operatorname{Arg}(z)\le \pi$. Let $F:\Gamma_\alpha\to\CC$ be analytic such that for some given $\beta\in(0,\infty)$ we have \eq{\abs{F(1)}\geq\ee^{-\beta}} and such that for some analytic \emph{nowhere-zero} function $U:\Gamma_\alpha\to\CC$ we have \eq{
            \abs{F(z)} \leq \abs{U(z)}\qquad(z\in\Gamma_\alpha)\,.
        }

        Then there is constant $C<\infty$ (independent of $\alpha,F$ and $\beta$) such that for all $R>2^\alpha$,
        \eq{
            \frac{1}{R}\int_{t=R}^{2 R}\mLog{F(t)}\dif{t}&\leq C\alpha 2^{1/\alpha}\left(\beta+\log\left(\abs{U(1)}\right)\right) R^{1/\alpha}+\\
            &\qquad+\frac{1}{R}\int_{t=R}^{2R}\mLog{U(t)}\dif{t}\,.
        } as well as
        \eq{
            \frac{1}{R}\int_{t=R}^{2R}\abs{\mLog{F(t)}+\log\left(\abs{U(t)}\right)-t^{1/\alpha}\mu_0}\dif{t}\  {=}\ o\left(R^{1/\alpha}\right)\quad \textrm{as $R\to\infty$, }
        } for some $\mu_0\geq0$.
        \end{lem}
%\overset{R\to\infty}

\subsection{Open questions}
We mention some natural questions to consider: 
\begin{enumerate}
    \item  Throughout our analysis, $\ve>0$ is a fixed and arbitrarily small independent parameter. Can one take $\ve\to0^+$ as $\lambda\to\infty$ with the goal of  obtaining the improved lower bound
    \eq{
    \Delta_0(\lambda) \gtrsim \exp\br{-c \lambda}.\qquad \textrm{for $\lambda\gg1$}\,\ ?
    } 
    \item Were we able to take $\varepsilon=0$ in the lower bounds for $\rho_0$, could one then extend the work of \cite{ShapWein22} on the tight binding reduction in strongly bound and strongly magnetic systems to generic double wells?
    
    \item Can we take the magnetic parameter, $b$,  to be large in \cref{thm:main theorem}? 
    Instead of assuming that $\abs{\arg(\lambda)}<\frac\pi2-\ve$ and $b\in(0,b_\ve)$, with $b_\ve$ small, let us suppose that $b\in(0,b_{\rm max})$ and $\abs{\arg(\lambda)} < \ve(b_{\rm max})$, with $b_{\rm max}$ large and $\ve(b_{\rm max})\ll 1$. 
 Then, we believe that our analytic continuation arguments would go through with minor changes. In place of upper bounds of the form $\exp\br{\abs{\lambda}^{1+\ve}}$ on the relevant analytic function, we would find instead upper bounds $\exp\br{\abs{\lambda}^{K(b_{\rm max})}}$ for a constant $K(b_{\rm max})\gg 1$. 
  Unless $\rho_0(\lambda)$ and $\Delta_0(\lambda)$ vanish identically for the given $b$, that would imply lower bounds of the form \eq{
        \abs{\rho_0(\lambda)},\abs{\Delta_0(\lambda)} > \exp\br{-\lambda^{K(b_{\rm max})}}
    } for $\lambda$ outside a set of density zero.  
    Joint analyticity in $(b,\lambda)$ for $\abs{\arg(\lambda)},\abs{\arg(b)}<\ve(b_{\rm max})$ would imply that $\rho_0(\lambda),\Delta_0(\lambda)$ vanish identically as a function of $\lambda$ at most for finitely many $b\in(0,b_{\rm max})$
    We haven't carried this out, and we do not know how to rule out the existence of $b$ for which, say, $\rho_0(\lambda)=0$ for all large $\lambda$.
    \item What is the precise dependence of the lower bounds of $\Delta_0$ and $\rho_0$ on the well-separation parameter, $d_1$, and what is the  leading term in their asymptotics in the generic case?
    \item Our results concern the splitting of
 the lowest two eigenvalues of the magnetic double well Hamiltonian. These eigenvalues perturb from and lie below the first Landau level.
 Do analogous results hold for higher Landau levels? 
 \item Can our results be extended to non-compactly supported single wells and systems with weak disorder,  and three dimensions? Can the regularity assumptions on $v$ be relaxed?

\end{enumerate}

%\footnote{\textcolor{red}{We prove that $\gamma^2\Delta(\lambda)^2$ is analytic. But taking into account the counterexamples, what about $E_1(\lambda)$ and $E_0(\lambda)$ individually ??}}

\subsection{Organization of this paper.}
\cref{sec:the_proof} develops the analytic setup for the hopping coefficient $\rho_0$ and the splitting $\Delta_0$, proves analyticity on wedge-type regions in $\mathbb C$, and reduces the double–well problem to a question concerning a $2\times 2$ matrix 
which depends analytically on the parameters $\lambda$ and $b$.   \cref{sec:bounds on analytic functions} 
develops a general complex-analytic result, in  particular the proof of  \cref{lem:complex-analysis}, our tool for proving the generic lower bounds of \cref{thm:main theorem}.  \cref{sec:mesoscopic annuli} deals with the ``mesoscopic annuli'' construction used for the parametrix presented in \cref{sec:the_proof}. Background on pseudodifferential operators, used in \cref{sec:the_proof},  is presented in \cref{sec:pseudo diff operators}, and \deleted{and }\added{the relevant partition of unity is given in }\cref{lem:partition of unity}. The Landau resolvent for complex magnetic fields and the obstruction to a global analytic continuation of the resolvent as an operator on $L^2(\RR^2)$, are treated in \cref{sec:The Landau Resolvent at complex magnetic fields}. Finally an independent appendix,  \cref{sec:MHO}, contributed by Tal Shpigel, establishes resolvent analyticity and off–diagonal decay for the anisotropic magnetic harmonic oscillator needed throughout.

\subsection{Notation and conventions.}
\begin{enumerate}
    \item $\CharFun_B$ denotes the indicator function of the set $B\subseteq\RR^2$. $\CharFun_B(X)$ is the projection operator on $L^2(\RR^2)$ onto the subspace of functions supported within $B$.
    \item ${\cal B}(X)$ is the space of bounded linear operators $X\to X$ for a Banach space $X$, and $\calB(X\to Y)$ denotes the space of bounded linear operators $X\to Y$.
    \item $h_{\lambda,b}=h_\lambda$, the single-well magnetic Hamiltonian; see \cref{eq:single-well Hamiltonian}.
    \item $H_{\lambda,b}=H_\lambda$, the double-well magnetic Hamiltonian; see \cref{eq:double-well Hamiltonian}.
    \item $\rho_0(\lambda,b) \equiv \rho_0(\lambda) \equiv  \rho(\lambda)$, the magnetic hopping coefficient; see \cref{eq:hopping coefficient}.
    \item  $\Delta_0(\lambda,b) \equiv \Delta_0(\lambda) \equiv \Delta(\lambda) \equiv E_1(\lambda)-E_0(\lambda)$, the magnetic double-well eigenvalue splitting; see \cref{eq:splitting}. 
    \item If $(x,p)\mapsto a(x,p)$ is a smooth symbol, we write $\Op{a}$ for the corresponding pseudo-differential
     operator; see \cref{eq:action of pseudo-diff op}.
     \item $\mathbb{S}^1\equiv\Set{z\in\CC|\abs{z}=1}$.
     \item \added{Inner products are conjugate-linear in the first argument and linear in the second.}
     \item Throughout, $\Lambda>0$ is a minimal sufficiently large threshold ensuring that for any $\abs{\lambda}\geq\Lambda$ our standing hypotheses hold. The constant $\Lambda=\Lambda_\ve$ is chosen subject to finitely many constraints. It is thus a finite constant dependent on $v,\ve$ and the other constants in our assumptions. The parameter  $b$ is an \emph{independent} parameter that encodes the relative strength of the magnetic field compared to the potential well depth. We assume that $b>0$ and is less than a small enough $\ve$ dependent  constant, but is independent of $\Lambda_\ve$ (and of $\lambda$ of course). In particular, we assume no lower bound on $b>0$. 
     
    It follows that an expression like $\exp\br{C b \abs{\lambda}}$ may well be of order $1$ and so may not dominate polynomial factors in $\abs{\lambda}$. Hence in principle one would have to keep track of all polynomial powers of $\abs{\lambda}$ in addition to exponential expressions, if they include in the exponent a constant proportional to $b$ (most of our upper bounds do). 

     Instead of keeping explicit track of \deleted{explicit }polynomial powers of $\abs{\lambda}$ throughout, we write 
     \eq{
        A \lesssim \abs{\lambda}^{\sharp}\exp\br{C b \abs{\lambda}}
     } to indicate
     \eq{
        A \leq \abs{\lambda}^k \exp\br{C b \abs{\lambda}} 
     } where $k$ is a universal constant and $C$ depends only on the potential $v$. In different occurrences, $k$ and $C$ may denote different constants. For instance,
     \eq{
        A \lesssim \abs{\lambda}^{\sharp}\exp\br{C b \abs{\lambda}} \ \textrm{and}\  B  \lesssim \abs{\lambda}^{\sharp}\exp\br{C b \abs{\lambda}} \Longrightarrow AB \lesssim \abs{\lambda}^{\sharp}\exp\br{C b \abs{\lambda}}\,.
     }\added{ }We follow the same convention when estimates involve $\Lambda$ instead of $\abs{\lambda}$.
\end{enumerate}
%We set $\hbar=2m=1$ and write $\langle x\rangle=(1+\norm{x}^2)^{1/2}$. The vector potential is taken in symmetric gauge, $A(x)=\tfrac{B}{2}(-x_2,x_1)$, and norms are in $L^2(\mathbb{R}^2)$ unless indicated otherwise. For $f\in L^2$ and $y\in\mathbb{R}^2$, $\tau_y f$ denotes the magnetic translation of $f$ by $y$ (see \cite{Zak1964}). Constants $C,c>0$ may change from line to line and dependence on parameters is indicated when relevant.
\subsection*{Acknowledgements}
 CF was supported in part by NSF grant DMS-1700180. JS was supported in part by NSF grant DMS-2510207. MIW was supported in part by NSF grants DMS-1908657, DMS-1937254, DMS-2510769  and Simons Foundation Math + X Investigator Award \# 376319 (MIW). Part of this research was carried out during the 2023-24 academic year, when MIW was a Visiting Member in the School of Mathematics - Institute of Advanced Study, Princeton, supported by the Charles Simonyi Endowment, and a Visiting Fellow in the Department of Mathematics at Princeton University.
The authors wish to thank Antonio Cordoba and David Huse for stimulating discussions. We thank P.A. Deift and J. Lu for their careful reading of the \cite{FSW24} manuscript, and insightful comments and questions.

\section{Proof of the main theorem, \cref{thm:main theorem}}\label{sec:the_proof}

In the introduction, we used the zero subscript to indicate an object's association with  the \emph{ground state} of $h_{\lambda,b}$ or of $H_{\lambda,b}$.
To lighten up the notation, when there is no ambiguity we shall typically drop the zero subscript. Thus, we often replace $\rho_0(\lambda)$ by $\rho(\lambda)$ and  $\vf_{0,\lambda}$ by $\vf_\lambda$ etc. Similarly, when convenient we suppress the dependence on the magnetic parameter, $b$, as it is not an asymptotic parameter.

\subsection{Setup and strategy for analytic continuation argument and a lower bound for the hopping coefficient $\lambda\mapsto\rho_0(\lambda)$}\label{sec:setup-strategy}

To prove the main theorem on lower bounds for $\rho(\lambda)$ and $\Delta(\lambda)$, we shall apply the complex analytic \cref{lem:complex-analysis}, whose inputs are: an analytic function on a wedge which (a) satisfies a lower bound at a \emph{single point} in this wedge, and (b) an exponential upper bound throughout the wedge. 

We let $\ve\in(0,\frac\pi2)$ (so $\pi-2\ve$ is the aperture of the wedge), and choose $\Lambda$ to exceed some positive constant which depends on $\varepsilon$. We then consider the  unbounded and truncated wedge-shaped  regions: 
\begin{subequations}
\label{eq:Omega-annulus} 
\begin{align}
\Omega_{\Lambda,\varepsilon}&:=\Set{z\in\mathbb{C}|\Lambda<\left|z\right|\quad{\rm and}\quad\left|\arg\left(z\right)\right|<\frac{\pi}{2}-\varepsilon},\quad {\rm and}\\
\Omega_{\Lambda,2\Lambda,\varepsilon} &:=\Set{z\in\mathbb{C}|\Lambda<\left|z\right|<2\Lambda\quad{\rm and}\quad\left|\arg\left(z\right)\right|<\frac{\pi}{2}-\varepsilon}. 
\end{align}
\end{subequations} The latter set is where intermediate estimates will be derived, and is introduced for technical convenience.
We introduce the following representation of $\rho(\lambda)$: 
\eql{\label{eq:rMHO-ext}
\rho(\lambda) &:= \ip{\widehat{R}^{-d}\vf_{_{\overline{\lambda}\nc}}}{\br{H_\lambda-e_{\lambda}\Id}\widehat{R}^{d}\vf_{\lambda}}\notag\\
&= \lambda^2\int_{x\in\RR^2} \overline{\vf_{_{\overline{\lambda}\nc}}(x+d)}v(x+d)
   \exp\br{\ii b\lambda d_1 x_2}\vf_{\lambda}(x-d)\dif{x}\, ,
}
which  agrees with \cref{eq:hopping coefficient} for $\lambda\in\RR$, and which is more amenable to analytic extension (note the $\overline{\lambda}$ in the subscript). 
We have nothing to say about whether $\vf_{\lambda}$,  $\rho(\lambda)$ and $\Delta(\lambda)$ admit analytic continuations to $\Omega_{\Lambda,\varepsilon}$\color{black}. However, when these quantities are multiplied by appropriate $\lambda$-dependent factors, they do indeed have analytic continuations, and we can control those factors for large real values of $\lambda$. We now explain.

Assume $\lambda\in\RR$. Our study of the analyticity properties of $\vf_\lambda$ can be transferred to the  analyticity properties of the resolvent of $h_\lambda$ by introducing the Riesz projection:
 \eql{\label{eq:Riesz projection}
\widetilde\vf_{\lambda} 
&= \frac{\ii}{2\pi}\oint_{\Gamma_\lambda} \br{h_\lambda-z\Id}^{-1} \dif{z}\ \Phi_\lambda=
\frac{\ii}{2\pi}\oint_{|\xi|=1}\br{h_\lambda-z_\lambda(\xi)\Id}^{-1} z'_\lambda(\xi) \dif{\xi}\  \Phi_\lambda\,,
}
where $\Phi_\lambda$ is an analytic function of $\lambda\in\Omega_{\Lambda,\varepsilon}$, chosen below.
Choosing the contour $\Gamma_\lambda$ as a function $\mathbb S^1\ni \xi\mapsto z_\lambda(\xi)\in\CC$ which sweeps out a circle within which 
$e_\lambda$ is the only eigenvalue of $h_\lambda$ (see \cref{eq:spectral parameter for Riesz} below), we have 
\eql{
\widetilde\vf_\lambda =\left\langle\vf_\lambda,\Phi_\lambda\right\rangle\vf_\lambda\quad {\rm for}\quad \textrm{$\lambda\in\RR$ and sufficiently large}  .
\label{eq:Phi-project}
}
%
%First, $\xi\in\mathbb C \mapsto z_\lambda(\xi)$ (see \cref{eq:spectral parameter for Riesz}) is chosen to  sweep out circle within which $e_\lambda$  and which is of distance $\lambda$ from the second eigenvalue of $h_\lambda$, uniformly in large $\lambda$; 
%
It follows that for $\lambda\in\RR$ the  representation \cref{eq:Riesz projection} yields 
a non-zero multiple of the ground state $\vf_\lambda$ if  $\Phi_\lambda$ is chosen so that it is not orthogonal to $\vf_\lambda$.
To ensure that $\left\langle\vf_\lambda,\Phi_\lambda\right\rangle\ne0$,  we actually choose $\Phi_\lambda$ so that $\|\vf_\lambda- \Phi_\lambda\|\to0$ as $\lambda\to\infty$; see \cref{eq:Phi-def}, \cref{eq:Phi_lambda-to-vf_lambda} below. 
\eq{\textrm{Note: $\widetilde\vf_\lambda$ is a non-normalized ground state of $h_\lambda$ with eigenvalue $e_\lambda$.}}

\subsubsection{The function $\mathbb{S}^1\ni\xi \mapsto z_\lambda(\xi)\in\mathbb C$ and the approximate ground state $\Phi_\lambda$.}

To proceed, we use the assumption that $v$ has a unique non-degenerate minimum to invoke the results of \cite{Matsumoto_1994}.
In particular, we assume \eql{\label{eq:v-Taylor} v\left(x\right) = -1+\frac{1}{2}\left\langle x,\Hessian{v}(0)x\right\rangle +\mathcal{O}\left(\norm{x}^{3}\right)\,\quad  \textrm{as}\quad x\to0.}\\
 Here, we have assumed $\Hessian{v}(0)>0$. As a result, the following magnetic harmonic oscillator Hamiltonian plays a central role for large $\lambda$:
\eql{\label{eq:h-MHO} h_{\lambda}^{\text{MHO}}:=\left(P-\frac{1}{2}b\lambda X^{\perp}\right)^{2}+\frac{1}{2}\lambda^{2}\left\langle X,\Hessian{v}(0)X\right\rangle\,.} 
Introduce the dilation $\br{\mathfrak{U}_{\alpha}f}(x)\equiv\alpha f(\alpha x)$, which is unitary in $L^2(\mathbb R^2)$ for $\alpha\in\mathbb R$ (see also \cref{eq:dilation operators}) and note that 
\eql{\mathfrak{U}_{\sqrt{\lambda}}^{\ast}h_{\lambda}^{\text{MHO}}\mathfrak{U}_{\sqrt{\lambda}}=\lambda h_1^{\text{MHO}} =: \lambda h^{\text{MHO}}.}
Denote the eigenvalues of $h^{\text{MHO}}$, in ascending order, by $e_j^{\text{MHO}}$, $j=0,1,2,\dots$. A consequence of \cite{Matsumoto_1994} is that for $\lambda$ \emph{real} and \deleted{and }$\lambda\ge\Lambda>0$ sufficiently large, we have
\eql{\label{eq:one-well eigenvalues}e_{\lambda,j}=-\lambda^{2}+e_{j}^{\text{MHO}}\lambda+\mathcal{O}\left(\lambda^{\frac{1}{2}}\right)\, ,\quad  j=0,1.} Below, we  make finitely many further assumptions on the size of $|\lambda|$ implying that $\Lambda=\Lambda_\varepsilon$ is a sufficiently large $\varepsilon$-dependent constant.

\added{Fix a constant $c_{\mathrm{ctr}}>0$, chosen sufficiently small in terms of $\Hessian v(0)$, so that the MHO resolvent bounds in \cref{prop:MHO analyticity and bounds} apply along the contour below.}

Let $\lambda\ge\Lambda_\varepsilon$. In \cref{eq:Riesz projection} we choose the contour $\xi\mapsto z_\lambda(\xi)$ as:
 \eql{\label{eq:spectral parameter for Riesz}
z_\lambda(\xi) := -\lambda^{2}+e_{0}^{\text{MHO}}\lambda+\replacedm{\frac{1}{2}}{c_{\mathrm{ctr}}}\left(e_{1}^{\text{MHO}}-e_{0}^{\text{MHO}}\right)\xi\lambda
\qquad(\xi\in\mathbb{S}^1), } 
and note that $\lambda\mapsto z_\lambda$ is an entire analytic function. 
This contour encircles exactly one simple eigenvalue of $h_\lambda$, the ground state $e_\lambda$ and is at a distance of order $\lambda$ from the spectrum of $h_\lambda$.  Further, by the approximation of $\varphi_\lambda$ by $\varphi_{\lambda}^{\rm MHO}=\varphi_{0,\lambda}^{\rm MHO}$, the normalized \deleted{the }ground state of the quantum magnetic harmonic oscillator Hamiltonian $h_\lambda^{\rm MHO}$
(see \cref{lem:MHO ground state approximates one-well ground state}),  we 
 have that $\|\varphi_\lambda-\varphi_\lambda^{\rm MHO}\|\lesssim \lambda^{-1/2}$ for $\lambda\in\RR$ and large.  
 Finally, we choose $\Phi_\lambda$ to be a cutoff version of $\varphi_\lambda^{\rm MHO}$:
 \eql{
      \label{eq:Phi-def}
      \Phi_\lambda := \CharFun_{B_a(0)}(X)\varphi_\lambda^{\rm MHO} 
  }
  where
 \eql{
  \label{eq:exact ground state of MHO1}
\vf_\lambda^{\rm MHO}(x) := \varphi_{0,\lambda,b}^{\rm MHO}(x) = C_b \sqrt{\lambda}\exp\br{-c\lambda \deletedm{\sqrt{1+b^2/2}}\ip{x}{ S_b x}}\ .
 }
 Here, $S_b$ is some symmetric positive definite $2\times2$ matrix and the constant $C_b$ is a positive normalization constant depending on $\Hessian{v}(0)$ and $b$, varying smoothly in the latter, and is uniformly bounded away from zero for $b\ge0$. In \cref{eq:exact ground state of MHO1}, $\lambda\mapsto\sqrt\lambda$ is defined on the cut plane $\CC\setminus(-\infty,0]$. Finally, the corresponding ground state eigenvalue is:\
 \eq{
  e_{0,\lambda,b}^{\rm MHO} = \lambda\sqrt{b^2+\frac12\br{\sqrt{\mu_1}+\sqrt{\mu_2}}^2}\,,
 }
 where $0<\mu_1\le\mu_2$ are the eigenvalues of $\Hessian{v}(0)$.

We make the following observations:
\begin{enumerate}
    \item By direct inspection of \cref{eq:exact ground state of MHO1}, the mapping $\lambda\mapsto \varphi_\lambda^{\rm MHO}$ analytically extends from $\RR_+$ to $\Omega_{\Lambda,\varepsilon}$.
    
    \item Due to the Gaussian tail of $\varphi_\lambda^{\rm MHO}$, we have the approximation:
    \eql{\label{eq:Phi_lambda-to-vf_lambda}
    \|\varphi_\lambda-\Phi_\lambda \|\lesssim \lambda^{-1/2},\textrm{$\lambda\in\RR$\quad  and sufficiently large.}
    }
    
    \item Due to the spatial projection factor $\CharFun_{B_a(0)}$ in $\Phi_\lambda$, the resolvent in $\br{h_\lambda-z_\lambda(\xi)\Id}^{-1}$ in \cref{eq:Riesz projection} now acts on the range of $\CharFun_{B_a(0)}(X)$; the significance is discussed earlier in items 5 and 6 of \cref{sec:sketch} and below in \cref{rem:analytic_ext}. Introduce the spatial projection operators
    \eq{
    Qf := \CharFun_{B_a(0)}(X)f\quad {\rm and}\quad 
    Q^\perp := \Id -Q.
    }
    Below, we'll construct an analytic map
    \eq{
    \lambda \mapsto A_\lambda,\quad  \Omega_{\Lambda,\ve}\longrightarrow\calB(L^2(\RR^2)),
    }
    which extends the resolvent acting on functions supported in $B_a(0)$, i.e.,
    \eq{
    A_\lambda= \br{h_\lambda-z_\lambda(\xi)\Id}^{-1}Q\quad \textrm{for}\quad \lambda\in \Omega_{\Lambda,\varepsilon}\cap \mathbb R,
    }
    where $z_\lambda(\xi)$ is given by \cref{eq:spectral parameter for Riesz}. {\it Note, since $z_\lambda$ depends on $\xi$, the operator $A_\lambda$ also depends on $\xi$. We typically suppress this $\xi$-dependence, since all statements below hold uniformly for $\xi\in\mathbb{S}^1$.}
    \footnote{In place of $Q$, one may use $\CharFun_B(X)$ for \emph{any} compact $B\subseteq\RR^2$, although our particular proof below uses the fact $\supp(v)\subseteq B_a(0)$.}
\end{enumerate}

Our strategy going forward  is to:
\begin{enumerate}
    \item Use the representation \cref{eq:Riesz projection} and the analyticity of $\lambda\mapsto A_\lambda =\br{h_\lambda-z_\lambda\Id}^{-1}Q$ to conclude that the non-normalized ground state of $h_\lambda$,
\eq{\textrm{ $\widetilde\vf_\lambda$ has an analytic extension to $\Omega_{\Lambda,\varepsilon}$.}} 
\item Conclude that  $\lambda\mapsto\widetilde{\rho}(\lambda)$, the  {\it non-normalized hopping coefficient}, obtained from $\rho(\lambda)$ by replacing $\vf_\lambda$ with $\widetilde\vf_\lambda$:
\eql{\label{eq:non-norm-rMHO-ext} 
\widetilde\rho(\lambda,b) :=
\lambda^2\int_{x\in\RR^2} \overline{\widetilde\vf_{\overline{\lambda},b\nc}(x+d)}
        v(x+d) \exp\br{\ii b\lambda d_1 x_2}\widetilde\vf_{\lambda,b}(x-d)\dif{x}\,
}  is analytic on $\Omega_{\Lambda,\varepsilon}$. We then,  apply complex analytic \cref{lem:complex-analysis} to first prove, for all $0<b<b_\ve$ sufficiently small and all $\lambda>\Lambda_\ve$, lower bounds  for $|\widetilde\rho(\lambda,b)|$ of the type that are asserted in \cref{thm:main theorem} of $|\rho(\lambda,b)|$.
\item To transfer these lower bounds on $|\widetilde\rho(\lambda,b)|$ to the asserted lower bounds on 
 $|\rho(\lambda,b)|$, take $\lambda$ in \cref{eq:Riesz projection} to be real and apply \cref{eq:Phi-project} to obtain 
that 
\eql{\label{eq:non-norm-rMHO-ext-real} 
 \widetilde{\rho}(\lambda,b) =
|\left\langle\vf_\lambda,\Phi_\lambda\right\rangle|^2\rho(\lambda,b)\,\quad \lambda\in\RR.}
\item %
Finally,  using that   
\eql{\label{eq:C_lambda_to_1}
\abs{\left\langle\vf_\lambda,\Phi_\lambda\right\rangle|^2-1} \lesssim \lambda^{-1}\qquad \textrm{for $\lambda\in\RR$ and $\lambda>\Lambda_\ve$,}}
 we obtain the desired lower bounds for $|\rho(\lambda,b)|$ for all $\lambda>\Lambda_\ve$ and $0<b<b_\ve$.
\end{enumerate}
All details required to complete this argument are presented in the subsections that follow, along with the 
 parallel arguments to establish lower bounds for $\Delta(\lambda)$ in \cref{subsec:lower bound on Delta}.

But first we briefly sketch the proof of \cref{eq:C_lambda_to_1}. 
 Clearly, 
 \eq{ \replacedm{\norm{\Phi_\lambda}^2-1}{1-\norm{\Phi_\lambda}^2} = \norm{Q^\perp \vf_\lambda^{\rm MHO}}^2\leq \exp\br{-c \abs{\lambda} a^2},\quad \lambda>\Lambda.}
 Introduce the ground state projection and its orthogonal complement:  
 \eq{
 P_\lambda := \vf_\lambda\otimes\vf_\lambda^\ast\quad {\rm and}\quad P_\lambda^\perp\equiv\Id-P_\lambda,\quad \lambda\in\RR.
 } 
By \cref{lem:MHO ground state approximates one-well ground state} below, 
\eq{
\norm{P_\lambda^\perp \vf_\lambda^{\rm MHO}}_{L^2} \lesssim \lambda^{-1/2}, \quad \lambda>\deletedm{\RR}\addedm{\Lambda}\,.}
Hence,
\eql{\label{eq:Clambda-near1}
1-\abs{\ip{\vf_\lambda}{\Phi}}^2  = 1-\norm{\Phi}^2+\norm{P_\lambda^\perp \Phi_\lambda}^2 \leq \exp\br{-c \lambda a^2} + \norm{P_\lambda^\perp \Phi_\lambda}^2
}
And finally, 
$\norm{P_\lambda^\perp \Phi_\lambda} \leq \norm{P_\lambda^\perp \vf_\lambda^{\rm MHO}} + \norm{P_\lambda^\perp Q^\perp \vf_\lambda^{\rm MHO}}\lesssim \lambda^{-1/2}$, implying 

\eql{
\label{eq:bounds on the constant C_lambda}
\left| 1 - \left|\left\langle\vf_\lambda,\Phi_\lambda\right\rangle\right|^2\right|\lesssim \lambda^{-1},\qquad \lambda>\Lambda.
}

\begin{rem}[Analytic extension of $(h_\lambda-e_\lambda\Id)^{-1}\CharFun_{B_a(0)}$, not $(h_\lambda-e_\lambda\Id)^{-1}$] \label{rem:analytic_ext}
As mentioned in \cref{sec:sketch}, a major hurdle arises from 
$h_\lambda$ (whose potential has compact support) being a relatively compact perturbation of $H^{\mathrm{Landau}}_\lambda$, the Landau Hamiltonian \cref{eq:HLandau}, and not of $h_{\lambda}^{\text{MHO}}$. 
 $H^{\mathrm{Landau}}_\lambda$  is ill-behaved for non-real values of $\lambda$. Indeed, for the Landau gauge, it was shown in \cite{krejcirik2024abruptchangesspectralaplacian} that as soon as $\lambda$ gains a nonzero imaginary part, the spectrum of $H^{\mathrm{Landau}}_\lambda$ becomes the whole complex plane. While \cite{krejcirik2024abruptchangesspectralaplacian} does not cover our present case of the symmetric gauge (for complex $\lambda$ the two gauges are \emph{not} related by a unitary transformation and hence may not be isospectral) we show below in \cref{prop:Landau resolvent does NOT continue} that also in the symmetric gauge, the resolvents of $H^{\mathrm{Landau}}_\lambda$ and of $h_\lambda$ do \emph{not} extend to the complex plane. {\it This is in sharp contrast to the resolvent of $h_{\lambda}^{\text{MHO}}$ which does indeed analytically extend with good bounds; we study this below in \cref{prop:MHO analyticity and bounds}.}
 \end{rem}
 
 %

%This is the main reason why we only attempt to analytically continue $\br{h_\lambda-z_\lambda\Id}^{-1}Q$ and not $\br{h_\lambda-z_\lambda\Id}^{-1}$. 

\subsection{The cut-off Hamiltonian, $h_\lambda^\theta$, and the analytic extension of its resolvent }\label{subsubsec:extending h theta}

As an intermediate step, we consider the Landau Hamiltonian \emph{with a cut-off vector potential}. Let $\theta:\RR^2\to[0,1]$ be a smooth cut-off function with compact support such that $\theta=1$ on a generous neighborhood (to be specified  below) of $B_a(0)$, which contains the support of $v$; $\supp(v)\subset B_a(0)\Subset  \supp{(\theta)} $. Introduce the cut-off Landau Hamiltonian:  \eql{\label{eq:cut-off Landau def}
H^{\mathrm{Landau},\theta}_\lambda := \left(P-\frac{1}{2}b\lambda \theta(X)X^{\perp}\right)^{2}\,.
} Since $\theta$ has compact support, $H^{\mathrm{Landau},\theta}_\lambda$ is a relatively compact perturbation of the free Laplacian $P^2$, and so we expect much better behavior from its resolvent as $\lambda$ varies off the real axis into the complex plane. Further, we introduce  $h_\lambda^\theta$, the one-well Hamiltonian with cut-off vector potential:
\eql{
    h_\lambda^\theta := H^{\mathrm{Landau},\theta}_\lambda + \lambda^2v(X)\,.
}
We begin by extending the resolvent of $h_\lambda^\theta$, initially defined for $\lambda\in\RR$, to a well-defined operator, as summarized in \cref{lem:continuation of cut off one well resolvent}, for $\lambda $ varying in a complex region  $\Omega_{\Lambda,2\Lambda,\varepsilon}$; see \cref{eq:Omega-annulus}.

\begin{lem}[Analyticity and bounds for a cutoff resolvent]\label{lem:continuation of cut off one well resolvent}
   Fix $\ve\in(0,\pi/2)$ and $\Lambda$ sufficiently large (depending on $\varepsilon$).   
   The map \eq{
\Omega_{\Lambda,2\Lambda,\varepsilon}\cap\RR\ni\lambda\mapsto\br{h_\lambda^\theta-z_\lambda\Id}^{-1}\in\calB(L^2(\RR^2)) }
 extends to a map \eq{A^\theta_\lambda:\Omega_{\Lambda,2\Lambda,\varepsilon}\longrightarrow\calB(L^2(\RR^2))} which satisfies the following properties: 
\begin{enumerate}
    \item $\br{h_\lambda^\theta-z_\lambda\Id}A^\theta_\lambda = \Id$ for all $\lambda\in\Omega_{\Lambda,2\Lambda,\varepsilon}$.
    \item $\lambda\mapsto A_\lambda^\theta$ is analytic on $\Omega_{\Lambda,2\Lambda,\varepsilon}$.
    \item For some $\eta>0$, sufficiently small and independent of $\Lambda$, we have the bound \eql{\label{eq:estimate on analytic continuation of resolvent of cut off one well Hamiltonian}\norm{A_\lambda^\theta}\lesssim \Lambda^{-\eta}\,,} for all $\lambda\in\Omega_{\Lambda,2\Lambda,\varepsilon}$.
\end{enumerate}
\end{lem}

\begin{proof}[Proof of \cref{lem:continuation of cut off one well resolvent}]
Let $\varepsilon\in(0,\pi/2)$ be fixed throughout the proof.  Further, let $\lambda\in\Omega_{\Lambda,2\Lambda,\varepsilon}$, where  $\Lambda\geq\Lambda_\ve$: $\Lambda$ is arbitrarily large and $\Lambda_\ve$ is a sufficiently large $\ve$-dependent constant.
Next introduce a partition of unity $\Set{\chi_\nu}_{\nu=0}^{N_\Lambda}$, 
% \eql{\label{eq:partition of unity for h theta lambda}\Set{\chi_\nu}_{\nu=0}^{N_\Lambda}\,} 
where $N_\Lambda<\infty$ is specified below.  For any $\nu\in\Set{0,\cdots,N_\Lambda}$, 
\eq{
\chi_\nu \in C^\infty(\RR^2\to[0,1])\quad {\rm and}\quad \sum_{\nu=0}^{N_\Lambda}\chi_\nu =1.
}
 We divide the members of the partition into three categories:
\begin{enumerate}
    \item $\nu=0$: $\chi_0$ is centered at the origin and has support of diameter of order $\delta\to0$ as $\Lambda\to\infty$, for some $\delta\equiv\delta(\Lambda)$ to be chosen below in \cref{eq:the choice of delta}. To be concrete, we choose $\chi_0$ such that $\supp\chi_0\subset B_{2C\delta}(0)$ and   $\chi_0\equiv1$ on $B_{C \delta}(0)$, where $C$ is a sufficiently large order $1$ constant.
     
    \item $\nu=N_\Lambda$: $\chi_{N_\Lambda}$ is supported in the complement of a sufficiently large disc of radius $R$ of order $1$ (in $\Lambda$) chosen so that the disc entirely covers $\supp(\theta)$, i.e. \eq{ \chi_{N_\Lambda}=0\quad {\rm on}\quad  B_R(0)\supset \supp \theta.}
    \item $\nu=1,\cdots,N_\Lambda-1$: The support of $\chi_\nu$ satisfies:
    \eq{\supp\chi_\nu \subset \Set{ 2^{\nu-1}\delta<\norm{x}< 2^{\nu+1}\delta\,
}\,.}
\noindent Hence we choose 
\eql{\label{eq:N_lambda} N_\Lambda := \lceil\log_2(R/\delta)\rceil.} 
\end{enumerate}

The construction of this partition of unity is given by
    \begin{lem}\label{lem:partition of unity}
        Fix $R>0$ such that ${\rm supp}(\theta)\subset B_R(0)$, and set $N_\Lambda := \lceil\log_2(R/\delta)\rceil$ with $\delta := \Lambda^{-1/2+\eta}$ for fixed $\eta>0$ sufficiently small. Also, define $\delta_\nu := 2^{\nu}\delta$.
        
        There exists a partition of unity, i.e., a sequence $\Set{\chi_\nu}_{\nu=0}^{N_\Lambda}$ where each $\chi_\nu:\RR^2\to[0,1]$ is smooth and such that \eq{
            1 = \sum_{\nu=0}^{N_\Lambda} \chi_\nu\,,
        } with the following properties:
        \begin{enumerate}
               \item For $\nu=0$, $\chi_0$ is equal to $1$ on $B_{C \delta}(0)$ for some sufficiently large order $1$ constant $C$ and is supported on $B_{2C\delta}(0)$.
    \item $\nu=N_\Lambda$: $\chi_{N_\Lambda}$ is supported in the complement \added{of} a disc of radius $R$.
    \item $\nu=1,\cdots,N_\Lambda-1$: $\chi_\nu$ is supported in an annulus with radii between $2^{\nu-1} \delta$ to $2^{\nu+1} \delta$.
            \item For each $\nu$, there exists some $C_\alpha<\infty$ such that \eq{\norm{\partial^\alpha\chi_\nu}_\infty\leq C_\alpha\delta_\nu^{-|\alpha|}} where $\alpha\in\NN^2_{\geq0}$ is any multi-index.
            \item Further, there is a second sequence of smooth functions, $\Set{\psi_\nu}_\nu$,  such that $\psi_\nu=1$ on $\supp(\chi_\nu)$, satisfying \eq{
                \norm{\partial^\alpha \psi_\nu}_\infty \leq C_\alpha \delta_\nu^{-|\alpha|}\ .
            } Moreover,  for any $x\in\RR^2$, no more than, say, \emph{four} values of $\nu$ have $x\in\supp(\psi_\nu)$. Below, we shall introduce  open sets $U_\nu\subset \RR^2$,  such that:
\eql{\label{eq:U_nu-def}
\addedm{\supp\psi_\nu\subset }U_\nu \deletedm{\subseteq}\addedm{:=}\Set{ 2^{\nu-3}\delta<\norm{x}< 2^{\nu+3}\delta};\quad\nu=1,\cdots,N_{\Lambda}-1\,.
} and $U_0 := B_{2^3\addedm{C_0}\delta}(0)$, with \added{$C_0$ a sufficiently large order-one constant,} and $U_{\deletedm{N}\addedm{N_\Lambda}}:=\RR^2\setminus\overline{B_{2^{N_\Lambda-3}\delta}(0)}$. \added{Moreover, we choose $\psi_{N_\Lambda}$ so that $\supp\psi_{N_\Lambda}$ is contained in the set where $v$ and $\theta$ are both zero.}
        \end{enumerate}
    \end{lem}
We omit the proof.

\ifshowchanges
\noindent\deleted{In addition to }$\deletedm{\Set{\chi_\nu}_\nu}$\deleted{, we introduce a second collection of cut-off functions, }$\deletedm{\Set{\psi_\nu}_\nu}$\deleted{, such that}
\[
\deletedm{\psi_\nu=1\ \textrm{ on }\ \supp(\chi_\nu),\quad \nu=0,1,\dots,N_\Lambda}
\]
\deleted{and such that for any }$\deletedm{x\in\RR^2}$\deleted{, no more than, say, four values of }$\deletedm{\nu}$\deleted{ have }$\deletedm{x\in\supp(\psi_\nu)}$\deleted{. We shall also make use of open subsets }$\deletedm{U_\nu}$\deleted{ of }$\deletedm{\RR^2}$\deleted{ below, each defined so it contains the support of }$\deletedm{\psi_\nu}$\deleted{. We shall take}
\[
\deletedm{\supp\psi_\nu\subset U_\nu :=\Set{ 2^{\nu-3}\delta<\norm{x}< 2^{\nu+3}\delta}\qquad(\nu=1,\cdots,N_\Lambda-1)\,.}
\]
\deleted{Together with the endpoint definitions }$\deletedm{U_0:=B_{2^3C_0\delta}(0)}$\deleted{ and }$\deletedm{U_{N_\Lambda}:=\overline{B_{2^{N_\Lambda-3}\delta}(0)}^c}$\deleted{, these sets contain the supports of the corresponding }$\deletedm{\psi_\nu}$\deleted{. Moreover, }$\deletedm{\supp\psi_{N_\Lambda}}$\deleted{ is contained in the set where }$\deletedm{v,\theta}$\deleted{ are both zero.}
\fi
As a proposed approximate inverse to  $h_\lambda^\theta-z_\lambda\Id$ we define the operator:
\eql{
\label{eq:putative expression for the inverse of h_lambdatheta}
\widetilde{A}^\theta_\lambda := \widetilde{A}^\theta_{\lambda,I} + \widetilde{A}^\theta_{\lambda,II} + \widetilde{A}^\theta_{\lambda,III}\,,
}  where
\begin{subequations}
\label{eq:tildeA}
\eq{
\widetilde{A}^\theta_{\lambda,I} &:= \psi_0 \br{h_{\lambda}^{\text{MHO}}-\br{z_\lambda+\lambda^2}\Id}^{-1}\chi_0\\
\widetilde{A}^\theta_{\lambda,II} &:=
    \sum_{\nu=1}^{N_\Lambda-1}T_{\lambda,\nu}\\
\widetilde{A}^\theta_{\lambda,III} &:=   
  \psi_{N_\Lambda}\br{P^2-z_\lambda\Id}^{-1}\chi_{N_\Lambda}, 
}
\end{subequations}
and $z_\lambda=z_\lambda(\xi)$ is given \added{by} \cref{eq:spectral parameter for Riesz}. The tildes here indicate that this will turn out to be an \emph{approximate} inverse which will then have to be corrected to get the actual inverse, $A_\lambda^\theta$.

Each term in \cref{eq:putative expression for the inverse of h_lambdatheta} can be constructed to be an approximate inverse 
of $h_\lambda^\theta-z_\lambda\Id$ acting on functions supported, respectively, on (I) the inner region (support of $\chi_0)$, (II) the annular regions (the supports of $\chi_\nu$, $\nu=1,\dots,N_\Lambda-1$, and (III) the outermost region (support of $\chi_{N_\Lambda}$).

%\eql{\label{eq:putative expression for the inverse of %h_lambdatheta}
%\widetilde{A^\theta_\lambda} := \psi_0 %\br{h_{\lambda}^{\text{MHO}}-%\br{z_\lambda+\lambda^2}\Id}^{-1}\chi_0 + %\sum_{\nu=1}^{N_\Lambda-%1}T_{\lambda,\nu}+\psi_{N_\Lambda} \br{P^2-%z_\lambda\Id}^{-1}\chi_{N_\Lambda}\,,\nonumber\\
%} 

The operators $\widetilde{A}^\theta_{\lambda,I}$ and  $\widetilde{A}^\theta_{\lambda,III}$  reflect that, on the support of $\chi_0$,  we have
$h^\theta_\lambda-z_\lambda\Id\approx (h_\lambda^{\rm MHO}-\lambda^2\Id) -z_\lambda\Id$ (see \eqref{eq:v-Taylor} and \eqref{eq:h-MHO}) and that on the support of $\chi_{N_\Lambda}$, we have $h^\theta_\lambda-z_\lambda\Id= P^2-z_\lambda\Id$ (since $v$ and $\theta$ vanish there).  The operators $T_{\lambda,\nu}$ will be introduced below and are used to construct an approximate inverse to $h^\theta_\lambda-z_\lambda\Id$ on functions supported in  the remaining annular region; their construction is done via the formalism of pseudo-differential operators whose symbols are elliptic in their respective annuli.

We first show that $\Omega_{\Lambda,2\Lambda,\varepsilon}\ni\lambda \mapsto \widetilde{A}^\theta_\lambda\in\mathcal{B}(L^2(\RR^2))$ is well-defined and analytic, satisfies the bound \cref{eq:estimate on analytic continuation of resolvent of cut off one well Hamiltonian} and is an approximate inverse to $h_\lambda^\theta-z_\lambda\Id$. Finally, we'll construct an exact inverse  $A^\theta_\lambda$ with these analyticity properties.

\subsubsection{Inner region; definition, analyticity and operator bound for   $\widetilde{A}^\theta_{\lambda,I}$  } 
The choice $\widetilde{A}^\theta_{\lambda,I}$ reflects that on the support of $\chi_0$  we have
$h^\theta_\lambda-z_\lambda\Id\approx (h_\lambda^{\rm MHO}-\lambda^2) -z_\lambda\Id$; see \cref{eq:v-Taylor,eq:h-MHO}. Here, $h_\lambda^{\rm MHO}$ denotes a magnetic harmonic oscillator Hamiltonian, whose quadratic potential is not required to be isotropic. In \cref{sec:MHO} below we construct its resolvent
\eq{ r_\lambda^{\rm MHO}(\lambda^2+z_\lambda) = \br{h_\lambda^{\rm MHO} -\br{\lambda^2+z_\lambda}\Id}^{-1},} along the contour $\Gamma_\lambda$ in \cref{eq:Riesz projection}, 
prove the required analyticity on $\lambda\in\Omega_{\Lambda,\varepsilon}$, and the bound
\eq{ \left\| \widetilde{A}^\theta_{\lambda,I}\right\| \lesssim \left\|    r_\lambda^{\rm MHO}(\lambda^2+z_\lambda)\right\| \lesssim |\lambda|^{-1}\,.}

\subsubsection{Outer region: definition, analyticity and operator bound for $\widetilde{A}^\theta_{\lambda,III}$}
The choice $\widetilde{A}^\theta_{\lambda,III}$ reflects that on the support of $\chi_{N_\Lambda}$, where $v$ and $\theta$ vanish, we have
\eq{
	h^\theta_\lambda-z_\lambda\Id= P^2-z_\lambda\Id\ .
}
 Thus the outer-region parametrix is constructed using the free resolvent
$\br{P^2-z_\lambda\Id}^{-1}$. It remains to prove that this free resolvent is analytic and satisfies the required bound for $\lambda\in\Omega_{\Lambda,\ve}$.

Since $z_\lambda=-\lambda^2+\Ord{\lambda}$, writing $\lambda=r\ee^{\ii\vartheta}\in\Omega_{\Lambda,\ve}$ gives
\eq{
	z_\lambda=-r^2\ee^{2\ii\vartheta}+\delta_\lambda,\qquad
	\abs{\delta_\lambda}\leq Cr,\qquad
	\abs{\vartheta}<\frac{\pi}{2}-\ve.
}
Moreover,
\eq{
	{\rm dist}\br{-r^2\ee^{2\ii\vartheta},[0,\infty)}
	=
	r^2\inf_{E\geq0}\abs{E+\ee^{2\ii\vartheta}}
	\geq r^2\sin\br{2\ve},
}
where the final inequality follows by minimizing
\eq{
	\abs{E+\ee^{2\ii\vartheta}}^2
	=
	E^2+2E\cos\br{2\vartheta}+1,\qquad E\geq0.
}
Therefore, if $r=\abs{\lambda}>\Lambda_\ve$ sufficiently large, then
\eq{
{\rm dist}\br{z_\lambda,\sigma(P^2)} = {\rm dist}\br{z_\lambda,[0,\infty)}
	&\geq
	{\rm dist}\br{-\lambda^2,[0,\infty)}-\abs{z_\lambda+\lambda^2}
	\\
	&\geq
	r^2\sin\br{2\ve}-Cr
	\geq
	\frac{1}{2}\sin\br{2\ve}\abs{\lambda}^2.
}
In particular, $z_\lambda\notin\sigma(P^2)$ for $\lambda\in\Omega_{\Lambda,\ve}$. Hence the analyticity of
\eq{
	\Omega_{\Lambda,\ve}\ni\lambda\mapsto \br{P^2-z_\lambda\Id}^{-1}\in\calB\br{L^2(\RR^2)}
}
follows from the analyticity of the free resolvent on $\CC\setminus[0,\infty)$ and that of $\lambda\mapsto z_\lambda$. Finally, by self-adjointness of $P^2$,
\eql{\label{eq:upper bound on free resolvent}
	\norm{\br{P^2-z_\lambda\Id}^{-1}}_{\calB\br{L^2(\RR^2)}}
	=
	\frac{1}{{\rm dist}\br{z_\lambda,[0,\infty)}}
	\leq
	\frac{2}{\sin\br{2\ve}\abs{\lambda}^2}
	\leq
	\frac{2}{\sin\br{2\ve}\Lambda^2},
	\qquad \lambda\in\Omega_{\Lambda,\ve}.
}
Consequently, $\Omega_{\Lambda,\ve}\ni\lambda\mapsto\widetilde{A}^\theta_{\lambda,III}\in\calB\br{L^2(\RR^2)}$ is analytic, and
\eq{
	\norm{\widetilde{A}^\theta_{\lambda,III}}_{\calB\br{L^2(\RR^2)}}
	\lesssim
	\frac{1}{\sin\br{2\ve}\Lambda^2},
	\qquad \lambda\in\Omega_{\Lambda,\ve}.
}

\subsubsection{Intermediate region; definition, analyticity and operator bound for  $\widetilde{A}^\theta_{\lambda,II}$}

In the intermediate annular region, the operator  $h^\theta_\lambda-z_\lambda\Id$, and hence its resolvent, is not well-represented throughout by a model operator with an explicit resolvent kernel representation. Here is where the dyadic decomposition into annuli enters. The insight is that provided $v$ is smooth it is possible to decompose the big annulus into sub-annuli $U_\nu$ on which $h^\theta_\lambda-z_\lambda\Id$ behaves like $-\Delta+M_\nu^2\Id$ after a rescaling (with $M_\nu^2\sim 2^{2\nu}\Lambda^{4\eta}$ and $\eta>0$  small),  which enables us to make sense of the operators: 
\eq{
T_{\lambda,\nu}\quad  ``="\quad  \psi_\nu \br{h^\theta_\lambda-z_\lambda\Id}^{-1}\chi_\nu\, ,
} 
as bounded operators on  $L^2(\RR^2)$.
We  implement this  via the theory of pseudo-differential operators. We now summarize the result, and defer the detailed proof to \cref{sec:mesoscopic annuli}.
 %The issue is that we do not know a-priori that $h^\theta_\lambda-z_\lambda\Id$ is invertible for $\lambda\in\Omega_{\Lambda,\ve}$ (indeed, that's precisely what we're trying to prove) so we cannot write down such an expression, even if it's restricted between $\chi_\nu$ and $\psi_\nu$. To deal with this problem, we use the formalism of pseudo-differential operators (see \cref{sec:pseudo diff operators} below). 
 
 We work with the symbol $U_\nu\times\RR^2\ni(x,p)\mapsto s_\lambda^\theta(x,p)\in\CC$ associated with $h_\lambda^\theta-z_\lambda\Id$.
 %but consider it as a symbol where the position variable ranges over $U_\nu$ ( \cref{eq:U_nu-def} ). 
 Ellipticity of $s_\lambda^\theta$ and arguments based on \cref{lem:local elliptic lemma} below (applied to a rescaling of $s_\lambda^\theta$), imply the existence of a bounded operator $T_{\lambda,\nu}$ on $L^2(\RR^2)$, which is analytic for $\lambda$ varying in the truncated wedge $\Omega_{\Lambda,2\Lambda,\varepsilon}$ and satisfies:
\eql{\label{eq:estimate on approximate resolvent in mesoscopic annuli}
    \norm{T_{\lambda,\nu}}_{\calB(L^2(U_\nu))} \lesssim \frac{1}{\Lambda^2 2^{2\nu} \delta^2.}
} Further, $T_{\lambda,\nu}$ is an approximate inverse in the $\nu$-th annuli, in the sense that:
\eql{\label{eq:estimate on error term in approximate resolvent in mesoscopic annuli}
    \norm{\br{h_\lambda^\theta-z_\lambda\Id}T_{\lambda,\nu}-\chi_\nu}_{\calB(L^2(U_\nu))} \lesssim \frac{1}{\Lambda 2^{2\nu} \delta^2}\,.
} This requires us to  choose $\delta=\delta(\Lambda)$ so that
$\lim_{\Lambda\to\infty}\Lambda\delta^2 = \infty$. The latter is satisfied  by taking the radius of the inner disc to be:
\eql{\label{eq:the choice of delta}\delta(\Lambda) = \Lambda^{\eta-1/2},\quad \textrm{for some small $\eta>0$.}
}  Hence, $N_\Lambda\to\infty$, the number of annular subdivisions of the set $\delta(\Lambda)\le \|x\|\le R$  given by \cref{eq:N_lambda}, tends to infinity as $\Lambda\to\infty$.

%\footnote{\textcolor{red}{this choice \cref{eq:the choice of delta} is ALSO used to handle the inner MHO region. Maybe this is something to comment on. There are constraints on parameters coming from regions. On the one hand, $\Lambda\delta^2\to\infty$. On the other hand, one needs to satisfy  $\dist(S,T)>\Lambda^{-1/2}$? So $\delta^2= \log\Lambda/\Lambda$ is not ok?}}

%
\subsubsection{Proof that $\widetilde{A}^\theta_\lambda: L^2(\RR^2)\to H^2(\RR^2)$ is an approximate inverse to $h_\lambda^\theta-z_\lambda\Id$}
We rewrite  $\widetilde{A}^\theta_\lambda$, given by \cref{eq:tildeA} as 
\eq{
\widetilde{A}^\theta_\lambda = \psi_0 r_{\lambda,0}(z_\lambda)\chi_0 + \sum_{\nu=1}^{N_\Lambda-1}T_{\lambda,\nu}+\psi_{N_\Lambda}r_{\lambda,N_\Lambda}(z_\lambda)\chi_{N_\Lambda},
}
where 
 \eql{ r_{\lambda,0}(z_\lambda):=r_\lambda^{\rm MHO}(\lambda^2+z_\lambda),\qquad r_{\lambda,N_\Lambda}(z_\lambda):=\br{P^2-z_\lambda\Id}^{-1}.
\label{eq:r_lamnu-def}}
Here, 
$r^{\rm MHO}_\lambda(\lambda^2+z_\lambda)$, denotes the resolvent of $h_{\lambda}^{\text{MHO}}-\lambda^2\Id$ at spectral parameter $z=z_\lambda$:
\eql{ r_\lambda^{\rm MHO}(\lambda^2+z_\lambda) := \br{h_{\lambda}^{\text{MHO}}-\br{\lambda^2+z_\lambda}\Id}^{-1} . \label{eq:rMHO}}
Let us calculate $\br{h_\lambda^\theta-z_\lambda\Id}\widetilde{A^\theta_\lambda}$: 
\eq{
    \br{h_\lambda^\theta-z_\lambda\Id}\widetilde{A^\theta_\lambda} &= \sum_{\nu=0,N_\Lambda} \br{h_\lambda^\theta-z_\lambda\Id} \psi_\nu r_{\lambda,\nu}(z_\lambda)  \chi_\nu+\sum_{\nu=1}^{N_\Lambda-1} \br{h_\lambda^\theta-z_\lambda\Id} T_{\lambda,\nu}\\
   % &= \sum_{\nu=0,N_\Lambda} \left[h_\lambda^\theta ,\psi_\nu\right] r_{\lambda,\nu}(z_\lambda)  \chi_\nu + \psi_\nu  \br{h_\lambda^\theta-z_\lambda\Id}r_{\lambda,\nu}(z_\lambda)  \chi_\nu + \sum_{\nu=1,\cdots,N_\Lambda-1}\chi_\nu + \br{\br{h_\lambda^\theta-z_\lambda\Id} T_{\lambda,\nu}-\chi_\nu} \\
    &=: \Id+K_\lambda^\theta,
}
where $K_\lambda^\theta$ is given explicitly by
\begin{align}\nonumber
K_\lambda^\theta &:= \left[h_\lambda^\theta ,\psi_0\right] r_{\lambda}^{\mathrm{MHO}}(\lambda^2+z_\lambda)  \chi_0 + \psi_0  \br{h_\lambda^\theta+\lambda^2\Id-h_\lambda^{\mathrm{MHO}}}r_{\lambda}^{\mathrm{MHO}}(\lambda^2+z_\lambda)  \chi_0\\
&+\sum_{\nu=1}^{N_\Lambda-1} \br{\br{h_\lambda^\theta-z_\lambda\Id} T_{\lambda,\nu}-\chi_\nu}+\left[h_\lambda^\theta ,\psi_{N_\Lambda}\right] \br{P^2-z_\lambda\Id}^{-1}  \chi_{N_\Lambda} .\nonumber\\
&{\ }\label{eq:error term for theta-resolvent}
\end{align}

%We shall prove the following:
%\begin{enumerate}
%    \item For all $\nu=0,\dots,N_\Lambda$, $\lambda\mapsto \psi_\nu r_{\lambda,\nu}(z_\lambda)\chi_\nu$   analytic from $\Omega_{\Lambda,2\Lambda,\ve}$ to $\mathcal{B}(L^2(\RR^2))$ and satisfies the bound $\|\psi_\nu r_{\lambda,\nu}(z_\lambda)\chi_\nu\|\lesssim \Lambda^{-\eta}$. and hence
 %   \[ \|\widetilde{A^\theta_\lambda}\|\lesssim \Lambda^{-\eta}.\]
%\end{enumerate}

In writing this expression, we have used the properties 
\begin{enumerate}
    \item $\psi_\nu\chi_\nu=\chi_\nu$, 
    \item $\psi_{N_\Lambda}\br{h^\theta_\lambda-P^2}=0$, and 
    \item $\sum_\nu\chi_\nu=\Id$.
\end{enumerate}

We claim that $\Omega_{\Lambda,2\Lambda,\ve}\ni\lambda\mapsto K_\lambda^\theta\in\mathcal{B}(L^2(\RR^2))$  is analytic and moreover, with the choice \cref{eq:the choice of delta}, satisfies the bound:
\eql{\label{eq:K_lam_theta-bound}
\norm{K_\lambda^\theta}_{\calB(L^2(\RR^2))} \lesssim  %\log\br{\Lambda}
\Lambda^{-\eta}.}
To prove the bound \cref{eq:K_lam_theta-bound},  we'll estimate the norm of each term in \cref{eq:error term for theta-resolvent} in $\mathcal{B}(L^2(\RR^2))$ and sum over $0\le\nu\le N_\Lambda$.
%and we note that the number of terms is  $N_\Lambda+1\sim \log(\Lambda)$.
Provided \cref{eq:K_lam_theta-bound} holds, for $\Lambda$ sufficiently large, we can use a Neumann series to construct $(\Id+K_\lambda^\theta)^{-1}$, which is  also analytic as $\lambda$ ranges in $\Omega_{\Lambda,2\Lambda,\ve}$. The desired operator in the statement of \cref{lem:continuation of cut off one well resolvent} is thus given by: 
\eql{\label{eq:A_lambdatheta}
\boxed{A_\lambda^\theta := \widetilde{A^\theta_\lambda} \br{\Id+K_\lambda^\theta}^{-1}}\,.
}

The proof of  \cref{lem:continuation of cut off one well resolvent} is completed  by establishing  the assertions concerning $K_\lambda^\theta$, in particular the bound \eqref{eq:K_lam_theta-bound}, which depends on 
the existence, analyticity and bounds for the operators $T_{\lambda,\nu}$; see  \cref{eq:estimate on approximate resolvent in mesoscopic annuli} and \cref{eq:estimate on error term in approximate resolvent in mesoscopic annuli}. 
In the remainder of this subsection we prove the bound \eqref{eq:K_lam_theta-bound} on $K_\lambda^\theta$ using \cref{eq:estimate on approximate resolvent in mesoscopic annuli} and \cref{eq:estimate on error term in approximate resolvent in mesoscopic annuli}.
 We then complete the proof of  \cref{lem:continuation of cut off one well resolvent} in \cref{sec:mesoscopic annuli} by constructing the operators $T_{\lambda,\nu}$ and \replaced{derive}{deriving} their bounds.

\subsubsection{Bounding $K_\lambda^\theta$; \eqref{eq:K_lam_theta-bound} }

We proceed by bounding each term in the sum \cref{eq:error term for theta-resolvent}.

\paragraph{The MHO ($\nu=0$) contributions to \cref{eq:error term for theta-resolvent}.} We deal with the first two terms in \cref{eq:error term for theta-resolvent} separately using \cref{prop:MHO analyticity and bounds}.  

\noindent{\it Bound on the first term in \cref{eq:error term for theta-resolvent}:}
 Consider the commutator term:
\eq{
\left[h_\lambda^\theta ,\psi_0\right] r_{\lambda}^{\mathrm{MHO}}(\lambda^2+z_\lambda)  \chi_0 = \left[H_\lambda^{\mathrm{Landau},\theta} ,\psi_0\right] r_{\lambda}^{\mathrm{MHO}}(\lambda^2+z_\lambda)  \chi_0\,.
} 
where  $H_\lambda^{\mathrm{Landau},\theta}$, given in \cref{eq:cut-off Landau def}, can be expanded as:
\eq{
   H_\lambda^{\mathrm{Landau},\theta} = P^2-\frac{b\lambda}{2}\left( \theta X^\perp\cdot P + P\cdot\theta X^\perp\right) + \frac{b^2\lambda^2}{4}\br{X_1^2+X_2^2}\theta^2\ .
}
Therefore, 
\eql{
    \left[H_\lambda^{\mathrm{Landau},\theta} ,\psi_0\right] &=
     [P^2,\psi_0] - \frac{b\lambda}{2} [\theta X^\perp\cdot P,\psi_0] - \frac{b\lambda}{2}[ P\cdot\theta X^\perp,\psi_0]\notag \\
     &= -2i\nabla\psi_0\cdot P -\Delta\psi_0  - b\lambda\theta X^\perp\cdot (P\psi_0).\label{eq:HLandau-commutator}
}
Hence, 
\eql{
 \left[h_\lambda^\theta ,\psi_0\right]
 r_{\lambda}^{\mathrm{MHO}}(\lambda^2+z_\lambda)\chi_0
 &= \left[(\Delta\psi_0)(X) +  (P\psi_0)(X)\cdot P\right.\notag\\
 &\qquad\left. + b\lambda X^\perp\cdot(P\psi_0)(X)\right]
 r_{\lambda}^{\mathrm{MHO}}(\lambda^2+z_\lambda) \chi_0(X).
\label{eq:model-r_MHO1}}
Since $\psi_0$ is equal to one on a somewhat enlarged neighborhood of $\supp(\chi_0)$, we have
\begin{equation}\label{eq:sep-support}
\dist(\supp(\chi_0),\supp(\partial_i\psi_0)) \gtrsim c\delta = c \Lambda^{-1/2+\eta},
\end{equation}
where the latter equality holds by the choice in \cref{eq:the choice of delta}.
Therefore  
%\eq{
%\left[H_\lambda^{\mathrm{Landau},\theta} ,\psi_0\right] = \left[P^2,\psi_0\right] - b \lambda \theta(X) X \wedge \left[P, \psi_0\right] = -\br{\Delta \psi_0} -2 \ii \br{\nabla \psi_0} \cdot P + \ii b \lambda \theta(X) X \wedge \br{\nabla \psi_0}\,.
%}
the spatial localization on the left in \cref{eq:model-r_MHO1} is on a compact subset of $\RR^2$ which is bounded away from the support of $ \chi_0$, which localizes on the right side of \cref{eq:model-r_MHO1}.
This is the setting of \cref{prop:MHO analyticity and bounds}.

By part 1 of \cref{prop:MHO analyticity and bounds}, $\lambda\mapsto\left[h_\lambda^\theta ,\psi_0\right] r_{\lambda}^{\mathrm{MHO}}(\lambda^2+z_\lambda)  \chi_0$ is  analytic as a $\mathcal{B}(L^2(\RR^2))$-valued  function on $\Omega_{\Lambda,\ve}$.
 We next apply the estimate in part 2  of \cref{prop:MHO analyticity and bounds}
for $\lambda\in\Omega_{\Lambda,\ve}$,  $S,T\subseteq\RR^2$ $j=1,2,\alpha=0,1$:
\eql{\label{eq:local-decay}
\norm{\chi_S(X) P_j^\alpha r_{\lambda}^{\mathrm{MHO}}(\lambda^2+z_\lambda)\chi_T(X)}_{\calB(L^2(\RR^2))} \lesssim \exp\br{-c \Lambda\dist(S,T)^2}
} \\
which holds provided
%
%\footnote{{\textcolor{red}{isn't it sufficient for $dist(S,T)>0$?}}}
%
$ \dist(S,T)>\Lambda^{-1/2}$.  
We conclude from \cref{eq:local-decay} and \cref{eq:sep-support} that \eq{ \norm{\left[h_\lambda^\theta ,\psi_0\right] r_{\lambda}^{\mathrm{MHO}}(\lambda^2+z_\lambda)  \chi_0}\lesssim\exp\br{-c \Lambda^{\replacedm{\eta}{2\eta}}}=o(1)\quad \textrm{as $\Lambda\to\infty$}.}

\noindent{\it Bound on the second  term in \cref{eq:error term for theta-resolvent}:}
\eq{
\psi_0  \br{h_\lambda^\theta+\lambda^2\Id-h_\lambda^{\mathrm{MHO}}}r_{\lambda}^{\mathrm{MHO}}(\lambda^2+z_\lambda)  \chi_0
} we have, using the fact that $\theta=1$ on $\supp(\psi_0)$,
\eq{
\psi_0\br{h_\lambda^\theta+\lambda^2\Id-h_\lambda^{\mathrm{MHO}}} &= 
\lambda^2\psi_0(X)\left[ v(X) - \left(-\Id+\frac12\left\langle X,\Hessian{v}(0) X\right\rangle\right) \right]\\
&= \lambda^2 \psi_0(X) q_3(X),
}
where $|q_3(x)|\lesssim \|x\|^3$, and the implicit constant depends on $\max_{1\le i,j,k\le2}|\partial_i\partial_j\partial_k v\|_\infty$; this is the origin of the regularity assumption on $v$ in \cref{thm:main theorem}. Further,  in \cref{prop:MHO analyticity and bounds} below, we prove the bound:
\eql{\label{eq:rMHO-bound}
    \norm{r_{\lambda}^{\mathrm{MHO}}(\lambda^2+z_\lambda) }_{\calB(L^2(\RR^2))} \lesssim \Lambda^{-1}.
} Since the support of $\psi_0$ has radius $C'\delta$ for some constant $C'$, we have for $\lambda\in \Omega_{\Lambda,2\Lambda,\varepsilon}$
\begin{subequations}
\label{eq:harmonic_oscillator_approx_estimate}
\eql{
&\norm{\psi_0  \br{h_\lambda^\theta+\lambda^2\Id-h_\lambda^{\mathrm{MHO}}}r_{\lambda}^{\mathrm{MHO}}(\lambda^2+z_\lambda)  \chi_0}_{\calB(L^2(\RR^2))}\\
&\qquad \lesssim \left(\max_{\lambda\in \Omega_{\Lambda,2\Lambda,\varepsilon}}|\lambda|^2\right)  \Lambda^{-1}\times  (C'\delta)^3\\
&\lesssim (2\Lambda)^2\times \Lambda^{-1}\times  (C'\delta)^3\lesssim \Lambda \delta^3.  \label{eq:determine_delta} 
}
\end{subequations}
We see that if $\delta$ is chosen to be $o(\Lambda^{-1/3})$ we can make this term arbitrarily small by taking $\Lambda$ large; this is compatible with the choice $\delta(\Lambda)=\Lambda^{-1/2+\eta}$ in \cref{eq:the choice of delta}.

\begin{rem}[On the harmonic oscillator approximation]
For  $\lambda\in\Omega_{\Lambda,2\Lambda,\ve}$, the error bound \cref{eq:harmonic_oscillator_approx_estimate}, and the choice of $\delta(\Lambda)$,  specify the size of the neighborhood of  $x=0$ for which a ``good  local resolvent approximation'' of $h_\lambda-z_\lambda\Id$ is given by  
the magnetic harmonic oscillator resolvent $r_{\rm MHO}(z_\lambda)$. In \cref{eq:determine_delta}, the factor $\Lambda^{-1}$ coming from \cref{eq:rMHO-bound}
is set by the lower bound $|\lambda|\ge\Lambda $, while the factor $(2\Lambda)^2$ arises from the upper bound, $|\lambda|\le2\Lambda$. So clearly an upper bound on $|\lambda|$ is needed to approximate  the resolvent of 
$h_\lambda-z_\lambda\Id$ in any fixed open neighborhood of $x=0$. To access all $\lambda\ge \Lambda$ we shall, in 
\cref{sec:annulus-to-all-Omega-Lambda}, apply the construction of this section to a countable family of overlapping regions which cover all of $\Omega_{\Lambda,\varepsilon}$, and then extend  from $\Omega_{\Lambda,2\Lambda,\varepsilon}$ to all $\Omega_{\Lambda,\varepsilon}$  by an analytic continuation argument.
\end{rem}

\paragraph{Contributions from $T_{\lambda,\nu}$ terms.} Assuming the bound \cref{eq:estimate on error term in approximate resolvent in mesoscopic annuli}, which we prove  in \cref{sec:mesoscopic annuli}, we have that the sum of terms, corresponding  to $1\le\nu\le N_\Lambda-1$, is  bounded, using that $\replacedm{\delta^2\Lambda}{(\Lambda\delta^2)^{-1}}=\Lambda^{-2\eta}$, as follows
\eq{
\left\|\sum_{\nu=1}^{N_\Lambda-1} \br{\br{h_\lambda^\theta-z_\lambda\Id} T_{\lambda,\nu}-\chi_\nu}\right\| \le 
\Lambda^{-2\eta} \sum_{\nu=1}^{N_\Lambda-1}  2^{-2\nu}\lesssim \Lambda^{-2\eta}\,.
}

\paragraph{The free Laplacian contributions ($\nu=N_\Lambda$) to \cref{eq:error term for theta-resolvent}.}  The support of $\chi_{N_\Lambda}$ is disjoint from $\supp(\theta)$. Since $\supp(\theta)\supseteq\supp(v)$, we have $h_\lambda^\theta =P^2$ on the support of $\chi_{N_\Lambda}$ and so
\eql{
 \left[h_\lambda^\theta ,\psi_{N_\Lambda}\right] \br{P^2-z_\lambda\Id}^{-1}  \chi_{N_\Lambda}=\left[(P^2\psi_{N_\Lambda})(X) +  (P\psi_{N_\Lambda})(X)\cdot P \right]
\br{P^2-z_\lambda\Id}^{-1} \chi_{N_\Lambda}(X).
\label{eq:model-r_MHO}}
Note that  $S:=\supp(P \psi_{N_\Lambda})$ and 
  $T:=\supp(\chi_{N_\Lambda})$ satisfy  ${\rm dist}(S,T)\ge c>0$, since $\psi_{N_\Lambda}\equiv1$ on an order one neighborhood of $\supp(\chi_{N_\Lambda})$.

 For $\lambda\in\Omega_{\Lambda,2\Lambda,\varepsilon}$, with $\Lambda$ large,  we have  $z_\lambda= -\lambda^2+\calO(\lambda) = -|\lambda|^2 \exp(2\ii \arg(\lambda))+\calO(\lambda)$. 
Therefore,  $|E-z_\lambda|=|E+|\lambda|^2 \exp(2\ii \arg(\lambda))|+\calO(\lambda)\ge \Lambda^2|\sin(2\varepsilon)|$, since $|\lambda|>\Lambda$ and $\arg(\lambda)< \pi/2-\varepsilon$.\\
Therefore, 
\eql{\label{eq:upper bound on free resolvent repeated}
\norm{\br{P^2-z_\lambda\Id}^{-1}}_{\mathcal{B}(L^2(\RR^2))} &= \frac{1}{\inf_{E\in[0,\infty)}|E-z_\lambda|}\lesssim \frac{1}{\sin(2\ve)\Lambda^2}\,.
} 
From \cref{eq:upper bound on free resolvent} a Combes-Thomas type bound for $P^2=-\Delta$,   $j=1,2$ and $\alpha=0,1$ \cite[Eq.~(7.42) and the discussion after Eqs.~(7.44)–(7.45)]{Teschl2014},
\cite[\href{https://dlmf.nist.gov/10.29.E4}{(10.29.4)},
      \href{https://dlmf.nist.gov/10.40.E2}{(10.40.2)},
      \href{https://dlmf.nist.gov/10.40.E8}{(10.40.8)}]{DLMF}:
\eql{\label{eq:diagonal-decay}
\norm{\chi_S P_j^\alpha \br{P^2-z_\lambda\Id}^{-1} \chi_T}_{\mathcal{B}(L^2(\RR^2))}  \lesssim \exp\br{-c\sqrt{\sin(\ve)\Lambda^2}\dist(S,T)}\,.
}
Therefore, the $\nu=N_\Lambda$-contribution to \cref{eq:error term for theta-resolvent},
$\left[h_\lambda^\theta ,\psi_{N_\Lambda}\right] \br{P^2-z_\lambda\Id}^{-1}  \chi_{N_\Lambda}$, satisfies the same   exponentially  decaying bound, \cref{eq:diagonal-decay}, as  $\Lambda\to\infty$. 

The proof of \cref{lem:continuation of cut off one well resolvent} will be completed in \cref{sec:mesoscopic annuli} below, where we construct the operators $T_{\lambda,\nu}$, $\nu=1,\cdots,N_\Lambda-1$, and prove that they \replaced{satisfies}{satisfy} the bounds  \cref{eq:estimate on approximate resolvent in mesoscopic annuli,eq:estimate on error term in approximate resolvent in mesoscopic annuli}.
 \end{proof}

\subsection{Removing the cut-off to transfer results  from $h_\lambda^\theta$ to $h_\lambda$}\label{subsubsec:removing the cut off}

As explained in \cref{sec:sketch}, in the sketch of the proof, we cannot hope to analytically continue $\br{h_\lambda-z_\lambda\Id}^{-1}$ as an operator on $L^2(\RR^2)$. Fortunately, our strategy only requires an analytic extension of the operator $\br{h_\lambda-z_\lambda\Id}^{-1}Q$, where $Q=\CharFun_{B_a(0)}(X)$ is the projection onto functions supported in $B_a(0)\supseteq\supp(v)$. 
\begin{lem}[Analyticity and bounds for the resolvent, with cut-off $\theta$ removed]\label{lem:continuation of one-well resolvent}
    The one-well resolvent, acting on functions supported in $B_a(0)$,
    \eq{
    \RR_{\geq \Lambda}\ni\lambda  \mapsto \br{h_\lambda-z_\lambda\Id}^{-1}Q \in \calB\br{QL^2(\RR^2)\to L^2(\RR^2)}
    } extends to an operator 
    \eq{
      \Omega_{\Lambda,\ve} \ni\lambda \mapsto A_\lambda \in \calB\br{QL^2(\RR^2)\to L^2(\RR^2)}}
 with the following properties:
    \begin{enumerate}
    \item $\br{h_\lambda-z_\lambda\Id}A_\lambda = Q$ for all $\lambda\in\Omega_{\Lambda,\varepsilon}$.
    \item $\lambda\in \Omega_{\Lambda,\varepsilon} \mapsto A_\lambda\in\calB(QL^2(\RR^2)\to L^2(\RR^2))$ is analytic. \\
    \item For all $\lambda\in\Omega_{\Lambda,\varepsilon}$,  \eql{\label{eq:bound on analytic continuation of resolvent of single well Hamiltonian}\norm{A_\lambda}_{\calB(QL^2(\RR^2)\to L^2(\RR^2))}\lesssim \replacedm{\abs{\lambda}^{-\eta}}{\abs{\lambda}^{\sharp}}\,.} 
    \item $A_\lambda$ and its derivatives \replaced{satisfies}{satisfy} an off-diagonal exponential decay bound, in the sense that, for all $j=1,2$ and $\alpha=0,1$, \eql{\label{eq:off-diagonal exp decay of resolvent}
        \norm{\chi_S P_j^\alpha A_\lambda }_{\calB(QL^2(\RR^2)\to L^2(\RR^2))} \lesssim \exp\br{-c \abs{\lambda} \dist(S,B_{2a}(0))}\nonumber\\\qquad(S\subseteq\RR^2)\,.
    }
\end{enumerate}
\end{lem}

\begin{proof}[Proof of \cref{lem:continuation of one-well resolvent}]
    
Recall that $Q$ denotes the projection $Q\equiv \CharFun_{B_a(0)}(X)$ and $Q^\perp \equiv\CharFun_{B_a(0)^c}(X)$, with $Q+Q^\perp=\Id$.
Choose a \emph{new} smooth cut-off function $\chi\in C^\infty(\RR^2\to[0,1])$ (distinct from the $\chi_\nu$'s of the previous section) such that:
\eq{
\supp(\chi)\subset B_{2a}(0)\quad \textrm{and}\quad  \chi\equiv 1
\quad \textrm{ on\quad   $B_a(0)$}, 
}
and hence  $\chi Q=Q$.
Further, introduce a new cutoff, $\theta$, such that 
\eq{ \theta\equiv 1 \quad \textrm{on}\quad B_{2a}(0),\quad \supp(\theta)\subset B_R(0),}
 where $R>2a$.
%We arrange for $\theta$ (previously taken to have $\supp(\theta)\subset B_R(0)$ and equal to one on a neighborhood of $B_a(0)$) to be such that  
%  where $2a<R$.

We seek an operator $A_\lambda$ satisfying $(h_\lambda-z_\lambda)A_\lambda Q=Q$ and $QA_\lambda Q\in {\mathcal B}(L^2(\RR^2))$. Note that $\chi(h_\lambda-h_\lambda^\theta)=0$, since $h_\lambda-h_\lambda^\theta=0$ on $B_{2a}(0)$. We obtain
\begin{align}
    \br{h_\lambda-z_\lambda\Id}\chi A_\lambda^\theta Q  &= \left[\br{h_\lambda-z_\lambda\Id},\chi\right]A_\lambda^\theta Q+\chi \br{h_\lambda-h_\lambda^\theta}A_\lambda^\theta Q+Q\nonumber \\
    &= 
\left[h_\lambda,\chi\right]A_\lambda^\theta Q+Q \nonumber\\
&= \left(\left[h_\lambda,\chi\right]A_\lambda^\theta Q + \Id\right)Q.\label{eq:failed-right-inverse-0}
\end{align}
    % &= Q^\perp \left[h_\lambda,\chi\right]A_\lambda^\theta Q+Q\left[h_\lambda,\chi\right]A_\lambda^\theta Q+Q\\
    % &= \br{h_\lambda-z_\lambda\Id}R^{\mathrm{Landau}}_\lambda(z_\lambda) Q^\perp \left[h_\lambda,\chi\right]A_\lambda^\theta Q-\lambda^2v(X)R^{\mathrm{Landau}}_\lambda(z_\lambda)Q^\perp \left[h_\lambda,\chi\right]A_\lambda^\theta Q\\
    % &+Q\left[h_\lambda,\chi\right]A_\lambda^\theta Q+Q\,.
 
 Using $Q+Q^\perp=\Id$, we have 
\begin{equation}
\br{h_\lambda-z_\lambda\Id}\chi A_\lambda^\theta Q = Q\left( [h_\lambda,\chi]A_\lambda^\theta +\Id\right) + Q^\perp\left[h_\lambda,\chi\right]A_\lambda^\theta Q, 
\label{eq:failed-right-inverse}\end{equation}
where the latter term in \eqref{eq:failed-right-inverse}, $Q^\perp\left[h_\lambda,\chi\right]A_\lambda^\theta Q$, is nonzero; 
there is ``leakage of $\left[h_\lambda,\chi\right]$  outside of $Q$''.
Were this latter term absent, smallness of $\left[h_\lambda,\chi\right]A_\lambda^\theta$ (using \cref{eq:estimate on analytic continuation of resolvent of cut off one well Hamiltonian}) would then imply a right inverse for $h_\lambda-z_\lambda\Id$.

% Now we would like to argue that the term $\left[h_\lambda,\chi\right]A_\lambda^\theta Q$ is small, and  so we can find a right inverse on $QL^2(\RR^2)$. However, an issue arises since the commutator $\left[h_\lambda,\chi\right]$ "leaks" outside of $Q$; that is, $[h_\lambda,\chi]A_\lambda^\theta $
% \eq{
% \left[h_\lambda,\chi\right]A_\lambda^\theta Q = Q\left[h_\lambda,\chi\right]A_\lambda^\theta Q+Q^\perp \left[h_\lambda,\chi\right]A_\lambda^\theta Q\quad \textrm{with}\quad Q^\perp \left[h_\lambda,\chi\right]A_\lambda^\theta Q\neq 0.
% } 

To deal with this issue, note that
\eq{ h_\lambda- z_\lambda\Id = \br{P-\frac{\lambda b}{2}X^\perp}^2 - z_\lambda\Id +\lambda^2 v(X) = H_\lambda^{\rm Landau} - z_\lambda\Id +\lambda^2 v(X)
}
and hence, after applying $R^{\rm Landau}_\lambda(z_\lambda)$ on the right and rearranging terms we obtain:
\eq{ (h_\lambda- z_\lambda\Id)R^{\rm Landau}_\lambda(z_\lambda)-
\lambda^2 v(X)R^{\rm Landau}_\lambda(z_\lambda)= \Id  }

Applying this identity to \eqref{eq:failed-right-inverse-0}, we obtain
\eq{
    \left[h_\lambda,\chi\right]A_\lambda^\theta Q &=  \br{h_\lambda-z_\lambda\Id}R^{\mathrm{Landau}}_\lambda(z_\lambda)  \left[h_\lambda,\chi\right]A_\lambda^\theta Q-\lambda^2v(X)R^{\mathrm{Landau}}_\lambda(z_\lambda) \left[h_\lambda,\chi\right]A_\lambda^\theta Q \,.
}
The previous equation is valid for all $\lambda\in\Omega_{\Lambda,2\Lambda,\varepsilon}$, thanks to \cref{prop:analytic extension of Landau resolvent on compact sets} \replaced{down}{below}, which provides an extension of $R^{\mathrm{Landau}}_\lambda(z_\lambda)$, when acting on functions with (fixed) compact support. Indeed,   $\left[h_\lambda,\chi\right]$ is an operator of the form $f(X)+g(X)P$ where both $f,g$ have compact support within $\supp(\partial_j \chi)$.

Re-arranging and using $Q v(X) = v(X)$ and $Q^2=Q$, we find
\eq{
\br{h_\lambda-z_\lambda\Id}  A^\flat_\lambda &= Q + QK^\flat_\lambda = 
 Q(\Id + K^\flat_\lambda)
} where $A^\flat_\lambda: QL^2(\RR^2)\to L^2(\RR^2)$ is given by:
\eq{
A^\flat_\lambda := \br{\chi-R^{\mathrm{Landau}}_\lambda(z_\lambda)  \left[h_\lambda,\chi\right]} A_\lambda^\theta Q
} and $K^\flat_\lambda:QL^2\to QL^2$ is given by
\eq{
K^\flat_\lambda := -\lambda^2 Q v(X)R^{\mathrm{Landau}}_\lambda(z_\lambda)\left[h_\lambda,\chi\right]A_\lambda^\theta Q\,.
} 
If $\norm{K^\flat_\lambda}<1$, then $\left(\Id+K_\lambda\right)^{-1}$ exists and therefore
\eq{
    A_\lambda := A^\flat_\lambda
\br{\Id+K^\flat_\lambda}^{-1}:QL^2\to L^2}
satisfies $(h_\lambda-z_\lambda)A_\lambda=Q$, i.e. $A_\lambda$ behaves as an inverse of $h_\lambda-z_\lambda\Id$ when acting on functions supported in $B_a(0)$. 
 
We now estimate $\norm{K^\flat_\lambda}$. To do so, we factorize
\eql{\label{eq:Kflat-factored}
K^\flat_\lambda = \br{-\lambda^2 Q v(X)R^{\mathrm{Landau}}_\lambda(z_\lambda) \chi_T  } \times \br{\left[h_\lambda,\chi\right]A_\lambda^\theta Q}
} with $T=\supp(\partial\chi)$. 

By \cref{prop:analytic extension of Landau resolvent on compact sets}, the first factor in \cref{eq:Kflat-factored} satisfies the bound 
\eq{
\norm{\lambda^2 Q v(X)R^{\mathrm{Landau}}_\lambda(z_\lambda) \chi_T }_{\mathcal{B}(L^2(\RR^2))} \leq C\abs{\lambda}^\sharp\exp\br{\replacedm{-C}{C} b \abs{\lambda}}\,.
}

By  the off-diagonal decay used in the proof of \cref{lem:continuation of cut off one well resolvent}, in particular, \cref{eq:local-decay,eq:diagonal-decay},
the second factor in \eqref{eq:Kflat-factored} satisfies
\eq{
\norm{\left[h_\lambda,\chi\right]A_\lambda^\theta Q }_{\mathcal{B}(L^2(\RR^2))} \leq C\exp\br{-C  \abs{\lambda}}\,.
} Since $K^\flat_\lambda$ maps $QL^2$ to itself, we conclude, by taking $b$ sufficiently small, that
\eq{
\norm{K^\flat_\lambda}_{\mathcal{B}(Q L^2)} \leq C' \exp\br{-c \abs{\lambda}}\,.
}

% Note that $\left[h_\lambda,\chi\right]=f(X)+g(X)\cdot P$ where $\supp(f,g)\subset \supp(\partial \chi)$. Indeed, we have $g=-2\ii\nabla\chi$ and $f=x\mapsto -\br{\Delta\chi}(x)+\ii b\lambda x^\perp\cdot\nabla\chi(x)$. Hence $\norm{K^\flat_\lambda}$ satisfies a bound of the following type:
% \eq{
% \norm{K^\flat_\lambda}\lesssim |\lambda|^2 
% \norm{ Q v(X) R^{\mathrm{Landau}}_\lambda(z_\lambda) \partial\chi(X)} \norm{P^\alpha A^\theta_\lambda}.
% }
% Using \cref{lem:derivative of resolvent also bounded} \corr FIX THIS \nc, we bound the second factor by $\Lambda^{-1/2+\eta}$. Further,  using the off-diagonal bound \cref{eq:Landau off diagonal decay} of \cref{prop:analytic extension of Landau resolvent on compact sets} to bound this term from above 
%  by $\exp(-C\Re{\lambda})$, we find
%  $\dist(\supp(v),\supp(\partial_j\chi))\ge Ca$ where $Ca$ is of order $1$.
%   Hence, $\norm{K^\flat_\lambda}\lesssim \exp(-C'\Lambda)$ for some order one constant $C'$. Hence, $\norm{K^\flat_\lambda}<1$ is realized by taking $\Lambda$ sufficiently large.

To complete the proof of \cref{lem:continuation of one-well resolvent}, we  derive the bound \cref{eq:off-diagonal exp decay of resolvent}. Let $S\subseteq\RR^2$ be such that $\dist(S,B_{2a}(0))>0$. Then we have using $\chi_S P_j^\alpha \chi A_\lambda^\theta Q=0$,
\eql{\label{eq:estimate on continuation of resolvent}
    \norm{\chi_S P_j^\alpha A_\lambda } &= \norm{\chi_S P_j^\alpha \br{\chi-R^{\mathrm{Landau}}_\lambda(z_\lambda)  \left[h_\lambda,\chi\right]} A_\lambda^\theta Q \br{\Id+K^\flat_\lambda}^{-1} } \notag \\
    &= \norm{\chi_S P_j^\alpha R^{\mathrm{Landau}}_\lambda(z_\lambda)  \left[h_\lambda,\chi\right] A_\lambda^\theta Q \br{\Id+K^\flat_\lambda}^{-1} } \notag %\\
    %&\leq \norm{\chi_S P_j^\alpha R^{\mathrm{Landau}}_\lambda(z_\lambda)  \left[h_\lambda,\chi\right]  }\norm{A_\lambda^\theta }\frac{1}{1-\norm{K^\flat_\lambda}}\notag\\
    %&\lesssim \exp\left(-C\dist(S,B_{2a}(0))\Lambda\right)} 
}
Now, $[h_\lambda,\chi]=\sum_{\abs{\beta}=0,1}a_\beta(X)P^\beta$ with $a_\beta$ supported in an annulus $T$ that avoids both $S$ and $B_a(0)$. Hence the above product operator is a sum of terms of the form
\eq{
\chi_S P_j^\alpha R^{\mathrm{Landau}}_\lambda(z_\lambda)  \chi_T\times \chi_T a_\beta(X) P^\beta A_\lambda^\theta Q \br{\Id+K^\flat_\lambda}^{-1}\,.
} The first factor is a bounded $L^2\to L^2$ operator with norm bounded by some order $1$ constant. The second factor is a bounded $QL^2\to L^2$ operator with norm bounded by \eq{
C \exp\br{-c\abs{\lambda}\dist(T,B_a(0))}\,.
}

Summarizing the above, we have the following 

\paragraph{Intermediate conclusion.}  $A_\lambda\in\calB(QL^2(\RR^2))$ is an operator with the following properties  
\begin{enumerate}
    \item $\br{h_\lambda-z_\lambda\Id}A_\lambda = Q$ for all $\lambda\in\Omega_{\Lambda,2\Lambda,\varepsilon}$.
    \item The mapping $\Omega_{\Lambda,2\Lambda,\varepsilon}\ni\lambda\mapsto A_\lambda\in\calB(QL^2(\RR^2))$ is analytic.
    \item We have the bound $\norm{A_\lambda}_{\calB(QL^2(\RR^2))}\lesssim \replacedm{\Lambda^{-\eta}}{\Lambda^\sharp}$\deleted{ for some $C\in\RR$}, for all $\lambda\in\Omega_{\Lambda,2\Lambda,\varepsilon}$.
    \item We have off-diagonal exponential decay in the sense that \eq{
        \norm{\chi_S A_\lambda }_{\calB(QL^2(\RR^2)\to L^2(\RR^2))} \lesssim \exp\br{-c \abs{\lambda} \dist(S,B_{2a}(0))}
    } for some $c>0$, for all $\lambda\in\Omega_{\Lambda,2\Lambda,\varepsilon}$ and $S\subseteq\RR^2$ such that $\dist(S,B_{2a}(0))>0$.
\end{enumerate}

\subsubsection{From $A_\lambda$ for $\lambda\in \Omega_{\Lambda, 2\Lambda,\ve}$ to $A_\lambda$ for $\lambda\in \Omega_{\Lambda,\ve}$.}\label{sec:annulus-to-all-Omega-Lambda}

Fix a number \(r\in(1,2)\) and set
\eq{
\Lambda_k:=r^k\Lambda,\qquad
\mathcal{D}_k:=\Omega_{\Lambda_k,2\Lambda_k,\varepsilon}
=\left\{z\in\CC:\Lambda_k<|z|<2\Lambda_k,|\arg z|<\tfrac{\pi}{2}-\varepsilon\right\}.
}
We have  \(\bigcup_{k\ge0}\mathcal{D}_k=\Omega_{\Lambda,\varepsilon}\) and, since \(r<2\), consecutive domains have a nontrivial overlap:
\eq{
\mathcal{D}_k\cap\mathcal{D}_{k+1}
=\left\{z:\Lambda_{k+1}<|z|<2\Lambda_k,|\arg z|<\tfrac{\pi}{2}-\varepsilon\right\}\neq\varnothing\ .
}

For each \(k\) the construction above produces an analytic map
\eq{
\lambda\mapsto A^{(k)}_\lambda,\quad 
 \mathcal{D}_k\longrightarrow \calB(QL^2(\RR^2)\to L^2(\RR^2))
}
such that 
\eq{
\left(h_\lambda-z_\lambda\Id\right)A^{(k)}_\lambda=Q
\qquad(\lambda\in\mathcal{D}_k),
}
and satisfies the bounds  \eq{ \|A^{(k)}_\lambda\|\lesssim \Lambda_k^{\sharp}\quad {\rm and}\quad   \norm{\chi_S A_\lambda }_{\calB(QL^2(\RR^2)\to L^2(\RR^2))} \lesssim \exp\br{-c \Lambda_k \dist(S,B_{2a}(0))}.} 
Recall that $z_\lambda=z_\lambda(\xi)$ and that these bounds hold uniformly  in $\xi\in\partial B_1(0)$.

Fix \(k\) and consider \(\lambda\in\mathcal{D}_k\cap\mathcal{D}_{k+1}\) with \(\lambda>0\) real. By construction $z_\lambda(\xi)$ is in the resolvent set for all \(\xi\in\partial B_1(0)\); hence \(\br{h_\lambda-z_\lambda\Id}^{-1}\) exists as a bounded operator on \(L^2(\RR^2)\). Since both \(A^{(k)}_\lambda\) and \(A^{(k+1)}_\lambda\) are \emph{right} inverses on \(QL^2\),
\eq{
\left(h_\lambda-z_\lambda\Id\right)A^{(k)}_\lambda
=\left(h_\lambda-z_\lambda\Id\right)A^{(k+1)}_\lambda
=Q,
}
we must have
\eq{
A^{(k)}_\lambda=\br{h_\lambda-z_\lambda\Id}^{-1}Q
=A^{(k+1)}_\lambda
\qquad(\lambda\in(\Lambda_{k+1},2\Lambda_k)\subset\RR).
}
Thus the two analytic \(\calB(QL^2)\)-valued maps \(A^{(k)}\) and \(A^{(k+1)}\) agree on the real interval \((\Lambda_{k+1},2\Lambda_k)\). Since the overlaps \(\mathcal{D}_k\cap\mathcal{D}_{k+1}\)   are sets with accumulation points, the identity theorem for Banach space –valued holomorphic functions implies that they coincide on all of \(\mathcal{D}_k\cap\mathcal{D}_{k+1}\).

Define \(A_\lambda\) on \(\Omega_{\Lambda,\varepsilon}\) by choosing any \(k\) with \(\lambda\in\mathcal{D}_k\) and setting \(A_\lambda:=A^{(k)}_\lambda\).
The overlap consistency implies that $A_\lambda$ is well defined on all $\Omega_{\Lambda,\varepsilon}$ and 
\eq{
\lambda\mapsto A_\lambda,\quad \Omega_{\Lambda,\varepsilon}\longrightarrow\calB(QL^2(\RR^2))
}
is analytic. Furthermore for \(\lambda\in\mathcal{D}_k\) we have the  bounds
\begin{subequations}
\label{eq:A_lambda_bounds}
\eql{
\|A_\lambda\|=\|A^{(k)}_\lambda\|&\lesssim\Lambda_k^\sharp
\le\abs{\lambda}^\sharp\\
\norm{\chi_S A_\lambda } = \norm{\chi_S A^k_\lambda }&\lesssim \exp\br{-c \Lambda_k \dist(S,B_{2a}(0))}\\
 &\leq \exp\br{-\tilde{c} \abs{\lambda} \dist(S,B_{2a}(0))}.
}
\end{subequations}
%hold because \(|\lambda|\in(\Lambda_k,2\Lambda_k)\).

%If one repeats the construction on \(\mathcal{D}_k\) with different auxiliary data \((\theta,\chi,\delta,\{\chi_\nu\},\{\psi_\nu\})\), the two resulting analytic right inverses agree for all real \(\lambda\in(\Lambda_k,2\Lambda_k)\) (both equal \(\br{h_\lambda-z_\lambda\Id}^{-1}Q\)), hence agree on \(\mathcal{D}_k\) by the identity theorem. Therefore the global \(A_\lambda\) is canonical  and independent of these choices.

The proof of \cref{lem:continuation of one-well resolvent} is now complete, up to the technical points which we've left to later sections.
\end{proof}

% \paragraph{Resulting object.} We summarize we have obtained
% \begin{enumerate}
%     \item $\br{h_\lambda-z_\lambda\Id}A_\lambda = Q$ for all $\lambda\in\Omega_{\Lambda,\varepsilon}$.
%     \item $\Omega_{\Lambda,\varepsilon}\ni\lambda\mapsto A_\lambda\in\calB(QL^2(\RR^2))$ is analytic.
%     \item We have the bound $\norm{A_\lambda}_{\calB(QL^2(\RR^2))}\lesssim \exp\br{Cb\abs{\lambda}}$, for all $\lambda\in\Omega_{\Lambda,\varepsilon}$.
% \end{enumerate}

\subsection{From the analyticity of $\lambda\mapsto \widetilde\vf_{\lambda,b}$ to a lower bound on $\rho(\lambda,b)$ for $\lambda\ge\Lambda$}

The analyticity of $\lambda\in\Omega_{\Lambda,\varepsilon}\mapsto A_\lambda \in \mathcal{B}(L^2(\RR^2))$ implies, via the strategy sketch in \cref{sec:setup-strategy}, that $\lambda\mapsto \widetilde\vf_\lambda$, the non-normalized ground state, is analytic on $\Omega_{\Lambda,\varepsilon}$. It follows that the non-dimensional hopping coefficient \cref{eq:non-norm-rMHO-ext}
 \eq{
\Omega_{\Lambda,\ve}\ni\lambda\mapsto \widetilde{\rho}(\lambda)\in \CC\quad \textrm{is analytic.}
} 
Further, the strategy sketch in \cref{sec:setup-strategy} reduces the proof of the lower bound asserted in \cref{thm:main theorem} for $\rho(\lambda)$ to the corresponding lower bound for $\widetilde\rho(\lambda)$. For this we now 
apply Complex Analysis \cref{cor:main complex analytic lower bound} to $\widetilde\rho(\lambda)$.  That is,  we must prove (a) an exponential upper bound for   $|\widetilde\rho(\lambda)|$ on $\Omega_{\Lambda,\ve}$ and  (b) 
  a lower bound for   $|\widetilde\rho(\lambda)|$ at a single point in $\Omega_{\Lambda,\ve}$.
  
\paragraph{Exponential upper bound on $\tilde{\rho}_\lambda$ for  $\lambda\in\Omega_{\Lambda,\ve}$.} For the upper bound, we use the expression for $\widetilde\vf_\lambda$ in \cref{eq:Riesz projection}, the bounds on $\|A_\lambda\|$ (see \cref{eq:A_lambda_bounds}) and on $\|\Phi_\lambda\|=\|Q\vf_\lambda^{\rm MHO}\|$. 

For each $\xi\in\mathbb S^1$, let $A_{\lambda,\xi}$ denote the analytic continuation of
\eq{
\br{h_\lambda-z_\lambda(\xi)\Id}^{-1}Q
} constructed above, but now with the $\xi$ parameter explicit. We then define
\eq{
\Pi_\lambda:=\frac{\ii}{2\pi}\oint_{|\xi|=1}A_{\lambda,\xi}z_\lambda'(\xi)\dif{\xi},
\qquad
\widetilde\vf_\lambda:=\Pi_\lambda\Phi_\lambda.
}
Thus $\widetilde\vf_\lambda$ is the non-normalized ground state obtained from the Riesz projection. Since $\lambda\mapsto A_{\lambda,\xi}$ is analytic, uniformly in $\xi\in\mathbb S^1$, it follows that
\eq{
\Omega_{\Lambda,\varepsilon}\ni\lambda\mapsto \widetilde\vf_\lambda\in L^2(\RR^2)
}
is analytic. Consequently,
\eql{\label{eq:non-norm-rho-revised-2}
\widetilde\rho(\lambda,b):=
\lambda^2\int_{x\in\RR^2}\overline{\widetilde\vf_{\overline\lambda,b\nc}(x+d)}
\,v(x+d)\exp\br{\ii b\lambda d_1x_2}\widetilde\vf_{\lambda,b}(x-d)\dif{x}
}
is analytic on $\Omega_{\Lambda,\varepsilon}$ and we have the bound
\eq{
\abs{\widetilde\rho(\lambda)}
&=
\Bigg|\frac{\lambda^2}{4\pi^2}
\oint_{|\zeta|=1}\oint_{|\xi|=1}
\overline{z'_{\overline\lambda}(\zeta)}\,z'_\lambda(\xi)\\
&\qquad\times
\int_{x\in\RR^2}
\overline{\br{A_{\overline\lambda,\zeta,b\nc}\Phi_{\overline\lambda,b\nc}}(x+d)}
\,v(x+d)\exp\br{\ii b\lambda d_1x_2}\,
\br{A_{\lambda,\xi}\Phi_\lambda}(x-d)\dif{x}\dif{\xi}\dif{\bar\zeta}\Bigg|
\\
&\lesssim
\abs{\lambda}^4
\int_{x\in B_a(0)}
\exp\big(b\abs{\Im\lambda}d_1\abs{x_2}\big)
\norm{A_\lambda}^2\norm{\Phi_\lambda}^2\dif{x}\\
&\lesssim
\abs{\lambda}^\sharp \exp\br{Cb\abs{\lambda}}\,.
} \added{The factor $|\lambda|^4$ comes from the prefactor $\lambda^2$ and the two contour derivatives $z'_\lambda(\xi)=O(\lambda)$.} Indeed, since the integral in \cref{eq:non-norm-rMHO-ext} is over $x\in B_a(0)$, we have $|x_2|\le a$ and therefore
\eq{
\exp\br{b|\Im\lambda|d_1|x_2|}\le \exp\br{ab d_1|\lambda|},
}
while the domain of integration has finite measure independent of $\lambda$. Moreover, by \cref{eq:A_lambda_bounds} we have $\norm{A_\lambda}\lesssim |\lambda|^{\sharp}$, and $\norm{\Phi_\lambda}=\norm{\CharFun_{B_a(0)}\vf_\lambda^{\rm MHO}}$ is uniformly bounded on $\Omega_{\Lambda,\ve}$ since $\vf_\lambda^{\rm MHO}$ is a Gaussian and $\Re\lambda\ge |\lambda|\sin\ve>0$ there; after enlarging $C$ and $\sharp$, we get the estimate.

\paragraph{The lower bound at a point.}

In this section we shall take $b$ sufficiently small, and hence make explicit the dependence of $\rho$ on both $\lambda$ and  $b$; we write  
$\rho_0(\lambda,b)$.  
 
Set $b=0$,  and let $\lambda_\star\geq2\Lambda$. Previous work (see for example \cite[Section 15.3]{FLW17_doi:10.1002/cpa.21735}) implies
\eql{
\abs{\rho(\lambda_\star,b=0)} \geq C \lambda_\star^2 \exp\br{-2 c d_1 \lambda_\star}
} for some constants $C$ and $c$. Therefore, for  $b_\star>0$ sufficiently small, 
\eq{
\abs{\rho(\lambda_\star,b)}\geq \frac{1}{2}\abs{\rho(\lambda_\star,b=0)}\quad \textrm{ for all $b\in(0,b_\star)$ and $\lambda\ge 2\Lambda$.}
}
Therefore, by \cref{eq:non-norm-rMHO-ext-real} and \cref{eq:C_lambda_to_1}, we have the lower bound:
\eq{
	|\widetilde{\rho}(\lambda_*,b)| = |\left\langle\vf_{\lambda_*,b},\Phi_{\lambda_*,b}\right\rangle|^2 \rho(\lambda_*,b)\gtrsim \frac{1}{2}\br{1-\frac{C'}{\lambda_\star}} C\lambda_\star^2\exp\br{-2 c d_1\lambda_\star}\,
}
for any fixed $b\in(0,b_*)$ and $\lambda_*\in \Omega_{\Lambda,\varepsilon}$.

\paragraph{Applying Complex Analysis \cref{lem:complex-analysis}.}

To put our problem in the framework of \cref{cor:main complex analytic lower bound} we  first map the wedge $\Omega_{\Lambda,\varepsilon}$ to one with its vertex at the origin using the change of variables:
\eq{ \lambda := \br{\lambda_\star-\Lambda} z +\Lambda.}
Here, $\lambda=\Lambda$ is mapped to $z=0$ and $\lambda=\lambda_*$ is mapped to $z=1$. 
We have that
\eq{ \lambda\in \Omega_{\Lambda,\varepsilon}\Longleftarrow z\in \Gamma_\alpha := \Set{\zeta\in\CC:-\alpha\pi/2<\operatorname{Arg}(\zeta)<\alpha\pi/2},\textrm{with $\alpha := 1-\frac{2\ve}{\pi}$.}}  
The function $F(z) := (\widetilde{\rho}\circ\lambda)(z)$ is analytic on $\Gamma_\alpha$. Our lower bound on $|\widetilde{\rho}(\lambda,b)$ at $\lambda=\lambda_*$ takes the form:
\eq{
|F(1)|=\left.\left|(\widetilde{\rho}\circ\lambda)(z)\right|\right|_{z=1} \geq  \exp\br{\log\br{\frac{1}{2}\br{1-\frac{1}{\lambda_\star}}C \lambda_\star^2}-2c d_1\lambda_\star} =: \ee^{-\beta}, 
} where 
\eq{
\beta := 2c d_1\lambda_\star - \log\br{\frac{1}{2}\br{1-\frac{1}{\lambda_\star}}C \lambda_\star^2}\,.
}
Our upper bound for all $\lambda\in\Omega_{\Lambda,\varepsilon}$ becomes  
\eq{
|F(z)| = \left|(\widetilde{\rho}\circ\lambda)(z)\right| 
\leq \abs{U(z)},
}
\eq{ \textrm{where}\quad U(z) :=(\br{\lambda_\star-\Lambda} z+\Lambda)^\sharp\exp\br{C b\br{ (\lambda_\star-\Lambda)z+\Lambda}}.}

\cref{lem:complex-analysis} implies  
\eq{
\frac{1}{R}\int_{R}^{2R}-\log\br{ \abs{\widetilde{\rho}((\lambda_\star-\Lambda)\zeta+\Lambda)} }\dif{\zeta} &\leq C\alpha 2^{1/\alpha}\left(\beta+\log\left(\abs{U(1)}\right)\right) R^{1/\alpha}+\frac{3}{2}c b R\, ,
}
which, after the change of variables $\zeta=(\lambda-\Lambda)/(\lambda_*-\Lambda)$, yields  the required lower bounds  on averages  of  $\widetilde{\rho}_\lambda$.

And, as discussed earlier, together with \cref{eq:non-norm-rMHO-ext-real} and \cref{eq:C_lambda_to_1}, this implies part 3 of \cref{thm:main theorem}, concerning averages of $\rho(\lambda,b)$ for $\lambda>\Lambda$ and $0<b<b_\varepsilon$. 

\deleted{For the first part, we make the }

\begin{rem}[\cref{eq:avg lower bound on rho} implies \cref{eq:density of zeros}]
A standard argument shows how an averaged lower bound yields a pointwise lower bound outside a set of density zero.

Indeed, assume that for some function $F$,
\eq{
    \frac{1}{R}\int_R^{2R}\mLog{F(t)}\dif{t}\leq R^{1+\ve}
    \qquad(R>1)\,.
}
Set
\eq{
    R_k:=2^k\qquad(k\geq0)\,,
}
and define
\eq{
    E_k:=\Set{t\in[R_k,R_{k+1}]\,:\,\mLog{F(t)}>t^{1+2\ve}}\,,
}
as well as
\eq{
    E_k^+:=\Set{t\in[R_k,R_{k+1}]\,:\,\mLog{F(t)}>R_k^{1+2\ve}}\,.
}
Finally, let
\eq{
    E:=[0,1]\cup\bigcup_{k\geq0}E_k\,.
}

Then for all $t\in[0,\infty)\setminus E$ we have
\eq{
    \mLog{F(t)}\leq t^{1+2\ve}\,,
}
that is,
\eq{
    \abs{F(t)}\geq \exp\br{-t^{1+2\ve}}\,.
}

Since $E_k\subseteq E_k^+$, we have
\eq{
    \mu(E_k)
    \leq \mu(E_k^+)
    \leq \frac{1}{R_k^{1+2\ve}}\int_{R_k}^{2R_k}\mLog{F(t)}\dif{t}
    \leq R_k^{1-\ve}\,.
}
Hence
\eq{
    \mu\br{E\cap[0,R_N]}
    \leq \mu([0,1])+\sum_{0\leq k\leq N}\mu(E_k)
    \leq 1+\sum_{0\leq k\leq N}R_k^{1-\ve}
    \leq C\,R_N^{1-\ve}\,.
}
In particular,
\eq{
    \mu\br{E\cap[0,2^N]}\leq C(2^N)^{1-\ve}\,.
}

Now let $R$ be large, and choose $N$ such that $2^N\leq R\leq 2^{N+1}$. Then
\eq{
    \frac{\mu\br{E\cap[0,R]}}{R}
    \leq C'R^{-\ve}\,.
}
Therefore $E$ has density zero.
\end{rem}

Finally, we note that the statement on sparsity in \cref{thm:main theorem} follows
 from the bounds on sums of negative powers of the zeros \cref{eq:sketch-zero-sums} appearing in the Blaschke factorization used in the proof of the complex analytic \cref{lem:complex-analysis}.

\subsection{Lower bound on $\Delta(\lambda,b)$ of \cref{thm:main theorem}; the proof} \label{subsec:lower bound on Delta}

First, a preliminary remark. A strategy one can attempt to obtain lower bounds on $\Delta(\lambda)$ from lower bounds on $\rho(\lambda,b)$ would be to attempt to establish a relation like $\Delta(\lambda,b)/ 2\rho(\lambda,b)\to1$ as $\lambda\to\infty$; see, e.g. \cite{fefferman2025lowerboundsquantumtunneling}).  However, we have merely 
 $\rho(\lambda)\gtrsim\exp(-c\lambda^{1+\varepsilon})$ 
  which is small compared to error terms  $\lesssim \exp(-c\lambda)$. We therefore work directly with $\Delta(\lambda,b)$, and pursue a strategy analogous to that in our lower bound for  $\rho(\lambda,b)$.

We sketch the following outline, with certain technical details left to the end of the section.

\begin{enumerate}
    \item Let $H_\lambda=\br{P-\frac12b\lambda X^\perp}^2+\lambda^2 v(X+d)+\lambda^2 v(X-d)$ (see \cref{eq:double-well Hamiltonian}) denote the magnetic double-well Hamiltonian. Recall that the spectral parameter $z_\lambda=z_\lambda(\xi)$ is given by \cref{eq:spectral parameter for Riesz} and that 
$z_\lambda(\xi)$ is in the resolvent set of $H_\lambda$ for $\lambda>\Lambda$ and  all $|\xi|=1$. 

For the same reasons explained above, in the context of the single well 
    magnetic Hamiltonian, the  double-well resolvent 
    \eql{\label{eq:resolvent of double well Hamiltonian}\lambda\mapsto R_\lambda(z_\lambda)\equiv \br{H_\lambda-\replacedm{z}{z_\lambda}\Id}^{-1}
    } 
    does not have an analytic extension to  $\lambda\in\Omega_{\Lambda,\ve}$ (see \cref{rem:analytic_ext}) and we proceed with the analogous remedy: analytic continuation of $R_\lambda(z_\lambda)$ acting on functions supported, now, in $B_a(-d)\cup B_a(d) $.  
    Introduce the spatial projections  
    \eq{ Q_{\pm d}:=\CharFun_{B_{a}(\pm d)}(X)\quad {\rm  and}\quad  Q_{-d,d}:=Q_{-d}+Q_d.}
     We shall analytically continue the  operator-valued function \eq{
    \Omega_{\Lambda,\varepsilon}\ni \lambda\mapsto 
    R_\lambda(z_\lambda)Q_{-d,d}\in \mathcal{B}(L^2(\RR^2)).
    }
   We cannot directly use the proof of analytic continuation of $\br{h_\lambda-z_\lambda(\xi)\Id}^{-1}Q$, but just below  in \cref{sec:construct-resolvent-Q_pmd} we present a perturbative argument to pass from the analytic extension of $\lambda\mapsto A_\lambda = \br{h_\lambda-z_\lambda\Id}^{-1} Q$ to such an extension for $\lambda\mapsto R_\lambda(z_\lambda)Q_{-d,d}$.
    
    \item For real $\lambda$, let 
    $E_0(\lambda)\leq E_1(\lambda)$ denote the two lowest eigenvalues of $H_\lambda$ and let $\calV_\lambda\subseteq L^2(\RR^2)$ denote  the corresponding two-dimensional low-energy eigenspace of $H_\lambda$; as shown in \cite{fefferman2025magneticdoublewellsabsencetunneling}, in general these eigenvalues may turn out to be degenerate. The results of Matsumoto \cite{Matsumoto_1994} imply that for $\lambda>\Lambda$ large, $E_{0,\lambda}$ and $ E_{1,\lambda}$ are located in a neighborhood of the ground state eigenvalue of $h_\lambda$,  
    $e_{\lambda,0}=-\lambda^{2}+e_{0}^{\text{MHO}}\lambda+\mathcal{O}(\lambda^{+1/2})$; in particular they are within the circle: $\xi\mapsto z_\lambda(\xi)$, where $z_\lambda(\xi)$ is defined by \cref{eq:spectral parameter for Riesz}, and therefore
   \eq{
    \Pi_\lambda =\frac{\ii}{2\pi}\oint_{\xi\in\partial B_1(0)}R_\lambda(z_\lambda) z_\lambda'(\xi)\dif{\xi},\quad \lambda>\Lambda,    
    } is the spectral projection onto $\calV_\lambda$. Hence,
    \eq{ \sigma\left(\Pi_\lambda H_\lambda \Pi_\lambda\right)\setminus\Set{0} = \{E_0(\lambda), E_1(\lambda)\} . }
    %after projection with $Q_{-d,d}$.
% $\Pi_\lambda Q_{-d,d}=$
We shall construct a basis for $\calV_\lambda$ and a $2\times 2$ matrix representation of $\Pi_\lambda H_\lambda \Pi_\lambda:\calV_\lambda\to\calV_\lambda$.

 \item {\it Basis for $\calV_\lambda$}: It is natural to construct \replaced{at}{a} basis for $\calV_\lambda$ using a strategy, analogous to that in our construction and analytic continuation of the ground 
 state, $\vf_\lambda$, of the one-well magnetic Hamiltonian, $h_\lambda$. This would suggest applying the projection $\Pi_\lambda$ to the vectors $Q_{\pm d} \widehat{R}^{\pm d}\vf_\lambda^{\rm MHO}$, where $\vf_\lambda^{\rm MHO}$ denotes the magnetic harmonic oscillator ground state and $\widehat{R}^d$ denotes magnetic translation by $d$. Since, it turns out, that we will need to bound derivatives of our analytically continued basis vectors, we modify this scheme by inserting  a smooth cutoff function.

 Introduce the vectors 
\eql{\label{eq:Phi-pmd}
    \Phi_{\pm d} := Q_{\pm d}\widehat{R}^{\pm d}\chi(X)\vf_\lambda^{\rm MHO},
} where $\chi\in C^\infty(\RR^2\to[0,1])$ is a cutoff function equal to $1$ on $B_{a/2 }(0)$ and zero outside $B_{3a/4 }(0)$;
so that the projection onto $Q_{\pm d}$ is actually redundant. We use $\chi$ so that $\Phi_{\pm d}$ is smooth with compact support.  

The projections of $\Phi_{\pm d}$ by $\Pi_\lambda$:
\eql{\label{eq:psi-pmd} \psi_{\pm d} := \Pi_\lambda Q_{-d,d} \Phi_{\pm d}} 
are analytic in $\Omega_{\Lambda,\varepsilon}$, and 
form a (non-orthonormal) basis for $\calV_\lambda$:
\eql{\calV_\lambda =\operatorname{span}\Set{\psi_{\lambda,-d},\psi_{\lambda,d}}.
}
 As in our analysis of the magnetic single well Hamiltonian, $h_\lambda$, we work with non-normalized states; their normalizations are not necessarily analytic.

    \item {\it Matrix representation for   $\Pi_\lambda H_\lambda \Pi_\lambda$ and a relation for $\Delta(\lambda)$}: Introduce the  $2\times 2$ matrices \eql{\label{eq:Gramian and two by two matrix}
    M_\lambda := \left(\ip{\psi_{\alpha,\overline{\lambda}\color{black}}}{ H_\lambda \psi_{\beta,\lambda} }\right)_{\alpha,\beta=\pm d} \qquad 
    G_\lambda := \left(\ip{\psi_{\alpha,\overline{\lambda}\color{black}}}{ \psi_{\beta,\lambda} }\right)_{\alpha,\beta=\pm d}\,.
    }  
    The previous considerations imply that the mappings $\lambda\mapsto M_\lambda$ and 
    $\lambda\mapsto G_\lambda$ are analytic on $\Omega_{\Lambda,\varepsilon}$.

    For $\lambda\in\RR$, $G_\lambda$ is the Gramian matrix of inner products and  $\Pi_\lambda H_\lambda \Pi_\lambda $ has $2\times2$ matrix representation $G_\lambda^{-1} M_\lambda$. Therefore the two eigenvalues of 
    $\Pi_\lambda H_\lambda \Pi_\lambda$, $E_0(\lambda)$ and $E_1(\lambda)$, are the  eigenvalues of the $2\times2$ matrix $G^{-1}_\lambda M_\lambda$. Solving for the 
     two roots of $\det(G_\lambda^{-1} M_\lambda-E\Id)=0$, taking their difference and squaring yields the relation for $\Delta(\lambda)^2\equiv (E_1(\lambda)-E_0(\lambda))^2$:
    \eql{\label{eq:tilde-Delta-squared}
    \det(G_\lambda)^2 \Delta(\lambda)^2 = \tr\br{\operatorname{adj}(G_\lambda) M_\lambda}^2 - 4 \det(G_\lambda) \det(M_\lambda).
    } Here we have used that $G^{-1} = {\rm adj}(G_\lambda)/\det(G_\lambda)$, and 
     \eql{\label{eq:tilde-Delta-squared repeated}
    \operatorname{adj}(G) := 
    \begin{bmatrix}
        g_{22} &  - g_{12} \\
        -g_{21} & g_{11}
    \end{bmatrix}.  }
    %where $\operatorname{adj}(M)$ denotes the $2\times2$ matrix given by the expression:
   % \eql{\label{eq:splitting via 2 by 2 matrix}
    %\det(G_\lambda)^2 \Delta_0(\lambda)^2 = \tr\br{\operatorname{adj}(G_\lambda) %M_\lambda}^2 - 4 \det(G_\lambda) \det(M_\lambda)
    %\qquad(\lambda\in\RR)\nonumber\\}  and so clearly the map
    It follows that 
    \eql{
    \RR\cap\Set{\lambda>\Lambda}\ni\lambda \mapsto \widetilde\Delta(\lambda)^2:=\det(G_\lambda)^2 \Delta(\lambda)^2\in   \RR
    } continues to an analytic function  $\widetilde\Delta(\lambda)^2: \Omega_{\Lambda,\ve}\longrightarrow \CC$.
    \item {\it Obtaining a lower bound on $\Delta(\lambda,b)$:} Above, in order to bound $|\rho(\lambda,b)|$ from below, for $\lambda>\Lambda$ and $0<b<b_\varepsilon$, the key first step was to 
    prove a lower bound for the analytic function $\widetilde{\rho}(\lambda,b)$, for $\lambda\in\Omega_{\Lambda,\varepsilon}$ and $0<b<b_\varepsilon$,  using the complex analytic \cref{lem:complex-analysis}. We now apply a similar strategy to $\Delta(\lambda,b)^2$ for $\lambda>\Lambda$ and $0<b<b_\varepsilon$; we first (a) bound $\widetilde\Delta(\lambda,b)^2$ away from zero at a single point of $\Omega_{\Lambda,\varepsilon}$ and then (b) obtain an exponential upper bound for 
    $\widetilde\Delta(\lambda)^2$ throughout  $\Omega_{\Lambda,\varepsilon}$. 
    
    The lower bound for $\widetilde\Delta(\lambda_*,b)^2$, for some point $\lambda_*>\Lambda$ and all $0<b<b_\varepsilon$ sufficiently small follows from the same argument (a perturbative argument about the case $b=0$) which proved this property  for  $\widetilde{\rho}$. Specifically, in the non-magnetic case  we have $\widetilde\Delta(\lambda_*,b=0)^2>0$, and hence for  
   $0<b<b_\varepsilon$ sufficiently small and $\lambda_\star$ sufficiently large, $\widetilde\Delta(\lambda_*,b)^2$, given by \cref{eq:tilde-Delta-squared},  satisfies 
   a slightly weaker, but strictly positive, lower bound. 

   So to apply the complex analytic \cref{lem:complex-analysis}, it remains to provide an upper bound on $\widetilde\Delta(\lambda,b)^2$ on $\Omega_{\Lambda,2\Lambda,\varepsilon}$ for all $0<b<b_\varepsilon$. The expression for $\abs{\widetilde\Delta(\lambda,b)^2}$ is a polynomial in the matrix elements of $G_\lambda$ and $M_\lambda$, which can be bounded using bounds  on $\psi_{\pm d,\lambda}$ and their derivatives up to order two; see \cref{lem:extend-resolvent_Hlambda} below.

Finally, to pass from a lower bound on $\widetilde\Delta(\lambda,b)^2=\det(G_\lambda)^2 \Delta(\lambda,b)^2$ to one on $\Delta(\lambda,b)^2$ (for $\lambda>\Lambda$ and $0<b<b_\varepsilon$)  we use that
for real $\lambda>\Lambda$,  we have $\det G_\lambda\approx 1$; see \cref{lem:Gramian} below.
\end{enumerate}

This completes the outline of our lower bound results for $\Delta(\lambda,b)$; the remaining technicalities are treated below.

\subsubsection{Construction of $R_\lambda(z)Q_{-d,d}$}\label{sec:construct-resolvent-Q_pmd}
In our argument above for analytic continuation of  the single-well resolvent, $\br{h_\lambda-z_\lambda\Id}^{-1}Q$, the spatial projection $Q=\CharFun_{B_a(0)}$ is localized around $x=0$, where  the minimum of $v$ is  attained and  where $v$ is locally quadratic. To analytically continue the resolvent of the double-well Hamiltonian, $H_\lambda$, we adapt our argument to allow for a potential with two separated non-degenerate minima. 

\begin{lem}\label{lem:extend-resolvent_Hlambda}
Recall $R_\lambda(z_\lambda)=(H_\lambda-z_\lambda\Id)^{-1}$, the resolvent at spectral parameter $z=z_\lambda$ of the double-well Hamiltonian. The mapping  $\lambda\mapsto R_\lambda(z_\lambda)Q_{-d,d}$, extends from $\lambda>\Lambda$ to an operator-valued analytic function $A_{\lambda,\xi}^{-d,d}$ ($\xi\in \mathbb{S}^1$): \eq{\Omega_{\Lambda,\varepsilon}\ni\lambda \mapsto A_\lambda^{-d,d} \in \calB(Q_{-d,d}L^2(\RR^2)\to L^2(\RR^2)) 
%\quad \textcolor{red}{\calB(Q_{-d,d}L^2(\RR^2))?} 
 } with the bound
    \eql{\label{eq:estimates on the two-well resolvent}
    \norm{A_\lambda^{-d,d}}_{\calB(Q_{-d,d}L^2(\RR^2))} \lesssim \abs{\lambda}^{\sharp} \exp\br{C  b \abs{\lambda}} \,.
      } 
\end{lem}
% \paragraph{Route 1: direct argument}
% We proceed similarly to \cref{subsubsec:extending h theta} by analytically continuing the resolvent of $H_\lambda^\theta$, i.e., the two-well Hamiltonian with kinetic term $H^{{\rm Landau},\theta}_\lambda$ and then remove the cut off as in \cref{subsubsec:removing the cut off}. 

% The cut-off function $\theta$ we take now equals $1$ in $B_{2a}(\pm d)$ and it equals zero in, say, $\br{B_{3a}(d)\cup B_{3a}(-d)}^c$.

% We need to modify the partition of unity we take compared with \cref{eq:partition of unity for h theta lambda}. The term $\nu=N_\Lambda$ which is the exterior of some sufficiently large order $1$ disc where $\theta=0$ remains the same, but now, we need \emph{pairs} of terms for each $\nu=0,\cdots,N_\Lambda-1$: each centered at $\pm d$ but otherwise built exactly the same as in \cref{eq:partition of unity for h theta lambda}. The definition in \cref{eq:putative expression for the inverse of h_lambdatheta} is then generalized in the obvious way and the estimates on the three types of terms are identical. We obtain an analytic continuation of ${H_\lambda^\theta-z_\lambda\Id}^{-1}$ with the same estimates as in \cref{eq:estimate on analytic continuation of resolvent of cut off one well Hamiltonian}.

% To generalize \cref{subsubsec:removing the cut off}, we replace $Q$ there with $Q_{-d,d}$ and proceed in the obvious way (e.g., the remainder term now contains two copies of the potential, $v(X\pm d)$ instead of just $v(X)$, but both are in the range of $Q_{-d,d}$).
\begin{proof}
We shall use the magnetic translation operators given in \cref{eq:magnetic translations}, here factorized as
\eql{\widehat{R}^\xi =\exp\br{-\ii \xi\cdot P}\exp\br{-\ii \frac12b\lambda  \xi \cdot X^\perp},\qquad \xi\in\RR^2\,;} 
note $[\xi\cdot P,\xi\cdot X^\perp]=0$. These operators are unitary if $\lambda\in\RR$ and otherwise generally not \deleted{even }bounded; in factorized form the two exponentials commute. $\widehat{R}^z$ commutes with the magnetic kinetic energy $H^{\rm Landau}_{b,\lambda}$ and satisfies \eql{\label{eq:magnetic translations shift potential}
\widehat{R}^{\pm d} v(X) \widehat{R}^{\mp d} = v(X\mp d),
} and so we have the relation
\eql{\label{eq:magnetic translations shift Hamiltonian}
\widehat{R}^{\pm d} \br{H^{\rm Landau}_{b,\lambda} + \lambda^2v(X)} \widehat{R}^{\mp d} = H^{\rm Landau}_{b,\lambda} + \lambda^2v(X\mp d)\,.
}

Consequently, the analytic extension $\lambda\mapsto A_\lambda$  of $\lambda\mapsto\br{h_\lambda-z_\lambda\Id}^{-1}Q$ constructed above provides a right inverse of the translates of the one-well Hamiltonians as well:
\eq{
    \br{H^{\rm Landau}_{b,\lambda} + \lambda^2v(X\mp d)-z_\lambda\Id} A_\lambda^{\pm d} = Q_{\pm d} \, ,
} 
where
\eq{A_\lambda^{\pm d} := \widehat{R}^{\pm d} A_\lambda \widehat{R}^{\mp d} .}
Clearly, 
\eq{
 \Omega_{\Lambda,\ve} \ni\lambda\mapsto A_\lambda^{\pm d} \in \calB(Q_{\pm d}L^2(\RR^2)\to L^2(\RR^2))\quad 
 %\textcolor{red}{\calB(Q_{-d,d}L^2(\RR^2))?} 
} 
is also analytic. However,  if $\lambda$ is not real, we have a deteriorated bound since magnetic translations are not unitary. First note that in the magnetic translations, for complex $\lambda$, $\exp\br{-\ii \xi\cdot P}$ which is independent of $\lambda$, is still unitary and commutes with the other factor, so it is irrelevant in bounding the operator norm. We thus only keep the factor $\exp\br{-\ii\frac12b\lambda \xi \cdot X^\perp}$, where $\xi=\pm d$:
\eql{
\norm{A_\lambda^{\pm d}} &= \norm{\exp\br{\mp\ii \frac12b\lambda  d\cdot X^\perp}A_\lambda \exp\br{\pm\ii \frac12b\lambda  d\cdot X^\perp}}_{\calB(Q_{\pm d}L^2(\RR^2)\to L^2(\RR^2))} \nonumber \\
&\leq \norm{\exp\br{\mp\ii \frac12b\lambda  d\cdot X^\perp} \br{\chi_{B_{3a}(0)}(X)+\chi_{B_{3a}(0)}(X)^\perp} A_\lambda } \exp\br{\frac{b \abs{\Im{\lambda}}}{2} \norm{d}a}  \nonumber \\
&\leq \exp\br{3\frac{b \abs{\Im{\lambda}}}{2} \norm{d}a}\norm{ A_\lambda } \exp\br{\frac{b \abs{\Im{\lambda}}}{2} \norm{d}a} + \nonumber \\
&+ \norm{\exp\br{\mp\ii \frac12b\lambda  d\cdot X^\perp} \chi_{B_{3a}(0)}(X)^\perp A_\lambda } \exp\br{\frac{b \abs{\Im{\lambda}}}{2} \norm{d}a}\,.\label{eq:Alam_pm_d-split}
} 

Using  \cref{eq:bound on analytic continuation of resolvent of single well Hamiltonian}, the first term in \cref{eq:Alam_pm_d-split} is seen to have the upper bound   \[\lesssim |\lambda|^\sharp\exp\bigg(\ C b \|d\| |\Im{\lambda}|\ \bigg).\]
To deal with the second term in \cref{eq:Alam_pm_d-split}, write \eq{B_{3a}(0)^c = \bigsqcup_{n=0}^\infty \Set{x\in\RR^2 | 3a + n \leq \norm{x} < 3a + n + 1}}  and apply  the off-diagonal bound \cref{eq:off-diagonal exp decay of resolvent} to obtain, by again using the assumption that $b$ is sufficiently small that:
% \eq{
% &\norm{\exp\br{\mp\ii \frac12b\lambda  d\cdot X^\perp}
% \chi_{B_{3a}(0)}(X)^\perp A_\lambda }\\
% &\leq \sum_{n=0}^\infty \exp\Bigg\{
% \frac{b \abs{\Im{\lambda}}}{2}\norm{d}\br{3a+n+1}-c \abs{\lambda}\br{a+n}\Bigg\}.
% }
\begin{align*}
&\norm{\exp\br{\mp\ii \frac12b\lambda  d\cdot X^\perp}
\chi_{B_{3a}(0)}(X)^\perp A_\lambda }\\
&\lesssim \sum_{n=0}^\infty \norm{\exp\br{\mp\ii \frac12b\lambda  d\cdot X^\perp} \Id_{3a+n\le\|x\|\le 3a+n+1}(X) A_\lambda\ \Id_{\|x\|\le 2a}(X)}\\
&\lesssim \sum_{n=0}^\infty 
e^{\frac12 b |\Im{\lambda}| \|d\|(3a+n+1)}\ 
\norm{\Id_{3a+n\le\|x\|\le 3a+n+1}(X) A_\lambda\ \Id_{\|x\|\le 2a}(X)}\\
&\lesssim \sum_{n=0}^\infty 
\exp\big(\frac12 b |\Im{\lambda}| \|d\|(3a+n+1)\big)\ 
 \exp\big( -c|\lambda|(3a+n+1 - 2a)\big)\\
 &\lesssim \exp(-c|\lambda| a)
\end{align*}
Combining the bounds on the two terms in \cref{eq:Alam_pm_d-split} yields

\eq{
\norm{A_\lambda^{\pm d}} \lesssim \abs{\lambda}^{\sharp } \exp\br{C b \abs{\lambda}}.
} 

We now use $A_\lambda^{\pm d}$, to construct an approximate $R_\lambda(z_\lambda)Q_{-d,d}$. Let  
\eq{
\widetilde{A_{\lambda}^{-d,d}} :=  A_\lambda^{d} Q_d +  A_\lambda^{-d} Q_{-d}\,.
}

Then 

\eq{
\br{H_\lambda -z_\lambda\Id}\widetilde{A_{\lambda}^{-d,d}} &= Q_{-d,d}+\lambda^2 v(X+d) A_\lambda^{d} Q_d + \lambda^2 v(X-d) A_\lambda^{-d} Q_{-d}\,\\
&= Q_{-d,d}\br{\Id + \lambda^2 v(X+d) A_\lambda^{d} Q_d + \lambda^2 v(X-d) A_\lambda^{-d} Q_{-d}}\,.
} 
 For large $|\lambda|$, the correction 
\eq{\lambda^2 v(X+d) A_\lambda^{d} Q_d + \lambda^2 v(X-d) A_\lambda^{-d} Q_{-d}} has  small operator norm thanks to \cref{eq:off-diagonal exp decay of resolvent} of \cref{lem:continuation of one-well resolvent}, since  $v(X\pm d)$ is supported within $Q_{\mp d}$. The desired right inverse

\eq{
A_\lambda^{-d,d} := \widetilde{A_{\lambda}^{-d,d}}\br{\Id + \lambda^2 v(X+d) A_\lambda^{d} Q_d + \lambda^2 v(X-d) A_\lambda^{-d} Q_{-d}}^{-1}
}  exists and is analytic on $\Omega_{\Lambda,\varepsilon}$. The proof of \cref{lem:extend-resolvent_Hlambda} is now complete.\end{proof}

%The bound on its derivatives is obtained from the $\alpha>0$ case of \cref{eq:off-diagonal exp decay of resolvent}. 

\subsubsection{Estimates on the Gramian matrix, $G_\lambda$}

\begin{lem}\label{lem:Gramian} For $\lambda\in\Omega_{\Lambda,\ve}\cap\RR$,  the set  $\{\psi_{\lambda, -d},\psi_{\lambda, d}\}$ is linearly independent, nearly orthonormal:
    \eql{
        \abs{\ip{\psi_{\lambda,\alpha}}{\psi_{\lambda,\beta}}-\delta_{\alpha,\beta}}\lesssim \lambda^{-1/2},\qquad \alpha,\beta=\pm d\, .
    }
Therefore,  the Gramian matrix $G_\lambda$ satisfies the bound: $\norm{G_\lambda-\Id}\lesssim \lambda^{-1/2}$.
\end{lem}
\begin{proof} Recall that for $\lambda\in\RR$, $\vf_\lambda^{\rm MHO}$ is the normalized MHO ground state. From \cref{eq:Phi-pmd} and \cref{eq:psi-pmd} we have that
 for $\alpha, \beta =\pm d$, 
 \eq{
  \ip{\psi_{\alpha}}{\psi_{\beta}} &=
   \ip{\Pi_\lambda Q_\alpha\widehat{R}^\alpha\chi\vf_\lambda^{\rm MHO} }
   {\Pi_\lambda Q_\beta\widehat{R}^\beta\chi \vf_\lambda^{\rm MHO}}\\
   &= \ip{Q_\alpha\widehat{R}^\alpha\chi\vf_\lambda^{\rm MHO} }
   {Q_\beta\widehat{R}^\beta\chi \vf_\lambda^{\rm MHO}} -
    \ip{\Pi_\lambda^\perp Q_\alpha\widehat{R}^\alpha\chi\vf_\lambda^{\rm MHO} }
   {Q_\beta\widehat{R}^\beta\chi \vf_\lambda^{\rm MHO}}\\
   &\quad -  \ip{   Q_\alpha\widehat{R}^\alpha\chi\vf_\lambda^{\rm MHO} }
   {\Pi_\lambda^\perp Q_\beta\widehat{R}^\beta\chi \vf_\lambda^{\rm MHO}}
    + \ip{\Pi_\lambda^\perp  Q_\alpha\widehat{R}^\alpha\chi\vf_\lambda^{\rm MHO} }
   {\Pi_\lambda^\perp Q_\beta\widehat{R}^\beta\chi \vf_\lambda^{\rm MHO}}
}
By \cref{lem:translates of MHO ground state approximates double-well ground state} below we see that for real $\lambda$, $\norm{\Pi_\lambda^\perp \widehat{R}^{\pm d}\vf_\lambda^{\rm MHO}} \lesssim \lambda^{-1/2}$. Using this, we obtain for $\alpha=\beta$
%
%    \eq{
  %      \norm{\Pi_\lambda^\perp \Phi_{\pm d}} \lesssim \lambda^{-1/2} + \norm{\Pi_\lambda^\perp Q_{\pm d}^\perp\widehat{R}^{\pm d}\vf_\lambda^{\rm MHO}}\lesssim \lambda^{-1/2}\,.
 %   }
    \eq{
   \abs{ 1-\ip{\psi_\alpha}{\psi_\alpha}} =  \abs{1-\norm{\psi_{\alpha}}^2} \lesssim \lambda^{-1/2},\qquad  \alpha\in\{-d,+d\},
    }
and 
    \eq{
    \abs{\ip{\psi_{-d}}{\psi_{d}}} 
    %\abs{\ip{\Phi_{-d}}{\Pi_\lambda \Phi_{d}}} \\
  %  &= \abs{\ip{Q_{-d} \widehat{R}^{- d}\vf_\lambda^{\rm MHO}}%{\Pi_\lambda Q_{d} \widehat{R}^{ d}\vf_\lambda^{\rm MHO}}} \\
  %  &= \abs{\ip{Q_{-d} \widehat{R}^{- d}\vf_\lambda^{\rm MHO}}5{\Pi_\lambda^\perp Q_{d} \widehat{R}^{ d}\vf_\lambda^{\rm MHO}}} \\
    &\lesssim \lambda^{-1/2}\,.
    }

    To see this, note that the main term
    \[
    \ip{Q_{-d}\widehat{R}^{-d}\chi\vf_\lambda^{\rm MHO}}
       {Q_d\widehat{R}^{d}\chi\vf_\lambda^{\rm MHO}}
    \]
    is zero because $Q_{-d}$ and $Q_d$ localize to disjoint wells. Hence every term in the expansion of $\ip{\psi_{-d}}{\psi_d}$ contains at least one factor involving $\Pi_\lambda^\perp$, and these factors are $\mathcal O(\lambda^{-1/2})$ in norm by \cref{lem:translates of MHO ground state approximates double-well ground state}. The estimate then follows immediately from Cauchy--Schwarz.
    
\end{proof}

\subsubsection{Bounds on matrix elements of $G_\lambda$ and $M_\lambda$}

As remarked above in Step 5 of the proof strategy, to obtain bounds on $\widetilde\Delta(\lambda,b)$, given by \cref{eq:tilde-Delta-squared}, it suffices to bound the matrix entries of $G_\lambda$ and $M_\lambda$. These are controlled by:
\eq{ \Omega_{\Lambda,\varepsilon}\ni\lambda \mapsto  \psi_{\alpha,\lambda}\qquad(\alpha=\pm d)}
and 
\eq{ \Omega_{\Lambda,\varepsilon}\ni\lambda \mapsto H_\lambda \psi_{\alpha,\lambda}\qquad(\alpha=\pm d)}
which are analytic with appropriate upper bounds, as we explain next. 

Recall that $\psi_\alpha = \Pi_\lambda Q_\alpha \widehat{R}^\alpha \chi  \vf_\lambda^{\rm MHO}$, $\alpha=\pm d$, with $\chi$ smooth and supported in $B_{3a/4}(0)$;  see \cref{eq:Phi-pmd} and \cref{eq:psi-pmd}.  Hence, 
\eq{
\norm{\psi_\alpha} \leq \norm{\Pi_\lambda Q_\alpha} \norm{\widehat{R}^\alpha \chi \vf_{\lambda}^{\rm MHO}}
} 
The appearance of the magnetic translation, $\widehat{R}^d$  leads, for  complex $\lambda$, to upper bounds which are exponentially growing   in $\lambda$. Indeed,  \emph{both} factors satisfy a bound $\lesssim |\lambda|^\sharp\exp\br{C b \abs{\lambda}}$. 

% Using \cref{eq:estimates on the two-well resolvent} with $r=0,1,2$ we get
% \eql{\label{eq:derivatives of quasimodes also bounded}
%     \norm{P^r_j \psi_{\pm d}} \lesssim \abs{\lambda}^{\sharp}\exp\br{C b \abs{\lambda}},\quad {\rm for}\quad j=1,2;\,r=0,1,2\,.
% } Indeed, for $\alpha=1,2$ the action of $P_j$ on $\psi_{\pm d}$ follows the Leibniz differentiation rule. Since $\chi$ has support within $B_{a/2}(0)$ the projections $Q_{\pm d}$ are in fact redundant. When $P_j$ acts on  $\Pi_\lambda$ we use the Riesz projection representation to estimate $P A_\lambda^{-d,d}$ and then apply the bound \cref{eq:estimates on the two-well resolvent}.  
%  Thus with \cref{eq:derivatives of quasimodes also bounded}, the matrix elements of $M_\lambda$  have the necessary exponential upper bound.

We next turn to bounds on $H_\lambda \psi_{\alpha,\lambda}$ which appear in the matrix elements of $M$. Recall we have, by definition of $A_{\lambda,\xi}^{-d,d}$, 
\eql{\label{eq:identity for H A}
H_\lambda A_{\lambda,\xi}^{-d,d} Q_{-d,d} = Q_{-d,d} + z_\lambda(\xi) A_{\lambda,\xi}^{-d,d} Q_{-d,d}\,.
} Here we have made explicit the $\xi\in\mathbb{S}^1$-dependence (through $z_\lambda(\xi)$) of $A_{\lambda,\xi}^{-d,d}$. We insert this identity into
\eq{
H_\lambda \psi_d = \frac{\ii}{2\pi}\oint_{\xi\in\mathbb{S}^1}H_\lambda A_{\lambda,\xi}^{-d,d}\Phi_d z_\lambda'(\xi)\dif{\xi}
} to obtain 
\eq{
H_\lambda \psi_d = \frac{\ii}{2\pi}\oint_{\xi\in\mathbb{S}^1}\br{\Phi_d + z_\lambda(\xi) A_{\lambda,\xi}^{-d,d} \Phi_d} z_\lambda'(\xi)\dif{\xi}\,. 
} The estimate \cref{eq:estimates on the two-well resolvent} yields now the desired bound
\eq{
\norm{H_\lambda \psi_d}_{L^2} \leq C \abs{\lambda}^{\sharp} \exp\br{C b \abs{\lambda}}\,.
}

The proof of the assertion in part 3 of \cref{thm:main theorem}, concerning averaged lower bounds on $\Delta(\lambda,b)$ for $\lambda>\Lambda$ and $0<b<b_\varepsilon$, is now complete. 
\section{Lower bounds on analytic functions}\label{sec:bounds on analytic functions}

Define 
\eq{
    \HH := \Set{z\in\CC:\Re{z}>0}
    \qquad\textrm{and}\qquad
    \DD := \Set{z\in\CC:\abs{z}<1}\,.
}

In this section, we prove the following result.

\begin{thm}\label{thm:sketch-half-plane-bound}
    Let $F:\HH\to\DD$ be analytic and suppose
    \eq{
        \abs{F(1)}\geq \exp\big(-\beta\big)\,.
    }
    Then, for $0<\delta<1/8$ we have
    \eql{\label{eq:sketch-main-bound}
        \frac{1}{\delta}\int_{t=\delta}^{2\delta}\mLog{F(t)}\dif{t}
        \leq
        C\frac{\beta}{\delta}
    }
    for an absolute constant $C$.
\end{thm}

\begin{proof}
    We write $c$, $C$, $C'$, etc. to denote absolute constants.
    These symbols may denote different constants in different occurrences.

    We use two textbook results from complex variables, namely
    the Blaschke factorization \cite[Theorem 15.21]{rudin1987real} or \cite[Section G.3]{Koosis_1988} and the Poisson integral \cite[Chapter 3 on the Herglotz representation]{SteinShakarchi-ComplexAnalysis-2003}.

    For $a=\alpha+\ii\beta\in\HH$ and $z\in\HH$, we define the Blaschke factor
    \eq{
        B_a(z)
        :=
        \frac{z-\alpha-\ii\beta}{z+\alpha-\ii\beta}\,\ee^{\ii\theta}\,,
    }
    with $\theta\in\RR$ picked so that $B_a(1)>0$.
 Then, any non-trivial analytic function $F:\HH\to\DD$ may be factorized as
    
    \eql{\label{eq:sketch-factorization}
        F(z)
        =
        \prod_{\nu} B_{a_\nu}(z)\cdot \exp\left(-G(z)\right)\,,
    }
    where $\Set{a_\nu}$ are the zeros of $F$ (multiplicities counted), and
    $G:\HH\to\HH$.

    Moreover, the non-negative harmonic function $u\equiv \Re{G}$ may be
    expressed as a Poisson integral\footnote{
        The term $At$ arises here when one carries over the result in \cite[Chapter 3 on the Herglotz representation]{SteinShakarchi-ComplexAnalysis-2003}
        from the unit disc to the half plane by a linear fractional
        transformation $\phi:\HH\to\DD$, in case the relevant measure on the
        unit circle has an atom at $\phi(\infty)$.
    }
    \eql{\label{eq:sketch-poisson}
        u(t+\ii s)
        =
        At+\frac{1}{\pi}\int_{-\infty}^{\infty}
        \frac{t}{t^2+(s-y)^2}\dif{\mu(y)}\,,
    }
    where $A\geq0$ is a constant and $\mu\geq0$ is a measure, with
    \eq{
        \int_{-\infty}^{\infty}\frac{\dif{\mu(y)}}{1+y^2}<\infty\,.
    }

    For a function $f(t)$, we write
    \eq{
        \mathrm{Av}_{[\delta,2\delta]}f
        :=
        \frac{1}{\delta}\int_{\delta}^{2\delta}f(t)\dif{t}\,.
    }

    We use the following elementary inequalities.
    Their proofs are left to the reader:
    \eql{\label{eq:sketch-small-a}
        \mLog{B_a(1)} > c\,\Re{a}
        \qquad\textrm{if }\abs{a}\leq1\,,
    }
    and
    \eql{\label{eq:sketch-large-a}
        \mLog{B_a(1)} \geq c\,\Re{\frac{1}{a}}
        \qquad\textrm{if }\abs{a}\geq1\,.
    }
    Also,
    \eql{\label{eq:sketch-average-Ba}
        \mathrm{Av}_{[\delta,2\delta]}\Bigl(\mLog{B_a(t)}\Bigr)
        \leq
        C\min\Set{\frac{\Re{a}}{\delta},\,\delta\,\Re{\frac{1}{a}}}\,.
    }
    for $a\in\HH$, $\delta>0$, and
    \eql{\label{eq:sketch-poisson-kernel}
        \frac{t}{t^2+y^2}
        \leq
        \frac{C}{t}\cdot\frac{1}{1+y^2}
        \qquad\textrm{for }0<t\leq\frac12,\ y\in\RR\,.
    }

    From \cref{eq:sketch-factorization,eq:sketch-poisson} we obtain
    \eql{\label{eq:sketch-decomposition}
        \mLog{F(t)}
        =
        \sum_{\nu}\mLog{B_{a_\nu}(t)}
        +
        At
        +
        \frac{1}{\pi}\int_{-\infty}^{\infty}\frac{t}{t^2+y^2}\dif{\mu(y)}\,.
    }

    In particular,
    \eql{\label{eq:sketch-at-one}
        \beta
        \geq
        \mLog{F(1)}
        =
        \sum_{\nu}\mLog{B_{a_\nu}(1)}
        +
        A
        +
        \frac{1}{\pi}\int_{-\infty}^{\infty}\frac{\dif{\mu(y)}}{1+y^2}\,.
    }

    All terms on the right hand side of \cref{eq:sketch-at-one} are non-negative, so
    \eql{\label{eq:sketch-sum-at-one}
        \sum_{\nu}\mLog{B_{a_\nu}(1)} \leq \beta\,,
    }
    \eql{\label{eq:sketch-A-bound}
        A\leq\beta\,,
    }
    and
    \eql{\label{eq:sketch-mu-bound}
        \frac{1}{\pi}\int_{-\infty}^{\infty}\frac{\dif{\mu(y)}}{1+y^2}\leq\beta\,.
    }

    Let
    \eq{
        S:=\Set{\nu:\abs{a_\nu}\leq1}\,,
        \qquad
        L:=\Set{\nu:\abs{a_\nu}>1}\,.
    }

    From \cref{eq:sketch-small-a}, \cref{eq:sketch-large-a}, and \cref{eq:sketch-sum-at-one}, we have
    \eql{\label{eq:sketch-zero-sums}
        \sum_{\nu\in S}\Re{a_\nu}\leq C\beta
        \qquad\textrm{and}\qquad
        \sum_{\nu\in L}\Re{\frac{1}{a_\nu}}\leq C\beta\,.
    }

    Now let $\delta<1/8$ be given.

    Then \cref{eq:sketch-average-Ba,eq:sketch-zero-sums} yield
    \eq{
        \sum_{\nu\in S}\mathrm{Av}_{[\delta,2\delta]}\Bigl(\mLog{B_{a_\nu}(t)}\Bigr)
        \leq
        \frac{C}{\delta}\sum_{\nu\in S}\Re{a_\nu}
        \leq
        \frac{C'}{\delta}\beta
    }
    and
    \eq{
        \sum_{\nu\in L}\mathrm{Av}_{[\delta,2\delta]}\Bigl(\mLog{B_{a_\nu}(t)}\Bigr)
        \leq
        C\delta\sum_{\nu\in L}\Re{\frac{1}{a_\nu}}
        \leq
        C'\delta\beta\,,
    }
    hence
    \eql{\label{eq:sketch-blaschke-average}
        \sum_{\nu}\mathrm{Av}_{[\delta,2\delta]}\Bigl(\mLog{B_{a_\nu}(t)}\Bigr)
        \leq
        C''\frac{\beta}{\delta}\,.
    }

    From \cref{eq:sketch-poisson-kernel,eq:sketch-mu-bound} we have
    \eq{
        \frac{1}{\pi}\int_{-\infty}^{\infty}\frac{t}{t^2+y^2}\dif{\mu(y)}
        \leq
        \frac{C}{t}\int_{-\infty}^{\infty}\frac{\dif{\mu(y)}}{1+y^2}
        \leq
        \frac{C'\beta}{t}\,,
        \qquad\textrm{for }0<t\leq\frac12\,,
    }
    hence
    \eql{\label{eq:sketch-poisson-average}
        \mathrm{Av}_{[\delta,2\delta]}\left(
            \frac{1}{\pi}\int_{-\infty}^{\infty}\frac{t}{t^2+y^2}\dif{\mu(y)}
        \right)
        \leq
        C''\frac{\beta}{\delta}\,.
    }

    Finally, by \cref{eq:sketch-A-bound},
    \eql{\label{eq:sketch-linear-average}
        \mathrm{Av}_{[\delta,2\delta]}(At)
        \leq
        CA\delta
        \leq
        C'\beta\delta\,.
    }

    Putting \cref{eq:sketch-blaschke-average,eq:sketch-poisson-average,eq:sketch-linear-average}
    into \cref{eq:sketch-decomposition}, we learn that
    \eq{
        \mathrm{Av}_{[\delta,2\delta]}\Bigl(\mLog{F(t)}\Bigr)
        \leq
        C\frac{\beta}{\delta}
        \qquad\textrm{for }\delta<1/8\,,
    }
    completing the proof of \cref{eq:sketch-main-bound}.
\end{proof}

\begin{rem}
    By working a little harder we can show also that
    \eq{
        \frac{1}{\delta}\int_{t=\delta}^{2\delta}
        \abs{\mLog{F(t)}-\frac{\gamma}{t}}\dif{t}
        =
        o\left(\frac{1}{\delta}\right)
        \qquad\textrm{as }\delta\to0^+\,,
    }
    with
    \eq{
        \gamma=\frac{\mu(\Set{0})}{\pi}\,.
    }
    We omit the proof, since we won't use the result here.
\end{rem}

We will apply \cref{thm:sketch-half-plane-bound} to the function
\eq{
    \frac{F\bigl(z^{\pi/(2\alpha)}\bigr)}{U(z)}\,,
}
where $F:\Gamma_\alpha\to\HH$ is analytic,
\eq{
    \Gamma_\alpha:=\Set{z\in\CC:\abs{\operatorname{arg}(z)}<\alpha}
    \qquad
    \left(\alpha<\frac{\pi}{2}\right),
}
and $U(z)$ is analytic and everywhere nonzero on $\Gamma_\alpha$.

Then, we have

\begin{cor}\label{cor:main complex analytic lower bound}
    Suppose $F:\Gamma_\alpha\to\CC$ is analytic, and
    \eq{
        \abs{F(z)}\leq \abs{U(z)}
        \qquad\textrm{on }\Gamma_\alpha\,.
    }
    If
    \eq{
        \abs{F(1)}>\exp\big(-\beta\big) \abs{U(1)}\,,
    }
    then
    \eq{
        \frac{1}{R}\int_{t=R}^{2R}\mLog{F(t)}\dif{t}
        \leq
        C\beta R^{\pi/(2\alpha)}
        +
        \frac{1}{R}\int_{t=R}^{2R}\mLog{U(t)}\dif{t}
    }
    for $R>8$.
\end{cor}
\section{Mesoscopic annuli; construction of $T_{\lambda,\nu}$ and completion of the proof of \cref{lem:continuation of cut off one well resolvent}}
\label{sec:mesoscopic annuli}

We refer to the setup introduced in  \cref{subsubsec:extending h theta}.  Our goal here is to construct the operators $T_{\lambda,\nu}$ with the properties detailed in \cref{sec:the_proof}, establish bounds on their norms as well as on how well they approximate inverses of $h_\lambda^\theta-z_\lambda\Id$. 
We use standard techniques from the theory of pseudo-differential operators; see \cref{sec:pseudo diff operators}.

For each $\nu=0,\cdots,\lceil\log_2(R/\delta)\rceil$, we seek an approximate inverse of $h_\lambda^\theta-z_\lambda\Id$ on functions supported in the annulus:
\eql{\label{eq:Unu-def}
U_\nu:=\Set{x\in\RR^2 | 2^{\nu-3} \delta < \norm{x}< 2^{\nu+3}\delta};
} 
see \cref{eq:U_nu-def}.
Here, $R$ is a parameter 
of order one, chosen so that $B_R(0)$ contains the support of the cutoff $\theta$; see \cref{subsubsec:extending h theta}. Recall that these open sets are chosen so that
\eq{
\supp(\psi_\nu)\subset \Set{x | 2^{\nu-2}\delta \leq \norm{x}\leq 2^{\nu+2}\delta}\subset U_\nu.} 

Restricted to the annulus $U_\nu$, the  symbol $s_\lambda^\theta:U_\nu\times\RR^2\to \CC$  of $h_\lambda^\theta-z_\lambda\Id$ is given by:
\eql{\label{eq:symbol in mesoscopic annulus}
(x,p)&\mapsto s_\lambda^\theta(x,p) \notag\\
 &:= \br{p-\frac12 b \lambda \theta(x)  x^\perp}^2 +\frac\ii2b \lambda\br{\nabla \theta}( x)\cdot x^\perp+ \lambda^2 \br{v(x)+1}-\mu\lambda \notag \\
%&=\norm{p}^2+\frac14 b^2\lambda^2\theta(\alpha x)^2\norm{x}^2-b\lambda\theta(x)p\cdot x^\perp  +\frac\ii2b\lambda \br{\nabla \theta}(x)\cdot x^\perp \nonumber\\
%&\qquad+\lambda^2 \br{v(x)+1}-\mu\lambda.
} where we recall that $z_\lambda=z_\lambda(\xi)=-\lambda^2 + \mu(\xi)\lambda$ with \added{$\mu(\xi):=e_0^{\rm MHO}+c_{\mathrm{ctr}}(e_1^{\rm MHO}-e_0^{\rm MHO})\xi$}; in particular, $|\mu(\xi)|\le C$ uniformly in $\xi\in\mathbb S^1$; see \cref{eq:spectral parameter for Riesz}.  We have that $\Op{s^\theta_\lambda} = h_\lambda^\theta-z_\lambda\Id$ when restricted to $L^2(U_\nu)$.

Our goal is to apply the local elliptic \cref{lem:local elliptic lemma} to invert $\Op{s_\lambda^\theta}$ on functions supported in   $U_\nu$. 

% For instance, one way to phrase the lemma we have in mind is as
% \begin{lem}[Local elliptic lemma]\label{lem:local elliptic lemma} Let $U,V\subseteq\RR^n$ be two open subset with $V\subseteq U$ relatively compact. Let $\calL$ be elliptic in $S^2_M(U)$. Let $\chi\in C^\infty(\RR^n)$ be supported in $V$. Then there exist linear operators $A:L^2(\RR^n)\to H^2(\RR^n)$ and $\calE:L^2(\RR^n)\to H^1(\RR^n)$ with the following properties:
%     \begin{enumerate}
%         \item $\calL A  -\calE= \chi$ as operators on $L^2(\RR^n)$. We interpret $A$ as a "local" resolvent of $\calL$ within $V$ up to the error $\calE$.
%         \item $A,\calE$ map $L^2(\RR^n)$ into $L^2(V)$, i.e., $\supp Af, \supp \calE f\subseteq V$ for all $f\in L^2(\RR^n)$.
%         \item $\norm{\partial^\alpha A f}_{L^2(\RR^n)} \leq C M^{|\alpha|-2}\norm{f}_{L^2(\RR^n)}$ for $|\alpha|\leq 2,f\in L^2(\RR^n)$.
%         \item $\norm{\partial^\alpha \calE f}_{L^2(\RR^n)}\leq C M^{|\alpha|-1}\norm{f}_{L^2(\RR^n)}$ for $|\alpha|\leq 1,f\in L^2(\RR^n)$. 
%     \end{enumerate} Here, the constant $C$ may be taken to depend only on the elliptic $S^2_M(U)$ constants of $\calL$ and the $C^\infty$ seminorms of $\chi$. Moreover, if the coefficients of $\calL$ depend analytically on a parameter $\lambda\in\CC$, then the operators $A$ and $\calE$ may be taken to depend analytically on $\lambda$.
%     \end{lem} Since it is standard in the theory of pseudo-differential operators, we omit giving context or proof.

Now both the annuli, $U_\nu$, and the  symbol $s_\lambda^\theta$ depend on the large asymptotic parameter $\abs{\lambda}$. 
We next rescale the spatial variable, $x$, and map our problem to a fixed ($\delta$ independent) spatial domain.
In these new variables, all parameters appear explicitly in the scaled symbol.
%It is natural to \emph{rescale} the spatial variable, $x$,  and map our problem to the setting of a fixed annulus with radus of order $1$, and a symbol containing all dependencies of parameters, {\it i.e.}  $\Lambda$, $b$, to be tuned. 
 To this end, for $x\in U_\nu$ we set 
\eq{ x':=\alpha^{-1}x,\quad {\rm where}\quad  \alpha := \alpha_\nu = 2^\nu \delta\,\quad \ \textrm{($\delta(\Lambda)=\Lambda^{\eta-\frac12}$)}\ ,}
and so $x'$ varies over the fixed annulus:
\begin{equation}
\widetilde{U} := \frac{1}{\alpha} U_\nu = \Set{x'\ :\ \frac18<\norm{x'}<8}.\label{eq:Udef}
\end{equation}
Since $x=\alpha x'$, we \deleted{have one }\added{have} $P = \alpha^{-1}P'=
\alpha^{-1} \ii^{-1}\partial_{x'}$ and so 
$s_\lambda^\theta(x,p)=s_\lambda^\theta(\alpha x',\alpha^{-1}p')$, with $x'\in \tilde{U}$ and $p'\in\RR^2$.
It is convenient to work with a symbol which is also multiplicatively scaled as follows:
\eq{ \widetilde{s_\lambda^\theta}(x',p') := \alpha^2 s_\lambda^\theta\big(\alpha x',\frac{1}{\alpha}p'\big),\quad 
 x'\in \tilde{U},\ \ p'\in\RR^2.
}
Let $\mathfrak{U}_\alpha: L^2(\RR^2)\to L^2(\RR^2)$ denote the unitary scaling: 
$\mathfrak{U}_\alpha[f](x)=\alpha f(\alpha x)$. We note, for later use, by \cref{eq:scalings of PDO} of \cref{sec:dilations}:
\eql{\label{eq:s-scaled}
\mathfrak{U}_\alpha^\ast\Op{\alpha^{-2} \widetilde{s_\lambda^\theta}}\mathfrak{U}_\alpha = \Op{s_\lambda^\theta}\ =\ h_\lambda^\theta-z_\lambda\Id,\quad \textrm{for $x\in U_\nu$.}
}

\noindent{\it For the remainder of this section, for notational convenience, we drop all primes from the independent variables notation and write the symbol as:} \eq{ \tilde U\times\RR^2 \ni (x,p)\mapsto \widetilde{s_\lambda^\theta} (x,p)\in \CC
}
where
\eql{ \widetilde{s_\lambda^\theta}(x,p)& = \br{p-\frac12 b \alpha^2 \lambda \theta(\alpha x) x^\perp}^2 +\frac\ii2b\alpha^3\lambda\br{\nabla \theta}(\alpha x)\cdot x^\perp\ +\ \alpha^4\lambda^2 \br{\frac{v(\alpha x)+1}{\alpha^2}}-\mu\alpha^2\lambda\notag\\
&=\norm{p}^2+\frac14 b^2\lambda^2\alpha^4\ \theta(\alpha x)^2\norm{x}^2-b\lambda\alpha^2\ \theta(\alpha x)p\cdot x^\perp  +\frac\ii2b\lambda \alpha^3\ \br{\nabla \theta}(\alpha x)\cdot x^\perp \nonumber\\
&\qquad+\lambda^2 \alpha^4 \ w(\alpha x)\ -\ \mu\lambda\alpha^2\ ,
\qquad (x\in \widetilde{U},\ p\in\RR^2)
\label{eq:scaled-symbol}} 
and 
\eql{
w(\alpha x) := \frac{v(\alpha x)+1}{\alpha^2}\, ,\ \ \textrm{\ where}\ \ \alpha=\alpha_\nu := 2^\nu\delta = 2^\nu \Lambda^{\eta-1/2},\ 
}
where, $\nu=1,\cdots,\lceil\log_2(R/\delta)\rceil$. Moreover, we let $M_\nu:=\alpha_\nu^2\Lambda=(2^\nu\Lambda^\eta)^2\gg1$.

Note that  $w$ depends on $\alpha_\nu$. Further,    $\widetilde{s_\lambda^\theta}$ (the rescaling of  $s_\lambda^\theta$ on $U_\nu$), which is now defined on the fixed ($\alpha_\nu$-independent)  domain  $\widetilde{U}$, see \cref{eq:Udef},  also depends explicitly on  $\alpha_\nu$. 
It is convenient to suppress the dependence of $\widetilde{s_\lambda^\theta}$ on $\nu$.

If we can prove, for each $\nu$, that the symbol $\widetilde{s_\lambda^\theta}$ is elliptic in $\widetilde{U}$, we can then apply \cref{lem:local elliptic lemma} 
 to construct 
 an inverse of $\Op{\widetilde{s_\lambda^\theta}}$ on $\widetilde{U}$. And then, by undoing the unitary rescaling, $\mathfrak{U}_{\alpha_\nu}$, we'll have the existence of a local inverse and corresponding bounds for $ h^\theta_\lambda-z_\lambda\Id $ on $U_\nu$. 

\replaced{Ellipicity}{Ellipticity} boils down to establishing the following upper and lower bounds on  the symbol $\widetilde{s_\lambda^\theta}$; For all 
$x\in \widetilde{U}$,  $p\in\RR^2$ and $\alpha,\beta\in\NN$:
\begin{align}
\label{eq:upper bound on symbol}
    \abs{\partial_x^\alpha \partial_p^\beta \widetilde{s_\lambda^\theta}(x,p)} &\leq {\rm const.}\br{M_\nu+ \norm{p}}^{2-\abs{\beta}}\\
    \label{eq:ellipticity lower bound}
\abs{\widetilde{s_\lambda^\theta}(x,p)} &\geq {\rm const.}\br{M_\nu+ \norm{p}}^{2}\ .
\end{align}

The idea is that if  \cref{eq:upper bound on symbol} and \cref{eq:ellipticity lower bound}  hold then, 
acting on functions supported in $\widetilde{U}$, 
$\Op{\widetilde{s_\lambda^\theta}}$ is approximately invertible with local resolvent  bounds which scale  like those of $P^2+M_\nu^2\Id\equiv-\Delta + M_\nu^2\Id $.
  The proof of \cref{eq:upper bound on symbol} is straightforward. Hence, we focus on \cref{eq:ellipticity lower bound}. 

We begin with a basic estimate about the size of $v$. 
 By the hypotheses of \cref{thm:main theorem}, $v\in C^3$ with a unique minimum at $x=0$, $v(0)=-1$, and $\Hessian{v}(0)>0$. Therefore, there are constants $c_v<C_v$ such that \eq{
c_v (2^{\nu}\delta)^2 \leq v(x)+1 \leq C_v (2^{\nu}\delta)^2, \qquad x\in U_\nu\ .
} 
 %In principle its size could vary both radially and azimuthally within $U_\nu$, but we are only interested in bounding the radial variations which cannot vanish or rise too quickly. Within $U_\nu$, since $v$ is smooth with a unique non-degenerate minimum at $0$, we know that
 Hence, 
\eq{
c_v \leq w(\alpha x) \leq  C_v \qquad,\qquad x\in  \widetilde{U}\ .
} 
%
%Let $M_\nu := \alpha^2 \Lambda $. According to the choice  $\delta=\Lambda^{\eta-1/2}$ in \cref{eq:the choice of delta}, this implies that for all $\nu=1,\cdots,\lceil\log_2(R/\delta)\rceil$, we have $M_\nu = 2^{2\nu} \Lambda^{2\eta}\gg1$ .
%
 Towards a proof of \cref{eq:ellipticity lower bound}, we first seek a lower bound on $\abs{\norm{p}^2+\alpha^4\lambda^2 w(\alpha x)}$. Note, using $|\lambda|\le2\Lambda$, that  $\alpha^4|\lambda|^2\le (2^\nu\Lambda^\eta)^4 \Lambda^{-2} (2\Lambda)^2 \le 4 (2^\nu\Lambda^\eta)^4$ or 
\eql{\label{eq:Mnu-def} \alpha^4|\lambda|^2\le M_\nu^2(\Lambda),\quad \textrm{where $M_\nu^2(\Lambda):=4 (2^\nu\Lambda^\eta)^4$} \ .} 
%
%\footnote{
%\eql{
%\abs{\norm{p}^2+\alpha^4\lambda^2 w(\alpha x)} \geq \frac{c_v \sin(2\ve)}{1+2C_v} \br{\norm{p}^2+M_\nu^2}\, ,\textrm{where}\quad M_\nu := 2^{2\nu} \Lambda^{2\eta}\gg1.}
%}
%

Then,
%\footnote{previous: if $\norm{p}^2\geq %2C_v M_\nu^2$, }
\eq{
\abs{\norm{p}^2+\alpha^4\lambda^2 w(\alpha x)} \geq \norm{p}^2/2 + \norm{p}^2/2 -  C_v M_\nu^2 \geq \frac{1}{2}\br{\norm{p}^2+M_\nu^2},\ \  
\textrm{if $\norm{p}^2\geq \addedm{C_1 M_\nu^2}$.}
} 
Next, consider the case where $\norm{p}^2<\addedm{C_1M_\nu^2}$. Here, we write $\lambda=|\lambda|\exp(\ii\arg(\lambda))$ where 
$\abs{\arg(\lambda)}<\pi/2-\ve$, and  find, using that $|\lambda|\ge\Lambda$, that
\eq{
\abs{\norm{p}^2+\alpha^4\lambda^2 w(\alpha x)} \geq |\sin(2\ve)|\ \alpha^4\abs{\lambda}^2 w(\alpha x)\ge 
c_v |\sin(2\ve)|\ (2^{\nu}\Lambda^{\eta})^4 = 
\addedm{c_{\ve,v}M_\nu^2}.
}
Writing $M_\nu^2 = M_\nu^2/2 + M_\nu^2/2$ and applying the assumed upper bound on $\|p\|^2$ we find
\eq{
\abs{\norm{p}^2+\alpha^4\lambda^2 w(\alpha x)} \geq 
\addedm{c_{\ve,v}\big(\|p\|^2+M_\nu^2\big)},\quad \textrm{if $\norm{p}^2\le \addedm{C_1 M_\nu^2}$.}
}
Hence, for $\varepsilon>0$ and small we have that 
\begin{equation}
\abs{\norm{p}^2+\alpha^4\lambda^2 w(\alpha x)} \geq 
\addedm{c_{\ve,v}\big(\|p\|^2+M_\nu^2\big)},\quad x\in\widetilde{U},\ p\in\RR^2.
\label{eq:s-dominant}\end{equation}
We next use \cref{eq:s-dominant} to bound the full symbol
 $\widetilde{s_\lambda^\theta}$ from below.
%
%\eq{
%\abs{\norm{p}^2+\alpha^4\lambda^2 w(\alpha x)} \geq \norm{p}^2 -  C_v M_\nu^2 \geq \frac12\norm{p}^2\geq\frac{C_v}{1+2C_v}\br{\norm{p}^2+M_\nu^2}
%} 
We have 
\eql{\label{eq:final ellipticity estimate}
&\abs{\widetilde{s_\lambda^\theta}(x,p)} \geq \abs{\norm{p}^2+\alpha^4\lambda^2 w(\alpha x)}\nonumber \\
&\qquad\qquad  -\frac14 |b|^2 \alpha^4\abs{\lambda}^2 \norm{x}^2 - |b| \alpha^2\abs{\lambda} \norm{x}\norm{p}  -\abs{\mu}\alpha^2\abs{\lambda} -\frac12|b|\abs{\lambda}\alpha^3\norm{\nabla\theta(\alpha x)}\norm{x} \nonumber\\
&\geq \addedm{c_{\ve,v}\br{\norm{p}^2+M_\nu^2}}-16 |b|^2 M_\nu^2- 4 |b| \br{M_\nu^2+\norm{p}^2}-  M_\nu-4 |b| M_\nu \alpha \norm{\nabla \theta}_\infty, \nonumber
} 
where we have used that $|\mu|\le C$ uniformly in $\xi$. 

In this lower bound, all subtracted terms, except  for the term   $-M_\nu$, can be bounded from above by  $|b|(M_\nu^2+\|p\|^2)$, and can therefore be absorbed in the leading positive term by taking $|b|$ sufficiently small. To obtain a lower bound on the term $-M_\nu$, note that 
$ \varepsilon M_\nu^2\gtrsim M_\nu \ \iff\ 
 \varepsilon M_\nu\gtrsim 1\iff \varepsilon (2^\nu\Lambda^\eta)^2\gtrsim 1$,  which can be satisfied by taking $\Lambda \gtrsim \varepsilon^{-1/(2\eta)}$. 
Hence, if we let $\varepsilon>0$ be arbitrary and small, and take $\Lambda_\varepsilon>\varepsilon^{-1/(2\eta)}$, then for a small positive constant, $b_\varepsilon$,  the following holds:
 \eql{
&\textrm{For all 
 $b\in(0,b_\varepsilon)$, $\Lambda>\Lambda_\varepsilon$,  and $\lambda\in\Omega_{\Lambda, 2\Lambda,\varepsilon}$\ ,}\nonumber \\
 & \abs{\widetilde{s_\lambda^\theta}(x,p)} \gtrsim 
\sin(2\ve)\br{\norm{p}^2+M_\nu^2}\gtrsim \sin(2\ve)(\|p\|+M_\nu)^2.
 }
 Here, we have used that $\|p\|^2+M_\nu^2\ge (\|p\|+M_\nu)^2/2$.

Summarizing, we have proved that for each $\nu=1,\dots, N_\Lambda-1$ 
\eq{
\widetilde{U}\times\RR^2\ni(x,p) \mapsto \widetilde{s_\lambda^\theta}(x,p)\in\CC } 
is elliptic on $\widetilde{U}$. By \cref{lem:local elliptic lemma}, if $\widetilde{\chi}_\nu := \mathfrak{U}_\alpha \chi_\nu \mathfrak{U}_\alpha^\ast$, whose  support \added{lies} in $\widetilde{V}$, a relatively compact \added{subset of} $\widetilde{U}$, then there exists $\widetilde{T}_{\lambda,\nu}\in\calB(L^2(\widetilde U))$, such that  
\eql{
\norm{\Op{\widetilde{s_\lambda^\theta}}\widetilde{T}_{\lambda,\nu}-\widetilde{\chi}_\nu}_{\calB(L^2(\widetilde U))} \leq \frac{C}{M_\nu}\quad\textrm{and}\quad 
\norm{\widetilde{T}_{\lambda,\nu}}_{\calB(L^2(\widetilde U))} \leq \frac{C}{M_\nu^2}\,.
} 
Recall now from \cref{eq:s-scaled} that $\mathfrak{U}_\alpha^\ast\Op{ \widetilde{s_\lambda^\theta}}\mathfrak{U}_\alpha = \alpha^{2}\Op{s_\lambda^\theta}$. Letting 
\eql{\label{eq:Tlamnu-def} T_{\lambda,\nu} :=\alpha^2\mathfrak{U}_\alpha^\ast\widetilde{T}_{\lambda,\nu}\mathfrak{U}_\alpha\in\calB(L^2(U_\nu)),
}
 we have 
\begin{align}
\Op{\widetilde{s_\lambda^\theta}}\widetilde{T}_{\lambda,\nu}-\widetilde{\chi}_\nu &= 
 \alpha^2 \mathfrak{U}_\alpha\Op{s_\lambda^\theta}\mathfrak{U}_\alpha^\ast\widetilde{T}_{\lambda,\nu} - \mathfrak{U}_\alpha \chi_\nu \mathfrak{U}_\alpha^\ast \nonumber\\
 &= \mathfrak{U}_\alpha\Big[
\Op{s_\lambda^\theta}\ T_{\lambda,\nu} - \chi_\nu   \Big]\mathfrak{U}_\alpha^\ast
\label{eq:Op-diff}\end{align}
 Taking operator norms of \cref{eq:Tlamnu-def} and \cref{eq:Op-diff} and using the unitarity of $\mathfrak{U}_\alpha$,  we have
\begin{equation}
\norm{T_{\lambda,\nu}}_{\calB(L^2(U_\nu))}
 = \alpha^2 \norm{\widetilde{T}_{\lambda,\nu}}_{\calB(L^2(\tilde{U}))} \leq \frac{\alpha^2C}{ M_\nu^2}\lesssim \Lambda^{-2\eta-1}\,.
\end{equation}
and \eql{
&\norm{\Op{s_\lambda^\theta}\ T_{\lambda,\nu} - \chi_\nu  }_{\calB(L^2(U_\nu))} =
\norm{\Op{\widetilde{s_\lambda^\theta}}\widetilde{T}_{\lambda,\nu}-\widetilde{\chi}_\nu}_{\calB(L^2(\tilde{U}))} 
\leq \frac{C}{M_\nu} \lesssim \Lambda^{-2\eta}
\nonumber \\
&{\ }
}
In other words,
\eql{
\norm{\br{h_\lambda^\theta-z_\lambda\Id} T_{\lambda,\nu}-\chi_\nu}_{\calB(L^2(U_\nu))} \lesssim \Lambda^{-2\eta}
}
The proof of \cref{lem:continuation of cut off one well resolvent} is now complete.

   %  \subsection{Dilations}\label{sec:dilations}
   %  For any $\alpha>0$ and $S\subseteq\RR^n$, we write \eq{\alpha S \equiv \Set{\alpha x\in \RR^n | x\in S}\,.} We then have a unitary dilation operator $\mathfrak{U}_\alpha$ on $L^2(\RR^n)$ given by
   %  \eql{\label{eq:dilation operators}
   %  \br{\mathfrak{U}_\alpha \psi}(x) \equiv \alpha^{n/2}\psi(\alpha x)\qquad(x\in\RR^n;\psi\in L^2(\RR^n))
   %  } and locally $\mathfrak{U}_\alpha : L^2(\alpha S)\to L^2(S)$ by the same formula. Then, 
   %  \eq{
   %  \mathfrak{U}_\alpha^\ast X \mathfrak{U}_\alpha = \frac{1}{\alpha} X\quad {\rm and}\quad 
   %    \mathfrak{U}_\alpha^\ast P \mathfrak{U}_\alpha = \alpha P\,.}
   %    Further, let $a \in S^m_M(U)$. Then, with \eq{
   %    \alpha U\times\RR^n\ni(x,p)\mapsto\widetilde{a}_\alpha(x,p) := a(x/\alpha,\alpha p)
   %    } we have $\widetilde{a}_\alpha\in S^m_M(\alpha U)$ and 
   % \eql{\label{eq:scalings of PDO}
   % \mathfrak{U}_\alpha^\ast \Op{ a}\mathfrak{U}_\alpha=\Op{\widetilde{a}_\alpha}
   % } \js{possibly remove this} as well as
   % \eq{
   %  C_{M,\alpha\beta}^m\br{\widetilde{a}_\alpha} \leq \alpha^{|\beta|-|\alpha|}\max\br{\Set{1,\alpha}}^{m-|\beta|}C_{M,\alpha\beta}^m\br{a}\,.
   % }

\appendix

    \section{Review of pseudodifferential operators}\label{sec:pseudo diff operators}
    Here we recount some facts about pseudodifferential operators. These will mainly be used in order to establish analyticity and bounds of resolvents. In this section,  $X$ is the position operator on $L^2$ and $P\equiv-\ii\nabla$ is the momentum operator on it. With lower case these symbols denote the corresponding real variables: by our convention, under the Fourier transform, $P$ becomes a multiplication by the function $p\mapsto p$. Also, for the sake of clarity here the space dimension is $n\in\NN$ although in the rest of the paper we make use of it with $n=2$.

    Let $M\geq1$, $m\in\ZZ$ and $U\in\Open{\RR^n}$. We study \emph{phase space symbols} $a\in C^\infty(U\times\RR^n\to\CC)$ and their regularity classes. To that end, for any multi-indices $\alpha,\beta\in\NN_{\geq0}^n$, we define a seminorm \eq{
        C^m_{M,\alpha\beta}(a) := \sup_{x\in U,p\in\RR^n}\frac{\left|\partial_x^\alpha\partial_p^\beta a(x,p)\right|}{\br{M+\norm{p}}^{m-|\beta|}}
    } and define the  symbol class $S^m_M$  using these seminorms \eql{\label{eq:def of S^m_M}
        S^m_M(U) := \Set{a\in C^\infty(U\times \RR^n_p) | C^m_{M,\alpha\beta}(a)<\infty\qquad\forall\alpha,\beta\in\br{\NN_{\geq0}}^n}\, . 
    } 
     This family of seminorms is separating, so this induces a metric on the space of symbols $S^m_M(U)$ which makes it into a complete locally convex metric vector space. We write $S^m_M\equiv S^m_M(\RR^n)$.

    To each phase space symbol we associate a \emph{pseudodifferential operator} $\Op{a}$, its (Kohn–Nirenberg) quantization. Even if the $x$ variable of the symbol $a$ is only defined on $U$, we let $\Op{a}$ be an operator on $L^2(\RR^n)$ for convenience. That operator acts on (a dense subspace of) the Hilbert space $L^2(\RR^n)$), i.e., we construct a linear map \eq{
    \operatorname{Op}:S^m_M(U) \to \calL(C^\infty_c(\RR^n)\to C^\infty(U))
    } and eventually extend the domain of the operator appropriately. Its action on $u\in C^\infty_c(\RR^n)$ is prescribed as \eql{\label{eq:action of pseudo-diff op}
        \br{\Op{a} u}(x) := \frac{1}{\br{2\pi}^n}\int_{p\in\RR^n}\ee^{\ii p\cdot x}a(x,p) \widehat{u}(p)\dif{p}\qquad(x\in U)
    } where $\widehat{u}$ is the Fourier transform of u:
    \eql{
        \widehat{ u}(p) \equiv \int_{x\in \RR^n}\ee^{-\ii p\cdot x}u(x)\dif{x}
    } and formally the integral kernel of $\Op{a}$ is given by \eql{
        U\times \RR^n \ni (x,y) \mapsto  \Op{a}(x,y) = \frac{1}{\br{2\pi}^n}\int_{p\in\RR^n}\ee^{\ii \br{x-y}\cdot p}a(x,p)\dif{p}\,.
        }

    %Each operator $\Op{a}$, once extended to a dense subset of $L^2(U)$; the natural choice is $H^m(U)$. It may then be further extended to a dense subset of $L^2(\RR^n)$ by writing $L^2(\RR^n)=L^2(U)\oplus L^2(U^c)$ and letting $\Op{a}$ act as zero on the second summand. Hence we can let $\Op{a}$ act on $L^2(\RR^n)$ functions as soon as they are compactly supported within $U$.
        
    % } Compare this with the integral kernel of an operator specified by the closely related \emph{Weyl symbol} $(x,p)\mapsto a_W(x,p)$, 
    % \footnote{\miw{if we're not working with the weyl  calculus perhaps omit \cref{eq:weyl},  ?}\js{I think it is nice to contrast the two definitions in case anyone would wonder about the relation.}}\eql{\label{eq:weyl}
    %     \RR^n\times\RR^n\ni (x,y) \mapsto \frac{1}{\br{2\pi}^n}\int_{p\in\RR^n}\ee^{\ii \br{x-y}\cdot p}a_W\br{\frac{x+y}{2},p}\dif{p}\,.
    % }

    Our main interest in pseudodifferential operators is through the fact that for $a\in S^2_M(U)$ with $\frac{1}{a}\in S^{-2}_M(U)$ we think of $\Op{a}$ as an operator which to some extent behaves very much like the free Laplacian at spectral parameter $M^2$, i.e., $$-\Delta+M^2\Id$$ and, formally restricted to act in $L^2(U)$. This is made precise below in \cref{lem:local elliptic lemma}.

    In common abuse of notation, if the risk of confusion is low, when dealing with an explicit formula $f(x,p)$ for a symbol $f$, we shall write
    \eq{
    \Op{f(x,p)}
    } instead of the more cumbersome 
    \eq{
    \Op{\br{x,p}\mapsto f(x,p)}
    } and similarly with the seminorms. 

A few facts may be readily verified:
\begin{lem}\label{lem:basic properties of PDO}    
We have    
    \begin{enumerate}
        \item If  $b\in C^\infty(U)$ then $\Op{b(x)} = b(X)$ where $X$ is the position operator.
        \item If  $m\in\NN_{\geq0}$ and $b_\alpha:U\to\CC$ are smooth functions  then \eq{
            \Op{\sum_{|\alpha|\leq m }b_\alpha(x) p^\alpha} = \sum_{|\alpha|\leq m} b_\alpha(X) P^\alpha
        } where $P_j\equiv-\ii\partial_{x_j}$ is the momentum operator in the $j$th direction, $j=1,\dots,n$.
        \item If $b\in\CC^\infty(U)$ then \eql{\label{eq:degree of product of two symbols where one is diagonal}
            b(X) \Op{a} = \Op{b(x)a(x,p)}\,.
        } 
        \item If $a\in S^m_M$ and $b\in S^{m'}_M$ then $ba\in S^{m+m'}_M$. Moreover, we have \eq{
        C^{m+m'}_{M,\alpha\beta}(ba) \leq \sum_{\alpha'\leq\alpha,\beta'\leq\beta}\begin{pmatrix}
            \alpha \\ \alpha'
        \end{pmatrix}\begin{pmatrix}
            \beta \\ \beta'
        \end{pmatrix}C^{m'}_{M,\alpha'\beta'}(b) C^{m}_{M,\alpha-\alpha',\beta-\beta'}(a)\,.
        }
        \item If $a\in S^m_M$ then \eql{\label{eq:momentum acting on PDO}
        P_j \Op{a} = \Op{ p_j a(x,p) -\ii(\partial_{x_j}a)(x,p)}
        } with the first term in $S^{m+1}_M$ and the second term in $S^m_M$. We have \eq{
        C^{m+1}_{M,\alpha\beta}( p_j a(x,p)) \leq C(m,n)\sup_{\beta'\leq\beta}C^{m}_{M,\alpha,\beta-\beta'}(a)
        } and \eq{
        C^{m}_{M,\alpha\beta}(-\ii \partial_{x_j}a) = C^{m}_{M,\alpha+e_j,\beta}(a)\,.
        }
        \item Repeatedly applying the above, if $l$ is a symbol of a \emph{differential} operator such that $l\in S^{m'}_M$ and $a\in S^m_M$ then the product of the associated pseudodifferential operators may be decomposed according to the order as follows: \eql{\label{eq:decomposition of product of two pdos into highest degree and remainder}
        \Op{l} \Op{a} =  \Op{la} + \Op{\ve}
        } where $\ve\in S^{m+m'-1}_M$. Moreover, the $S^{m+m'-1}_M$-seminorms of $\ve$ are controlled by the $S^m_M$  seminorms of $a$ and the $S^{m'}_M$ seminorms of $l$. %Note that while this is true as stated, it only follows from the above in the case that $\Op{l}$ is a differential, rather than pseudodifferential operator. This is the only case we shall need.
        \item If $a\in S^m_M$ depends on analytically on a parameter $\lambda\in\Omega\subseteq\CC$ for some open $\Omega$, and $\Omega\ni\lambda\mapsto C^{m}_{M,\alpha\beta}(a)$ are uniformly bounded, then $\Omega\ni\lambda\mapsto \Op{a}$ is an analytic function.
    \end{enumerate}
\end{lem}
    The proof of these facts is standard and is thus omitted.

    A standard property of pseudodifferential operators is that
%     \begin{prop}
%         Let $a\in S^0_M$ and $u\in L^2(\RR^n)$. Then \eql{
%         \norm{\Op{a} u}_{L^2(\RR^n)} \leq \br{C_{0,0}(a)} \norm{u}_{L^2(\RR^n)}\, ,
%         } 
%     \end{prop}
% \miw{ where the constant $C_{0,0}(a)$ is independent of $M$. $(M\ge1?)$}
%     \begin{proof}
%         \js{Is the main complication here the case when $u$ is in $L^2$ but not a Schwarz function? In the latter case this is immediate, no?}
%     \end{proof}

\begin{prop}\label{prop:L2 boundedness of zero order symbols}
Let $a\in S^0_M(U)$. Let $V\subseteq U$ be open and relatively compact. Then there exists $N\in\NN$ with $2N>n$ and a constant $C_{n,N}(U,V)<\infty$ such that
\[
\norm{\Op{a}}_{\calB(L^2(\RR^n)\to L^2(V))}
\;\le\;
C_{n,N}(U,V)\,\Big(\sum_{|\alpha|,|\beta|\le 2N} C^{\,0}_{M,\alpha\beta}(a)\Big).
\]
\end{prop}
\begin{proof}
    Let $\chi\in C^\infty_c(U\to[0,1])$ with $\chi=1$ on $V$.

    We then apply \cite[Theorem 1 on pp. 234]{Stein1993Harmonic} on the symbol $\RR^n\times\RR^n\ni (x,p)\mapsto a(x,p)\chi(x)$.
\end{proof}

     \begin{cor}\label{cor:regularity of symbols of negative degree}
        Let $a\in S^{-m}_M(U)$ for $m\geq0$ and $V\subseteq U$ be open and relatively compact. Then there exists $N\in\NN$ with $2N>n$ and a constant $C_{n,N}(U,V)$ such that\eql{\nonumber
        \norm{P^\alpha \Op{a}}_{\calB(L^2(\RR^n)\to L^2(V))} \leq C_{n,N}(U,V)\sum_{|\tilde{\alpha}|,|\beta|\leq 2N}C_{M,\tilde{\alpha}\beta}^0((x,p)\mapsto M^{m-|\alpha|}p^\alpha a) M^{|\alpha|-m}\nonumber\\ \qquad(\alpha\in\ZZ^n:|\alpha|\leq m)\,.\nonumber\\
        } 
        \end{cor}
\begin{proof}
Fix $\alpha\in\NN^n$ with $|\alpha|\le m$. By iterating the commutation rule
\[
P_j \Op{a} = \Op{p_j a(x,p)} - \ii\,\Op{(\partial_{x_j}a)(x,p)} \tag{\ref{eq:momentum acting on PDO}}
\]
we obtain an expansion
\begin{equation}\label{eq:Palpha-expansion}
P^\alpha \Op{a}
\;=\;
\Op{p^\alpha a(x,p)}
\;+\;
\sum_{0\neq\gamma\le \alpha} c_{\alpha,\gamma}\,\Op{p^{\alpha-\gamma}\,\partial_x^\gamma a(x,p)},
\end{equation}
for suitable combinatorial constants $c_{\alpha,\gamma}\in\CC$.

Set
\[
b_0(x,p):=p^\alpha a(x,p),\qquad
b_\gamma(x,p):=p^{\alpha-\gamma}\,\partial_x^\gamma a(x,p)\quad(0\neq\gamma\le\alpha).
\]
Since $a\in S^{-m}_M(U)$, we have
\[
b_0\in S^{|\alpha|-m}_M(U),\qquad
b_\gamma\in S^{|\alpha|-|\gamma|-m}_M(U)\subset S^{-1}_M(U)\quad(0\neq\gamma\le\alpha),
\]
and because $|\alpha|\le m$ these orders are $\le 0$. In particular,
\[
b_0,\;b_\gamma\in S^0_M(U)\qquad(\gamma\le\alpha).
\]

Define the zero-order symbol
\[
\tilde a_\alpha(x,p)\;:=\;M^{m-|\alpha|}\,p^\alpha a(x,p)\in S^0_M(U).
\]
Then $b_0 = M^{|\alpha|-m}\tilde a_\alpha$, and the symbol estimates in \cref{lem:basic properties of PDO} (Leibniz rule and the bounds for multiplication by $p$ and differentiation in $x$) imply that, for all multi-indices $|\tilde\alpha|,|\beta|\le 2N$,
\[
C^{\,0}_{M,\tilde\alpha\beta}(b_\gamma)
\;\le\;
C_{\alpha,m,N}\,M^{|\alpha|-m}
\sum_{|\tilde\alpha'|,|\beta'|\le 2N}
C^{\,0}_{M,\tilde\alpha'\beta'}(\tilde a_\alpha),
\qquad 0\le\gamma\le\alpha,
\]
for some constant $C_{\alpha,m,N}$ depending only on $\alpha,m,N,n$ (but not on $a$).

Apply \cref{prop:L2 boundedness of zero order symbols} to each term in \eqref{eq:Palpha-expansion}. Using that $V\subseteq U$ and the triangle inequality, we get
\begin{align*}
\norm{P^\alpha \Op{a}}_{\calB(L^2(\RR^n)\to L^2(V))}
&\le
\sum_{\gamma\le\alpha} |c_{\alpha,\gamma}|\,
\norm{\Op{b_\gamma}}_{\calB(L^2(\RR^n)\to L^2(V))}\\[1mm]
&\le
C_{n,N}(U,V)\!
\sum_{\gamma\le\alpha} |c_{\alpha,\gamma}|
\sum_{|\tilde\alpha|,|\beta|\le 2N} C^{\,0}_{M,\tilde\alpha\beta}(b_\gamma)\\[1mm]
&\le
C_{n,N}(U,V)\,M^{|\alpha|-m}
\sum_{|\tilde\alpha|,|\beta|\le 2N} C^{\,0}_{M,\tilde\alpha\beta}(\tilde a_\alpha),
\end{align*}
after enlarging $C_{n,N}(U,V)$ by a factor depending only on $\alpha,m,N,n$ (which we absorb into the notation).

Recalling that $\tilde a_\alpha(x,p)=M^{m-|\alpha|}p^\alpha a(x,p)$ gives the stated estimate.
\end{proof}

    \subsection{Elliptic second order operators}
    Let $M\geq1$ and $U\subseteq\RR^n$ open be given. We study the symbol
    \eql{
        l(x,p) =\norm{p}^2 + \gamma_1(x)\cdot p + \gamma_0(x)\qquad(x\in U,p\in\RR^n)
    } where $\gamma_{1,1},\dots,\gamma_{1,n},\gamma_0$ are complex-valued functions on $U$, and $\gamma_1(x)\cdot p \equiv \sum_{j=1}^n \gamma_{1,j}(x)p_j$. We denote the associated pseudodifferential operator $\calL \equiv \Op{l}$.

    \begin{defn}\label{def:elliptic operators} 
    We say that $\calL$ is \emph{elliptic in $S^2_M(U)$} iff \eq{l\in S^2_M(U)\qquad {\rm and}\qquad \frac{1}{l}\in S^{-2}_M(U)\,.} The $S^2_M(U)$ seminorms of $l(x,p)$ and the $S^{-2}_M(U)$ seminorms of $1/l(x,p)$ (there are only finitely many relevant ones) are called the \emph{elliptic $S^2_M(U)$ constants of $\calL$}.
    \end{defn}

    In the following result, we exhibit a \emph{local} resolvent to $\calL$ via its elliptic property.
    
    \begin{lem}[Local elliptic lemma]\label{lem:local elliptic lemma} Let $U,V\subseteq\RR^n$ be two open \replaced{subset}{subsets} with $V\subseteq U$ relatively compact. Let $\calL$ be elliptic in $S^2_M(U)$. Let $\chi\in C^\infty(\RR^n)$ be supported in $V$. Then there exist linear operators $A:L^2(\RR^n)\to H^2(\RR^n)$ and $\calE:L^2(\RR^n)\to H^1(\RR^n)$ with the following properties:
    \begin{enumerate}
        \item $\calL A  -\calE= \chi$ as operators on $L^2(\RR^n)$. We interpret $A$ as a "local" resolvent of $\calL$ within $V$ up to the error $\calE$.
        \item $A,\calE$ map $L^2(\RR^n)$ into $L^2(V)$, i.e., $\supp Af, \supp \calE f\subseteq V$ for all $f\in L^2(\RR^n)$.
        \item $\norm{\partial^\alpha A f}_{L^2(\RR^n)} \leq C M^{|\alpha|-2}\norm{f}_{L^2(\RR^n)}$ for $|\alpha|\leq 2,f\in L^2(\RR^n)$.
        \item $\norm{\partial^\alpha \calE f}_{L^2(\RR^n)}\leq C M^{|\alpha|-1}\norm{f}_{L^2(\RR^n)}$ for $|\alpha|\leq 1,f\in L^2(\RR^n)$. 
    \end{enumerate} Here, the constant $C$ may be taken to depend only on the elliptic $S^2_M(U)$ constants of $\calL$ and the $C^\infty$ seminorms of $\chi$. Moreover, if the coefficients of $\calL$ depend analytically on a parameter $\lambda\in\CC$, then the operators $A$ and $\calE$ may be taken to depend analytically on $\lambda$.
    \end{lem}
    \begin{proof}
        Let $l$ be the symbol associated to $\calL \equiv \Op{l}$. Define the symbol \eq{
            a(x,p) := \frac{\chi(x)}{l(x,p)}\qquad(x\in U,p\in\RR^n)
        } and the associated operator $A:=\Op{a}$. First, we note $a\in S^{-2}_M(\RR^n)$ thanks to the ellipticity of $\calL$. Moreover, \Cref{eq:action of pseudo-diff op} implies that for any $f\in L^2(\RR^n)$, $Af$ will be supported within $V$ since $x\mapsto a(x,p)$ is supported there, by its construction.
        
        Next, let $\tilde{\chi}\in C_0^\infty(U)$ satisfy $\tilde{\chi} =1 $ on $V$. Clearly $\tilde{\chi}\in S^0_M(U)$, and so, thanks to \cref{eq:degree of product of two symbols where one is diagonal}, $\tilde{\chi}l$ is in $\replacedm{S^0_M(U)}{S^2_M(U)}$. Applying the decomposition \cref{eq:decomposition of product of two pdos into highest degree and remainder} on the product of symbols $\tilde{\chi}l$ and $a$ we find there must exist some symbol $\ve\in S^{-1}_M(U)$ which obeys \eq{
        \tilde{\chi}(X) \calL A =  \Op{\tilde{\chi} l a}+\Op{\ve} 
        } with $\ve$ the associated symbol, which must be in $S^{-1}_M$ as the first term is in the top degree. By construction, $\tilde{\chi} l a = \tilde{\chi}\chi=\chi$, so we get \eql{\label{eq:definition of remainder term}
        \tilde{\chi}(X) \calL A =  \chi(X) +\Op{\ve}\,.
        } Let $\calE := \Op{\ve}$.

        Since $\calL$ acts by differentiation and multiplication with smooth coefficient functions, its application does not change the support of a function it acts on. This implies that $\supp(\calL Af)\subseteq V$ for any $f\in L^2(\RR^n)$ and $\tilde{\chi} \calL A f = \calL A f$. Then \Cref{eq:definition of remainder term} implies that $\calE f$ is also supported within $V$.

        Applying \Cref{cor:regularity of symbols of negative degree} on $A$ and $\calE$ implies their co-domains are $H^2$ and $H^1$ respectively, with the stated estimates.

        Next, the definition of $A,\calE$ via the symbols, implies that if $l$ depends analytically on a parameter $\lambda$, then so do the symbols $a,\ve$ and hence so do the operators $A,\calE$.
    \end{proof}

    \subsection{Dilations}\label{sec:dilations}
    For any $\alpha>0$ and $S\subseteq\RR^n$, we write \eq{\alpha S \equiv \Set{\alpha x\in \RR^n | x\in S}\,.} We then have a unitary dilation operator $\mathfrak{U}_\alpha$ on $L^2(\RR^n)$ given by
    \eql{\label{eq:dilation operators}
    \br{\mathfrak{U}_\alpha \psi}(x) \equiv \alpha^{n/2}\psi(\alpha x)\qquad(x\in\RR^n;\psi\in L^2(\RR^n))
    } and locally $\mathfrak{U}_\alpha : L^2(\alpha S)\to L^2(S)$ by the same formula. Then, 
    \eq{
    \mathfrak{U}_\alpha^\ast X \mathfrak{U}_\alpha = \frac{1}{\alpha} X\quad {\rm and}\quad 
      \mathfrak{U}_\alpha^\ast P \mathfrak{U}_\alpha = \alpha P\,.}
      Further, let $a \in S^m_M(U)$. Then, with \eq{
      \alpha U\times\RR^n\ni(x,p)\mapsto\widetilde{a}_\alpha(x,p) := a(x/\alpha,\alpha p)
      } we have $\widetilde{a}_\alpha\in S^m_M(\alpha U)$ and 
   \eql{\label{eq:scalings of PDO}
   \mathfrak{U}_\alpha^\ast \Op{ a}\mathfrak{U}_\alpha=\Op{\widetilde{a}_\alpha}
   } as well as
   \eq{
    C_{M,\alpha\beta}^m\br{\widetilde{a}_\alpha} \leq \alpha^{|\beta|-|\alpha|}\max\br{\Set{1,\alpha}}^{m-|\beta|}C_{M,\alpha\beta}^m\br{a}\,.
   }

\section{The Landau resolvent for complex magnetic fields}\label{sec:The Landau Resolvent at complex magnetic fields}
In this section we want to study the resolvent of the Landau Hamiltonian
\eql{
H^{\mathrm{Landau}}_{B} := \Big(P-\frac12  B X^\perp\Big)^2
} where $B>0$ is the magnetic field strength. Below, we'll allow $B\in\CC$.

Since $H^{\mathrm{Landau}}_{B}$ is quadratic in both $X$ and $P$ it may be explicitly diagonalized. An easy route is to write the heat kernel \cite{Avron_Herbst_Simon_78}: For $x,y\in\RR^2$ and $t>0$, 
\begin{equation}
\exp\br{-t H^{\rm Landau}_B}(x,y) = \frac{B}{4\pi\sinh\left(Bt\right)}\exp\left(-\frac{B}{4}\coth\left(Bt\right)\norm{x-y}^{2}-\ii\frac{B}{2}x\wedge y\right)\ .
\end{equation} 
Its Laplace transform  yields the following expression for the resolvent kernel at \deleted{at }spectral parameter $z\in\CC$, such that  $\Re{z}<B$ \cite[Lemma 5.1]{KOROTYAEV2004221}:
\begin{align}
  \br{H^{\mathrm{Landau}}_B-z\Id}^{-1}(x,y) &= \int_{t=0}^\infty\ee^{tz}\exp\left(-tH^{\mathrm{Landau}}_B\right)\left(x,y\right)\dif{t}\nonumber \\
  &= \frac{1}{4\pi}\Gamma\br{\frac12-\frac{z}{2B}}U\br{\frac12-\frac{z}{2B},1,\frac{B}{2}\norm{x-y}^2}\times \nonumber \\
  &\qquad\qquad\times\exp\br{-\ii\frac{B}{2}x\wedge y -\frac14 B \norm{x-y}^2}\,.
\label{eq:landau-res-kernel}  
\end{align}
Here $U$ is Tricomi's confluent hypergeometric function \cite[pp. 129]{Temme_2014_asymp_methods_For_integrals}:
\eq{
    U(a,1,z) := \frac{1}{\Gamma(a)}\int_{t=0}^\infty \ee^{-tz}t^{a-1}\br{1+t}^{-a}\dif{t}\qquad(a,z\in\CC:\Re{a}>0,\Re{z}>0)\,.
} %Its asymptotic behavior as $z\to\infty$ is given by $U(a,b,z)\sim z^{-a}$. 

Since $z\mapsto\Gamma(z)$ is meromorphic with simple poles on $\ZZ_{\leq0}$ and $U(a,1,z)$ is entire in $a\in\CC$ and analytic for $z\in\CC\setminus(-\infty,0]$ \cite{DLMF13}, we see that $\br{H^{\mathrm{Landau}}_B-z\Id}^{-1}(x,y)$ itself is meromorphic with poles at  $z\in \big\{B(2n+1):n\ge0\big\}$.

We are interested in setting $B:=\lambda b$ where $b>0$ and $\lambda\in\CC$ with $\Re{\lambda}>0$, and choosing a spectral parameter of the form $z_\lambda:=-\lambda^2+\mu\lambda$ for some $\mu\in\CC$. We are seeking an analytic continuation (from $\lambda\in\RR$ to complex values of $\lambda$) of the integral kernel operator \eq{\RR^2\times\RR^2\setminus{\rm diagonal}\ni(x,y)\mapsto\br{H^{\mathrm{Landau}}_{\lambda b}-z_\lambda\Id}^{-1}(x,y)} as given explicitly above. The $\Gamma$ factor has infinitely many poles at $\lambda\in\mu-b\br{2\NN_{\geq0}+1}$ so that if $\Re{\lambda}>\Re{\mu}-b$ we avoid all these poles. The $U$ factor is entire in its first argument, and has a branch cut $(-\infty,0]$ in its last argument \cite{DLMF13}. Hence for $\Re{\lambda}>\max\br{\Set{0,\Re{\mu}-b}}$ there are no poles or branch cuts of this integral kernel as a function of $\lambda$. However, it does \emph{not} continue to a bounded linear operator $L^2\to L^2$; indeed, see \cref{prop:Landau resolvent does NOT continue} below.

To remedy this situation, let $K\subseteq\RR^2$ be compact and $Q:=\chi_K(X)$. We use the integral kernel above to define an operator which deserves to be denoted by $R^{\rm Landau}_{\lambda b}(z_\lambda)Q$ even for complex $\lambda$ since it analytically continues $R^{\rm Landau}_{\lambda b}(z_\lambda)Q$ for real $\lambda$. This operator's domain, however, is not $L^2(\RR^2)$ but rather $QL^2(\RR^2)$. For $f\in QL^2(\RR^2)$ and $\lambda\in\CC$ with $\Re{\lambda}>\Re{\mu}-b$, it is given as follows. For all $x\in\RR^2$, we have by \cref{eq:landau-res-kernel} :
\eql{\label{eq:analytic continuation of Landau resolvent for complex magnetic fields}
(R^{\rm Landau}_{B}(z) f)(x) :=\int_{y\in K}\dif{y}\ K_{B,z}(\norm{x-y})\ee^{-\ii\frac{B}{2}x\wedge y } f(y)\,.
} with
\eql{\label{eq:kernel of Landau resolvent without phase}
K_{B,z}(r) := \frac{1}{4\pi}\Gamma\br{\frac12-\frac{z}{2B}} U\br{\frac12-\frac{z}{2B},1,\frac{B}{2}r^2}\ee^{ -\frac14 B r^2}\qquad(r>0)\,.
}

\begin{prop}\label{prop:analytic extension of Landau resolvent on compact sets}
	Let $K\subseteq\RR^2$ be compact and $Q=\chi_K(X)$. Let $b,\nu>0$ and $\lambda,\mu\in\CC$ be  such that $\nu\Re{\lambda}>\max\br{\Set{0,\Re{\mu}-b}}$, $B:=\lambda b$, $z_\lambda=-\nu\lambda^2+\mu \lambda$.
    
    Then\deleted{,} the integral operator  \cref{eq:analytic continuation of Landau resolvent for complex magnetic fields} is a bounded linear operator $QL^2(\RR^2)\to L^2(\RR^2)$ obeying the bound
    \begin{align}
    \norm{ R^{\rm Landau}_{\lambda b}(z_\lambda)Q}_{\calB(QL^2(\RR^2)\to L^2(\RR^2))}
    &\lesssim \abs{\lambda}^\sharp \exp\br{\frac{\br{\frac{1}{2}b\abs{\Im{\lambda}}D}^2}{b \Re{\lambda}}}.\label{eq:Landau boundedness}
    \end{align}
    where $D:= D_K = \sup_{y\in K}\norm{y}$.
    
    Moreover, if $\Re{\lambda}$ is sufficiently large, then for all $S\subseteq \RR^2$, we have the off-diagonal bound: for any $j=1,2$ and $\alpha=0,1$,
        \begin{align}
        &\norm{\chi_S(X) P_j^\alpha R^{\rm Landau}_{\lambda b}(z_\lambda)Q}_{\calB(QL^2(\RR^2)\to L^2(\RR^2))}\notag\\
        &\qquad \leq \widetilde{C}
        \exp\br{ -\frac12\br{\sqrt{\nu}\br{\Re{\lambda}-1}-bD\abs{\Im{\lambda}}}\dist(S,K)}.
        \label{eq:Landau off diagonal decay}
        \end{align}
\end{prop}
\begin{proof}
    
	Set $A := 1/2-z_\lambda/ (2\lambda b)\in\CC$ with $\Re{A}>0$ and $w := \frac12\lambda b\norm{x-y}^2$ with $\Re{w}>0$ (assuming $x\neq y$ for the moment).
    Then \eq{
    \abs{\Gamma(A) U(A,1,w)} &= \abs{\int_{t=0}^{\infty}\exp\br{-t w}t^{A-1}(1+t)^{-A}\dif{t}} \\
    &\leq \int_{t=0}^{\infty}\exp\br{-t \Re{w}}t^{\Re{A}-1}(1+t)^{-\Re{A}}\dif{t} \\
    &\leq \int_{t=0}^{1}t^{\Re{A}-1}\dif{t} + \int_{t=1}^{\infty}\exp\br{-t \Re{w}}t^{-1}\dif{t} \\ 
    &\leq \frac{1}{\Re{A}} + \int_{t=1}^{1/\Re{w}}t^{-1}\dif{t} + \int_{t=1/\Re{w}}^{\infty}\exp\br{-t \Re{w}}t^{-1}\dif{t} \\ 
    &\leq \frac{1}{\Re{A}} + \abs{\log\br{\Re{w}}} + 1\\
    &\leq \br{1+\frac{1}{\Re{A}}}\br{1+\abs{\log\br{\Re{w}}}}\,.
}

    Moreover, we also have 
    \eq{
    \abs{\exp\br{-\ii \frac{\lambda b}{2}x\wedge y }} = \exp\br{\frac{\Im{\lambda} b}{2}x\wedge y}\,.
    } %so that
    % \eq{
    % \abs{\br{H^{\mathrm{Landau}}_{\lambda b}-z_\lambda\Id}^{-1}(x,y)} &\leq \frac{\br{1+\frac{1}{\Re{A}}}}{4\pi}\br{1+\abs{\log\br{\Re{w}}}} \times \\
    % &\qquad \times \exp\br{-\frac12 \Re{w}+ \frac{\abs{\Im{\lambda }}b}{2}\abs{x\wedge y}}\,.
    % } 

    Now if $y\in K$ then, defining \[ D:=\ D_K\ =\ \sup_{y\in K}\norm{y},\] 
    and using that 
    $\abs{x\wedge y} = \abs{\br{x-y}\wedge y} \leq  D \norm{x-y} $ we find that the whole kernel is upper bounded by 
        \eq{
    \abs{\br{H^{\mathrm{Landau}}_{\lambda b}-z_\lambda\Id}^{-1}(x,y)} &\leq \frac{1+\frac{1}{\Re{A}}}{4\pi}\br{1+\abs{\log\br{\frac12\Re{\lambda} b\norm{x-y}^2}}} \times \\
    &\qquad \times \exp\br{-\frac{1}{4} b \Re{\lambda}\norm{x-y}^2 + \frac{1}{2} b \abs{\Im{\lambda}}D\norm{x-y}} \,.
    } 

        At this point it is convenient to "forget" the fact that $y$ only ranges in $K$ (as an upper bound) to simplify the invocation of  Schur's test 
    \eq{\norm{G}^2 \leq \br{\sup_x \int_{y} \abs{G(x,y)}\dif{y}}\br{\sup_y \int_{x}\abs{G(x,y)}\dif{x}}} to 
    \eq{
    \norm{ R^{\rm Landau}_{\lambda b}(z_\lambda)Q}_{\calB(QL^2(\RR^2)\to L^2(\RR^2))} &\leq \frac{1+\frac{1}{\Re{A}}}{4\pi}\int_{x\in\RR^2}\dif{x}\br{1+\abs{\log\br{\frac12\Re{\lambda} b\norm{x}^2}}} \times \\
    &\qquad \times \exp\br{-\frac{1}{4} b \Re{\lambda}\norm{x}^2 + \frac{1}{2} b \abs{\Im{\lambda}}D\norm{x}}\,.
    }
    Now using 
    \eq{
    \int_{r=0}^\infty r\br{1+\abs{\log\br{\alpha r^2}}}
    \exp\br{-\frac12\alpha r^2 + \beta r }\dif{r}
    &\leq \int_{r=0}^\infty r\br{1+2\sqrt{\alpha} r + \frac{2}{\sqrt{\alpha} r}}\\
    &\qquad\times \exp\br{-\frac12\alpha r^2 + \beta r }\dif{r} \\
    &= \frac{\beta +1}{\alpha }
    +\frac{\sqrt{2\pi}\left(\alpha+\beta^2+\beta+1\right)}{\alpha^{3/2}}
    \ee^{\frac{\beta^2}{2\alpha}}.
    } we find the estimate \cref{eq:Landau boundedness}.
    % \eq{
    % \norm{ R^{\rm Landau}_{\lambda b}(z_\lambda)Q}_{\calB(QL^2(\RR^2)\to L^2(\RR^2))} &\leq \frac{1}{2}\br{1+\frac{2b}{\nu\Re{\lambda}+b-\Re{\mu}}} \times \\
    % &\times (\frac{\frac{1}{2}b\abs{\Im{\lambda}}D+1}{\frac{1}{2}b \Re{\lambda}}+\frac{\sqrt{2\pi}}{\br{\frac{1}{2}b \Re{\lambda}}^{3/2}}\exp\br{\frac{\br{\frac{1}{2}b\abs{\Im{\lambda}}D}^2}{b \Re{\lambda}}}\times\\
    % &\times\br{\frac{1}{2}b \Re{\lambda}+1+\frac{1}{2}b\abs{\Im{\lambda}}D+\br{\frac{1}{2}b\abs{\Im{\lambda}}D}^2})
    % } which implies .

    To get the rate of off-diagonal decay, we return to the integral expression above. Now $\Re{\lambda}$ is a large parameter. Since
    \eq{
    A&=\frac12-\frac{-\nu\lambda^2+\mu\lambda}{2\lambda b}
    =\frac{\nu\lambda+b-\mu}{2b},\\
    2b\Re{A}&=\nu\Re{\lambda}+b-\Re{\mu},
    }
    with $\gamma:=\Re{\lambda}$, $r := \norm{x-y}$, and $w := \frac12\lambda b\norm{x-y}^2$, set
    \eq{
    m_b(r)&:=\frac14 r\sqrt{b^2r^2+4\nu}-\frac{br^2}{4}\\
    &\quad+\frac{\nu}{2b}\log\!\Bigg(1+\frac{b r\left(\sqrt{b^2r^2+4\nu}+br\right)}{4\nu}\Bigg).
    }
    Then
    \eq{
        \abs{\Gamma(A) U(A,1,w)}
        &\leq \int_{t=0}^{\infty}\frac{\br{1+\frac{1}{t}}^{-\frac{1}{2}+\frac{\Re{\mu}}{2b}}}{t}
        \exp\br{-\gamma\br{\frac12 b t r^2+\frac{\nu}{2b}\log\br{1+\frac{1}{t}}}}\dif{t}\\
        &\leq \int_{t=0}^{\infty}\frac{\br{1+\frac{1}{t}}^{-\frac{1}{2}+\frac{\Re{\mu}}{2b}}}{t}
        \exp\br{-\br{\frac12 b t r^2+\frac{\nu}{2b}\log\br{1+\frac{1}{t}}}}\dif{t}\\
        &\qquad\times \exp\br{-(\gamma-1)m_b(r)}\,.
    } Here we have calculated the minimal value of the function (of $t$) in the exponential and pulled it out of the integral, akin to a Laplace asymptotic calculation. Now since the logarithm term within the exponential is positive and $\sqrt{b^2 r^2 + 4\nu}\geq2\sqrt{\nu}$ we see that  
    \eq{
    &\abs{\br{H^{\mathrm{Landau}}_{\lambda b}-z_\lambda\Id}^{-1}(x,y)}\notag\\
    &\qquad\leq \br{1+\frac{2b}{\nu+b-\Re{\mu}}}
    \br{1+\abs{\log\br{\frac12 b \norm{x-y}^2}}} \\
    &\qquad\quad\times \exp\br{-\frac{\sqrt{\nu}}{2}\br{\Re{\lambda}-1} \norm{x-y}
    +\frac12 b \abs{\Im{\lambda}} D\norm{x-y} }.
    } \replaced{but}{But} now $\norm{x-y}\geq \dist(S,K)$ so Schur's test again yields the claimed bound. 
    
    For the case $\alpha\neq0$, we use
    \eq{
        \br{P_j^\alpha R^{\rm Landau}_{\lambda b}(z_\lambda)}(x,y)
        &= -\ii \partial_{x_j}  R^{\rm Landau}_{\lambda b}(z_\lambda)(x,y) \\
        &= -\ii \partial_{x_j} K_{\lambda b,z_\lambda}\br{\norm{x-y}}
        \exp\br{-\ii\frac{\lambda b}{2}x\wedge y } \\
        &= -\ii \Bigg[\frac{x_j-y_j}{\norm{x-y}}K'_{\lambda b, z_\lambda}(\norm{x-y})\\
        &\qquad\left.{}-\ii \frac{\lambda b}{2}\sum_{k=1}^2\ve_{jk}y_k
        K_{\lambda b,z_\lambda}\br{\norm{x-y}}\right]
        \exp\br{-\ii\frac{\lambda b}{2}x\wedge y }\,.
    } The second term is dealt with in the same manner as before, so we only study $K'$:
    \eql{\label{eq:derivative of Landau positive kernel}
    K'_{B, z}(r) &= -\frac{B r}{4\pi}\Gamma\br{\frac12-\frac{z}{2B}}\ee^{-\frac14 B r^2}\notag\\
    &\quad\times\br{\br{\frac12-\frac{z}{2B}}U\br{\frac32-\frac{z}{2B},2,\frac12 B r^2}
    + \frac12 U\br{\frac12-\frac{z}{2B},1,\frac{B}{2}r^2} }\,.
    } The second term is identical to what was studied above, hence we only need to study the term involving $U\br{\frac32-\frac{z}{2B},2,\frac12 B r^2}$:
    \eq{
    \Gamma\br{A}U\br{1+A,2,w} = \frac{1}{A}\int_{t=0}^\infty\ee^{-w t}t^{A}\br{1+t}^{-A}\dif{t} 
    } and we see that this leads to the same asymptotics of the integral as above, for the large parameter $\Re{\lambda}$.    
\end{proof}
To motivate our analysis throughout the paper, we show that we have no hope to continue the Landau resolvent without restricting to a compact set. In fact, we will even allow a relatively compact perturbation of the Landau Hamiltonian.
\begin{prop}\label{prop:Landau resolvent does NOT continue}
    Let $v:\RR^2\to\CC$ be bounded and of compact support.

    Then for $\nu>0,\mu\in\CC$ fixed, $\lambda\in\CC\setminus\RR$ with $\Re{\lambda}>0$, and setting $z_\lambda := -\nu \lambda^2+\mu\lambda$,
    \eq{
    H^{\rm Landau}_{b \lambda}+\lambda^2 v(X)-z_\lambda\Id 
    } is \emph{not} invertible for all $\lambda$ and $\Re{\lambda}$ sufficiently large.
     
\end{prop}

\begin{proof}
    We shall apply the Weyl criterion for the spectrum, i.e., we shall show that for any $\delta>0$ there exists  $\psi_\delta\in L^2(\RR^2)$ such that \eq{
        \norm{\br{H^{\rm Landau}_{b \lambda}+\lambda^2 v(X)-z_\lambda\Id }\psi_\delta}\leq\delta\norm{\psi_\delta}\,.
    } 

   The plan will be to define a smooth bump function,  $\widetilde{\psi_\delta}$, supported at a large distance from the support of  $v(x)$, in such a way that the magnitude of the complex exponential appearing in the Landau resolvent kernel, 
    $\exp\br{\frac{\Im{\lambda} b}{2}x\wedge y}$ can be made arbitrarily large
     on the support of $\widetilde{\psi_\delta}$.  Then,  $\psi_\delta := R_{\lambda b}^{\rm Landau}(z_\lambda)\widetilde{\psi_\delta}$, which by \cref{prop:analytic extension of Landau resolvent on compact sets} is well-defined for  complex $\lambda$ since $\widetilde{\psi_\delta}$ is supported on a compact set, will be large.
     
On the other hand, \eq{
    \br{H^{\rm Landau}_{b \lambda}+\lambda^2 v(X)-z_\lambda\Id }\psi_\delta = \widetilde{\psi_\delta} + \lambda^2 v(X) R_{\lambda b}^{\rm Landau}(z_\lambda)\widetilde{\psi_\delta}
    } will be bounded, thanks to the off-diagonal decay bound on $R_{\lambda b}^{\rm Landau}(z_\lambda)$, \cref{eq:Landau off diagonal decay}.
 
    Let $\xi\in\RR^2$ be chosen below with $\norm{\xi}\gg 1$. Let $\widetilde{\psi_\delta}\geq0$ be a smooth bump function supported within $K:=B_{\beta \abs{\lambda}^{-1}\norm{\xi}^{-1}}(\xi)$ for some $\beta>0$ and obeying $\int_{\RR^2}\widetilde{\psi_\delta} = 1$ so that $\norm{\widetilde{\psi_\delta}}_2\lesssim \beta^{-1}\abs{\lambda}\norm{\xi}$. We find, by the off-diagonal bound \cref{eq:Landau off diagonal decay}, for sufficiently large $\abs{\lambda}$:  \eq{
    \norm{\br{H^{\rm Landau}_{b \lambda}+\lambda^2 v(X)-z_\lambda\Id }\psi_\delta} &\leq C \beta^{-1}\abs{\lambda}\norm{\xi}\ .
    } 

    To obtain a lower bound on $\norm{\psi_\delta}$, we shall require $U\subseteq\RR^2$ to be a subset on which the Landau phase factor is large, i.e., for  some $c>0$ \eql{\label{eq:Landau phase is large}\abs{\exp\br{-\ii\frac{\lambda b}{2}x\wedge \xi}} \geq \exp\br{\frac12c \abs{\Im{\lambda}}\norm{\xi}}. } Then \eq{
    \norm{\psi_\delta}^2 & = \int_{x\in \RR^2}\abs{\psi_\delta(x)}^2\dif{x} \\ 
    &\geq \int_{x\in U}\abs{\br{R_{\lambda b}^{\rm Landau}(z_\lambda)\widetilde{\psi_\delta}}(x)}^2\dif{x} \\
    &\geq  \exp\br{c \abs{\Im{\lambda}}\norm{\xi}} \int_{x\in U}\abs{\exp\br{\ii \frac{\lambda b}{2}x\wedge\xi}\br{R_{\lambda b}^{\rm Landau}(z_\lambda)\widetilde{\psi_\delta}}(x)}^2\dif{x}\,.
    } 
  We claim we can choose the set $U$ such that both \cref{eq:Landau phase is large} holds and the latter integral over $U$ satisfies a lower bound which is independent of $\xi$. We make the following choice: \eq{
        U := \Set{x \in \RR^2 | C < \norm{x-\xi} < C' } \cap \Set{x \in \RR^2 | \sgn(\Im{\lambda})\ ( x\wedge \xi) \geq c \norm{\xi}}}
    The set $U$ has \deleted{an }area\replaced{, which}{ that} is strictly positive and uniformly bounded away from zero, independently of $|\lambda|$ and $\xi\in\mathbb R^2$.

    We then rewrite, via \cref{eq:analytic continuation of Landau resolvent for complex magnetic fields} and abusing $K(x-y)\equiv K(\norm{x-y})$,
    \eq{
    \exp\br{\ii \frac{\lambda b}{2}x\wedge\xi}
    \br{R_{\lambda b}^{\rm Landau}(z_\lambda)\widetilde{\psi_\delta}}(x)
    &= \int_{y\in K} K_{B,z}(x-y)\ee^{-\ii\frac{B}{2}x\wedge \br{y-\xi} }
    \widetilde{\psi_\delta}(y)\dif{y} \\
    &= \int_{y\in K} K_{B,z}(x-\xi)\widetilde{\psi_\delta}(y)\dif{y} \\
    &\quad+\int_{y\in K} \Big(K_{B,z}(x-y)\ee^{-\ii\frac{B}{2}x\wedge \br{y-\xi} }\\
    &\qquad\qquad{}-K_{B,z}(x-\xi)\Big)\widetilde{\psi_\delta}(y)\dif{y} \\ 
    &= K_{B,z}(x-\xi)\\
    &\quad+\int_{y\in K} \Big(K_{B,z}(x-y)\ee^{-\ii\frac{B}{2}x\wedge \br{y-\xi} }\\
    &\qquad\qquad{}-K_{B,z}(x-\xi)\Big)\widetilde{\psi_\delta}(y)\dif{y}\,.
    }

    Now we may establish, via \cref{eq:derivative of Landau positive kernel} and the lines below it, that
    \eq{
    \abs{K_{B,z}(\norm{x})-K_{B,z}(\norm{y})} \leq C \abs{B}\norm{x}\norm{x-y}\abs{K_{B,z}(\norm{x})}
    } for all $\norm{x-y}\leq c \abs{B}^{-1}\norm{x}^{-1}$. Using these estimates one may prove \eq{
    \abs{\br{K_{B,z}(\norm{x-y})\ee^{-\ii\frac{B}{2}x\wedge \br{y-\xi} }-K_{B,z}(\norm{x-\xi})}} &\leq C \beta \abs{K_{B,z}(\norm{x-\xi})}
    } and hence we have the desired bound
    \eq{
    \norm{\psi_\delta}^2 \gtrsim  \exp\br{c \abs{\Im{\lambda}}\norm{\xi}} C_\lambda
    } for some constant $C_\lambda$ independent of $\xi$. By choosing $\xi$ sufficiently large we get our result.
\end{proof}

\section{Anisotropic Magnetic Harmonic Oscillator;\\ Appendix by Tal Shpigel}\label{sec:MHO}

Central to our proof of analyticity properties  of the resolvent of $h_\lambda$ with respect to $\lambda$ is our construction \added{of }an approximation of $(h^\theta_\lambda-z\Id)^{-1}$ (see  
\cref{eq:putative expression for the inverse of h_lambdatheta}, for $z$ varying on a circle enclosing $e_\lambda$, the ground state eigenvalue of $h_\lambda$).

Recall that for $x$ in a small ($\lambda$-dependent) neighborhood of the origin,  $h^\theta_\lambda = h_\lambda \approx h_\lambda^{\rm MHO} -\lambda^2\Id$, where  $h_\lambda^{\rm MHO}$ is the quantum magnetic harmonic oscillator Hamiltonian
(see \cref{eq:h-MHO}): 
\eql{h_{\lambda}^{\rm MHO }:=\left(P-\frac{1}{2}b\lambda X^{\perp}\right)^{2}+\frac{1}{2}\lambda^{2}\left\langle X,\Hessian{v}(0) X\right\rangle .} 
Here, $b>0$, $\lambda\in\RR$ and $\Hessian{v}(0)$ is a strictly positive $2\times 2$ matrix, arising in the Taylor expansion of $v$ about a non-degenerate minimum at $x=0$. Set  $h^{\rm MHO}:=h^{\rm MHO}_1$ and label its ($\lambda$-independent) eigenvalues $e_j^{\rm MHO}$, $j=0,1,2,\dots$. 

  Introduce the circular contour in $\CC$:
    \eql{ \mathbb{S}^1\ni\xi\mapsto z_\lambda(\xi) +\lambda^2 :=  \lambda\Big[e_{0}^{{\rm MHO}}+
    \replacedm{\frac{1}{2}}{c_{\mathrm{ctr}}}\left(e_{1}^{{\rm MHO}}-e_{0}^{{\rm MHO}}\right)\xi\Big].\label{eq:z_lam-def}}
    Thus $z_\lambda$ is at distance of order $\lambda$ away from  eigenvalues of $h_\lambda^{\rm MHO}$ and encloses one simple eigenvalue of $h_\lambda$, namely, $e_\lambda\approx -\lambda^2 +\lambda e_0^{\rm MHO}$, the ground state; see \cref{eq:rMHO}.

The approximate resolvent \cref{eq:putative expression for the inverse of h_lambdatheta} has a contribution corresponding  to the resolvent of $h_\lambda^{\rm MHO} -\lambda^2\Id$ at spectral parameter $z=z_\lambda$, lying on a circle:
\begin{equation} r_\lambda^{\rm MHO}(\lambda^2+z_\lambda) := \br{h_{\lambda}^{\text{MHO}}-\br{z_\lambda+\lambda^2}\Id}^{-1}  \label{eq:rMHO-1}\end{equation}

\added{In this section we prove analyticity of $\lambda\mapsto r_\lambda^{\rm MHO}(\lambda^2+z_\lambda)$ and the localized estimates needed in \cref{lem:continuation of cut off one well resolvent}.}

\begin{prop}\label{prop:MHO analyticity and bounds}

Fix $\ve\in(0,\frac\pi2)$, $\Lambda>0$ sufficiently large and consider the wedge \eq{
\Omega_{\Lambda,\ve} := \Set{\zeta\in\CC | \abs{\zeta}>\Lambda\quad\rm{ and }\quad\abs{\arg(\zeta)}<\frac{\pi}{2}-\ve }\,.
}
\begin{enumerate}
    \item 
 There exists $b_{\rm max}(\ve) > 0$ such that, for all $b\in(0,b_{\rm max}(\ve))$,
 and for any $\xi\in\mathbb{S}^1\subset\CC$:
 the mapping \eq{
    \lambda \mapsto  r_{\lambda}^{\rm MHO }(\lambda^2+z_\lambda(\xi)) \equiv \br{h_{\lambda}^{\rm MHO }-\br{\lambda^2+z_\lambda(\xi)}\Id}^{-1}
    }
  initially defined from $\lambda\in(0,\infty)$ into the space $\calB(L^2(\RR^2))$,  
 analytically continues to a  mapping from the wedge $\Omega_{\Lambda,\ve}$ to $\calB(L^2(\RR^2))$. Moreover, we have
 \eq{
    \norm{r^{\rm MHO}(\lambda^2+z_\lambda(\xi))} \lesssim \abs{\lambda}^{-1}\,.
 }

 % \textcolor{red}{[We seem to need $\lambda\mapsto P_j^\alpha r_{\lambda}^{\rm MHO }(z_\lambda)$ analytic for $\alpha=0,1$ and $j=1,2$. See \cref{eq:model-r_MHO}\.]}\\

    \item Pick a small constant $\eta>0$. Set $\delta := \Lambda^{-1/2+\eta}$. For $\lambda\in\Omega_{\Lambda,\ve} $ and $S,T\subseteq\RR^2$ obeying $T\subset B_{C_T \delta}(0)$ for some $C_T<\infty$ of order 1, and $S\subseteq B_{C_2 \delta}(0) \setminus B_{C_1\delta}(0)$ with $C_1 > C_T$,
    \eql{\label{eq:Combes-Thomas type estimate for MHO resolvent}
\norm{\chi_S(X)\ P_j^\alpha\ r_{\lambda}^{\mathrm{MHO}}(\lambda^2+z_\lambda(\xi))\ \chi_T(X)}_{\calB(L^2(\RR^2)\to L^\infty(\RR^2))} \lesssim \exp\br{-c \abs{\lambda}\br{\dist(S,T)}^2},
} 
for $j=1,2$ and $\alpha=0,1,2$,  uniformly in $\xi\in\mathbb S^1$. Clearly this estimate implies the $L^2\to L^2$ estimate via $\norm{\chi_S(X) g}_{L^2}\leq \sqrt{|S|}\norm{g}_{L^\infty}$ for any $g$.
\end{enumerate}
\end{prop}

\added{The proof of \cref{prop:MHO analyticity and bounds} is deferred to \cref{subsec:proof_MHO_prop}.}
\bigskip

Note that by \replaced{rotational in invariance}{rotational invariance} of $(P-\frac{b\lambda}{2}X^\perp)^2$, we may take $D^2v(0)$ to be diagonal with positive entries: $k_1^2$ and $k_2^2$ along the diagonal. Thus, we consider the operator:
\eql{\label{eq:h_MHO} h_{\lambda}^{\rm MHO }:=\left(P-\frac{1}{2}b\lambda X^{\perp}\right)^{2}+\frac{1}{2}\lambda^{2}(k_1^2X_1^2 + k_2^2X_2^2). } 
\\
In the following section we begin our analysis of the resolvent of the general anisotropic magnetic quantum harmonic oscillator Hamiltonian $h_\lambda^{\rm MHO}$, given by \cref{eq:h_MHO}.

\subsection{The Setup} \label{sec:setup}

 We start by quoting the form of the quantum harmonic oscillator  Hamiltonian from \cite{MATSUMOTO1995168}:
\begin{equation}\label{eq:matsumoto_hamiltonian}
    H = \frac{1}{2}\sum_{j=1}^2\left(i\frac{\partial}{\partial x_j} - A_j\right)^2 + \frac{1}{2}(k_1^2x_1^2 + k_2^2x_2^2),
\end{equation}
where $A_1 = \frac{B}{2}x_2$ and $A_2 = -\frac{B}{2}x_1$. Here, $k_1,  k_2>0$\deleted{ both positive}, $B\in\mathbb{R}$ and $x=(x_1,x_2)\in\mathbb{R}^2$. The fundamental solution of the heat equation, $q(x,y,t)= \exp(-Ht)(x,y)$ is given by
\begin{equation}
    q(x,y,t) = \frac{1}{2\pi}\sqrt{\text{det}\left(\frac{\partial^2\tilde{S}_{cl}(t,x,y)}{\partial x\partial y}\right)}\exp(-\tilde{S}_{cl}(t,x,y)),
\end{equation}
which is displayed explicitly in \cite[Proposition 2.1]{MATSUMOTO1995168}. 
 However, there is a small mistake in the expression of the Hamiltonian in \eqref{eq:matsumoto_hamiltonian}. For $q(t,x,y)$ as presented in \cite{MATSUMOTO1995168} to be the fundamental solution of the heat equation, the Hamiltonian should have the signs of $A_1$ and $A_2$ flipped: 
 \begin{equation}\label{eq:A-corrected} A_1=-\frac{B}{2}x_2\quad \textrm{and}\quad  A_2=\frac{B}{2}x_1.\end{equation}
 We will work with the Hamiltonian $H$, where the vector potential $A$ has the corrected signs, as  in \cref{eq:A-corrected}.

We will use this to prove bounds on the Green's function of $H$ modified to remove the effect of the ground state.

It is convenient to now make the following change of notations and variables:
\begin{align}
    &k_i \to \lambda k_i, \notag\\
    &B \to 2\lambda B, \notag\\
    &\lambda t \to  s \in \mathbb{R}, \notag\\
    &q(x,y,t) \to q(x,y,s,\lambda), \notag\\
    &\tilde{S}_{cl}(x,y,t) \to \phi(x,y,s,\lambda), \label{eq:change of variables}
\end{align}
where, at first, $\lambda$ is a large positive number, and $s=\lambda t\in\mathbb{R}$.

After the change of notation and variables, the Hamiltonian becomes
\begin{align}
   \widetilde H &= \frac{1}{2} \left[-\nabla^2 + 2i\lambda B(x_2\partial_{x_1}-x_1\partial_{x_2}) + \lambda^2\left(k_1^2+B^2\right)x_1^2 + \lambda^2\left(k_2^2+B^2\right)x_2^2\right] \notag\\ 
   &= \frac12\left[\ (P-\lambda BX^\perp)^2 + \lambda^2(\ k_1^2x_1^2\ +\ k_2^2x_2^2\ )\ \right] \label{eq:H-MHO}
\end{align}
where $\nabla^2 = \sum_{i}\partial_{x_i}^2$ is the laplacian operator and $k_1$,  $k_2>0$, $B\in\mathbb{R}$ and $x=(x_1,x_2)\in\mathbb{R}^2$.

Denote the operator in \eqref{eq:H-MHO} by $H_\lambda^{\rm MHO}(B,k)$ and note the relation:
 \[ h^{\rm MHO}_\lambda = 2 H_\lambda^{\rm MHO}\big(\frac{b}{2},\frac{k}{\sqrt2}\big),\]
 where $h^{\rm MHO}_\lambda$ is given by \cref{eq:h_MHO}.
 Set $H_\lambda^{\rm MHO}=H_\lambda^{\rm MHO}\big(\frac{b}{2},\frac{k}{\sqrt2}\big) $ and observe that
 \begin{equation} r^{\rm MHO}(\zeta) =  \Big(h_\lambda^{\rm MHO}-(\zeta+\lambda^2)\Big)^{-1}=\frac12 \Big(\ H_\lambda^{\rm MHO} - \frac{\zeta+\lambda^2}{2}\ \Big)^{-1}.
 \label{eq:from rMHO to resolvent of H_MHO}
 \end{equation}
 Hence, Proposition \ref{prop:MHO analyticity and bounds} reduces to a study of the resolvent of $H_\lambda^{\rm MHO}$.

 For brevity we set $H= H_\lambda = H_\lambda^{\rm MHO}$ and introduce the associated heat kernel, $q(x,y,s,\lambda)=\exp\left(-\lambda^{-1}Hs\right)(x,y)$, which satisfies:
  \begin{subequations}
\begin{align} \label{eq:heat_equation}
    (\lambda\partial_s + H)q(x,y,s,\lambda) &= 0\\
    \lim_{s\downarrow0} q(x, y, s,\lambda) &= \delta(x-y),
\end{align}
\end{subequations}
 where $\delta(x-y)$ is the Dirac delta function.

 \begin{rem}
  For $\lambda$ real,   $H_\lambda = H_\lambda^{\rm MHO}$ is a strictly positive  self-adjoint operator. Its resolvent kernel is obtained from the Laplace transform of the heat kernel.\\
 For ${\rm Re}(\mu) <  \lambda^{-1}\inf \sigma(H_\lambda)$, 
 \begin{equation}\label{eq:LT_of_heat}
 \big(H_\lambda -\lambda\mu\big)^{-1}(x,y)
  = \frac{1}{\lambda} \int_0^\infty \exp(\mu s)\ \exp(-\lambda^{-1}H_\lambda s)(x,y) ds\ .
 \end{equation}
    For complex $\lambda$,  we  shall use an expression of the type on the right hand side of \cref{eq:LT_of_heat} to give a sense to the resolvent of $H_\lambda^{\rm MHO}$ and to bound its operator norm.\\
     Finally, recall that (via  the resolvent of $H_\lambda$) our aim is to  
  construct and bound the operator $r^{\rm MHO}(z_\lambda(\xi))$, where $\xi\in\mathbb S^1\mapsto z_\lambda(\xi)\in\CC$ is given by  \cref{eq:z_lam-def}. By
  \cref{eq:from rMHO to resolvent of H_MHO}, we require the resolvent of $H_\lambda$ evaluated at the spectral parameter \[ \frac{z_\lambda(\xi)+\lambda^2}{2} = : \lambda \tau(\xi), \]
  where $\tau(\xi)= \frac{1}{2}e_{0}^{{\rm MHO}}+
    \replacedm{\frac{1}{4}}{\frac{c_{\mathrm{ctr}}}{2}}\left(e_{1}^{{\rm MHO}}-e_{0}^{{\rm MHO}}\right)\xi$. This indicates  that we'll require the  behavior of the expression in \cref{eq:LT_of_heat} for $\mu\in\CC$ of order one.
 \end{rem}

Our heat kernel, after the change of variables \cref{eq:change of variables}, is given explicitly by
\begin{equation}\label{eq:heat_kernel}
    q(x,y,s,\lambda) = \lambda P(s)\exp(-\phi(x,y,s,\lambda)),
\end{equation}
where the prefactor and the term in the exponent are given by
\begin{align} 
    P(s) &= \frac{f_+f_-}{2\pi}\sqrt{\frac{2k_1k_2}{K(s)}},\label{eq:prefactor}\\
    \phi(x,y,s,\lambda) &= \frac{\lambda f_+f_- \alpha_{12}(s)}{2K(s)}(x_1^2+y_1^2) + \frac{\lambda f_+f_- \beta_{12}(s)}{K(s)}x_1y_1\notag \\
    &+ \frac{\lambda f_+f_- \alpha_{21}(s)}{2K(s)}(x_2^2+y_2^2) + \frac{\lambda f_+f_- \beta_{21}(s)}{K(s)}x_2y_2\notag \\
    &+i\frac{\lambda B \gamma_1(s)}{K(s)}(x_1x_2-y_1y_2) - i\frac{\lambda B \gamma_2(s)}{K(s)}(x_1y_2-x_2y_1), \label{eq:P_and_phi}
\end{align}
and
\begin{subequations}
\label{eq:heat_kernel_parameters}
\begin{align}
    &f_\pm = \sqrt{(k_1\pm k_2)^2 + 4B^2}\label{eq:f_pm} \\
    &K(s) = f_-^2(k_1+k_2)^2(\cosh(f_+s)-1)-f_+^2(k_1-k_2)^2(\cosh(f_-s)-1),\notag \\
    &\alpha_{ij}(s) = k_i\big[f_-(k_1+k_2)\sinh(f_+s) + f_+(k_j-k_i)\sinh(f_-s)\big],\notag \\
    &\beta_{ij}(s) = k_i\bigg[\Big(f_+(k_i-k_j)+f_-(k_1+k_2)\Big)\sinh\left(\frac{f_--f_+}{2}s\right) \notag\\
    &\qquad\qquad + \Big(f_+(k_i-k_j)-f_-(k_1+k_2)\Big)\sinh\left(\frac{f_-+f_+}{2}s\right)\bigg],\notag\\
    &\gamma_1(s) = (k_1^2-k_2^2)\big(f_+^2(\cosh(f_- s)-1) -f_-^2(\cosh(f_+ s)-1) \big),\notag\\
    &\gamma_2(s) = 8f_+f_-k_1k_2\sinh\left(\frac{f_+}{2}s\right)\sinh\left(\frac{f_-}{2}s\right).
\end{align}
\end{subequations}

The ground state eigenfunction of $H$ is given by
\begin{equation}\label{eq:ground_state_eigenfunction}
    \psi_0(x, \lambda) = \left(\frac{\lambda^2\eta\zeta}{\pi^2}\right)^{1/4}\exp\left(-\frac{\lambda}{2}\left(\eta x_1^2 + \zeta x_2^2 - i\xi x_1x_2\right)\right),
\end{equation}
where
\begin{align} \label{eq:ground_state_parameters}
    \eta &= \frac{f_+k_1}{k_1+k_2}\notag\\
    \zeta &= \frac{f_+k_2}{k_1+k_2}\notag\\
    \xi &= \frac{2B(k_1-k_2)}{k_1+k_2}.
\end{align}
The reader should note that there is a small mistake in the expression of the ground state eigenfunction of $H$ in \cite{MATSUMOTO1995168}, where the cross term is missing. 

For real $\lambda$, the eigenvalue corresponding to the ground state is given \deleted{defined }by $H\psi_0 = E_0\psi_0$, where $E_0$ is the ground state energy and given as
\begin{equation}\label{eq:ground_state_energy}
    E_0 = \lambda \frac{f_+}{2},
\end{equation}
where $f_+$ is displayed in \cref{eq:f_pm}.

Once we have these formulas we extend the definitions to $\lambda\in\mathbb{C}$ by analytic continuation, as described in \cref{sec:first_main_thm,sec:second_main_thm}.

\subsection{Controlled Constants} \label{sec:controlled_constants}
We are given basic constants $c_\lambda$, which will be used to control $|\text{Arg}(\lambda)|$ and $k_1,k_2$ which specify the harmonic oscillator. Constants that are determined by the basic constants are called controlled constants and are denoted by $c, C, C',$ etc. These symbols may denote different constants in different occurrences. We write $X=\mathcal{O}(Y)$ to indicate that $|X|\leq CY$ for a \added{controlled }constant $C$.

\subsection{First Main theorem} \label{sec:first_main_thm}

We set $(x,y,\lambda,\mu)\in\mathbb{R}^2\times\mathbb{R}^2\times\Omega\times\mathbb{C}$ with $x\neq y$ and $Re(\mu)<\frac{f_+}{2}+c_1$ (small enough $c_1$), where $\Omega=\{\lambda\in\mathbb{C}:|\lambda|>C_\Omega,|\mathrm{arg}(\lambda)|<\frac{\pi}{2}-c_\Omega\}$ with $C_\Omega$ large enough. We define the modified Green's function as
\begin{equation}\label{eq:resolvent_definition}
    \tilde{G}(x,y,\lambda,\mu) = \int_0^\infty \frac{1}{\lambda}e^{\mu s}\left[q(x,y,s,\lambda) -e^{-\frac{f_+}{2}s}
    \psi_0(x, \lambda)\psi_0^*(y, \lambda^*)\right]ds,
\end{equation}
where $q(x,y,s,\lambda)$ is the heat kernel defined above in \eqref{eq:heat_kernel} when all the parameters are real, $E_0$ is the ground state energy \eqref{eq:ground_state_energy}, and $\psi_0(x)$ is the ground state eigenfunction \eqref{eq:ground_state_eigenfunction}. $\psi_0^*$ denotes the complex conjugate of $\psi_0$. Our purpose is to estimate $\tilde{G}(x,y,\lambda,\mu)$.

We also introduce the function:
\begin{equation} \label{eq:D_function_definition}
    D(s) =
    \begin{cases} 
        C\log\left(\frac{C}{s}\right), & \text{if }0<s<1\\
        C'\exp(-cs), & \text{if } s\geq1,
    \end{cases}
\end{equation}
with $C, C'$ large enough and $c$ small enough, taken so that $D(s)$ is continuous.

Our first main theorem controls the size of $\tilde{G}(x,y,\lambda,\mu)$:
\begin{thm}\label{thm:first main thm}
    For $(x,y,\lambda,\mu)\in\mathbb{R}^2\times\mathbb{R}^2\times\Omega\times\{
        \mu\in\mathbb{C}:\Re \mu<\frac{f_+}{2}+c_\star\}$ with $x\neq y$ and $|B|$ less than a small enough constant $c_B$, we have
    \begin{equation} \label{eq:first_main_thm}
        |\tilde{G}(x,y,\lambda,\mu)| \leq CD(c|\lambda|\ |x-y|(|x|+|y|)).
    \end{equation}
\end{thm}

\subsection{Proof of the First Main theorem} \label{sec:proof_of_first_main_thm}
We will establish the stronger inequality
\begin{equation}\label{eq:main_inequality}
    \int_0^\infty \Big|e^{\mu s}\big[q(x,y,s,\lambda) - e^{-\frac{f_+}{2}s}
    \psi_0(x)\psi_0^*(y)\big]\Big| ds \leq CD(c|\lambda|\ |x-y|(|x|+|y|)),
\end{equation}
for $x\neq y$, $\lambda\in\Omega$, and $\Re\mu<\replacedm{\frac{\lambda f_+}{2}}{\frac{f_+}{2}}+c_\star$, for small enough $c_\star$.

To prove \eqref{eq:main_inequality}, we partition the set of all $s>0$ into the following six subsets:
\begin{itemize}
    \item[] \textbf{Region 1:} $0 < s < c_1$, $s\geq\frac{|x-y|}{|x|+|y|}$, and $|\lambda|\left(|x|^2+|y|^2\right)s\geq1$.

    \item[] \textbf{Region 2:} $0 < s < c_1$, $s\geq\frac{|x-y|}{|x|+|y|}$, and $|\lambda|\left(|x|^2+|y|^2\right)s<1$.

    \item[] \textbf{Region 3:} $0 < s < c_1$, and $s<\frac{|x-y|}{|x|+|y|}$.

    \item[] \textbf{Region 4:} $c_1 < s < C_1$.

    \item[] \textbf{Region 5:} $s > C_1$, and $|\lambda|\left(|x|^2+|y|^2\right)\leq\exp(c^\sharp s)$.

    \item[] \textbf{Region 6:} $s > C_1$, and $|\lambda|\left(|x|^2+|y|^2\right)>\exp(c^\sharp s)$.
\end{itemize}
We will estimate the integral over each of these regions separately. Here, $c_1$ and $c^\sharp$ are small constants, and $C_1$ is a large constant. These constants will be picked later.

The following \replaced{props}{propositions} will be useful for the estimates in all regions:
\begin{prop}\label{prop:integral_estimates_1}
    Suppose $a,b>0$ with $ab>c$. Then
    \begin{equation}
        \int_b^\infty \frac{1}{s}\exp(-as) ds \leq C\exp(-a b).
    \end{equation}
\end{prop}
\begin{prop} \label{prop:integral_estimates_2}
    Suppose $a,b>0$. Then
    \begin{equation}
        \int_0^b \frac{1}{s}\exp(-a/s) ds \leq CD(a/b).
    \end{equation}
\end{prop}
The proofs of these two propositions are left to the reader.
\begin{prop} \label{prop:ground_state_estimate}
    For $\lambda\in\Omega$, $|B|<c_B$, and $\psi_0(x,\lambda)$ defined as in \eqref{eq:ground_state_eigenfunction}, we have
    \begin{equation}
        \bigg|\frac{1}{\lambda}\psi_0(x,\lambda)\psi_0^*(y,\lambda^*)\bigg| \leq C\exp(-c|\lambda|(|x|^2+|y|^2)).
    \end{equation}
\end{prop}
\begin{proof} \deleted{possible sign error re Im(lamda)*xi}
    \begin{align}
        &\bigg|\frac{1}{\lambda}\psi_0(x,\lambda)\psi_0^*(y,\lambda^*)\bigg| = \sqrt{\frac{\eta\zeta}{\pi^2}} \Bigg|\exp\left( -\frac{\lambda}{2}(\eta x_1^2+\zeta x_2^2-i\xi x_1x_2)-\frac{\lambda}{2}(\eta y_1^2+\zeta y_2^2+i\xi y_1y_2)\right)\Bigg|\notag\\
        &= \sqrt{\frac{\eta\zeta}{\pi^2}} \exp\Bigg( -\frac{\Re\lambda}{2}\Big(\eta (x_1^2 + y_1^2)+\zeta (x_2^2 + y_2^2)\Big)\notag\\
        &\qquad\qquad \replacedm{+}{-} \frac{\Im\lambda \xi}{2}\Big(x_1x_2-y_1y_2\Big)\Bigg)\notag\\
        &\leq C\exp\Big(-c\Re\lambda\big(|x|^2+|y|^2\big)\notag\\
        &\qquad\qquad\qquad \replacedm{-c'}{C}|\Im\lambda| |B|\big(|x_1x_2|+|y_1y_2|\big)\Big).
    \end{align}
    Now by Cauchy-Schwarz,
    \begin{equation}
        |x_1x_2|+|y_1y_2|\leq 2\sqrt{(x_1^2+y_1^2)(x_2^2+y_2^2)} \leq |x|^2+|y|^2.
    \end{equation}
    Thus, for $|B|<c_B$ small enough, we have
    \begin{align}
        |\psi_0(x,\lambda)\psi_0^*(y,\lambda)|
        &\leq \exp\left(-c(\Re\lambda-|B||\Im\lambda|)(|x|^2+|y|^2) \right)\notag\\
        &\leq C\exp(-c|\lambda|(|x|^2+|y|^2)),
    \end{align}
    provided $|B|<c_B$ is small enough.
\end{proof}

\subsubsection{Regions 1, 2, and 3} \label{subsec:regions_1_2_3}
The following \replaced{lem}{lemma} will be useful for the estimates in these regions:
\begin{lem}\label{lem:regions_1_2_3}
    For $|B|<c_B$ small enough, $\lambda\in\Omega$, and $0<s<c_1$, we have
    \begin{align}
        \Re{\phi(x,y,s,\lambda)} &\geq c|\lambda| s (|x|^2+|y|^2) +\frac{c|\lambda|}{s}|x-y|^2,\ \text{and}\notag\\
        |P(s)| &\leq \frac{C'}{s}.
    \end{align}
\end{lem}

\begin{proof}
    We first move to $w$ coordinates defined by
    \begin{align}
        w_1 &= x_1 - y_1, \qquad w_2 = x_2 - y_2,\notag\\
        w_3 &= x_1 + y_1, \qquad w_4 = x_2 + y_2.
    \end{align}
    The function $\phi(x,y,s,\lambda)$ can be rewritten as
    \begin{align} \label{eq:phi_w_coordinates}
        \phi(w,s) &= \frac{\lambda f_+f_- a_{12}(s)}{2K(s)}w_1^2 + \frac{\lambda f_+f_- a_{21}(s)}{2K(s)}w_2^2 \notag\\
        &+ \frac{\lambda f_+f_- b_{12}(s)}{2K(s)}w_3^2 + \frac{\lambda f_+f_- b_{21}(s)}{2K(s)}w_4^2 \notag\\
        &+ \frac{i\lambda B g_-(s)}{2K(s)}w_1w_4 +\frac{i\lambda B g_+(s)}{2K(s)}w_2w_3,
    \end{align}
    where $K(s),f_\pm$ are defined as in \eqref{eq:heat_kernel_parameters} and
    \begin{align}
        a_{ij}(s) &= k_i\left[\cosh\left(\frac{f_+s}{2}\right) + \cosh\left(\frac{f_-s}{2}\right)\right]\notag\\
        &\qquad\qquad \times\left[f_-(k_1+k_2)\sinh\left(\frac{f_+s}{2}\right)+ f_+(k_j-k_i)\sinh\left(\frac{f_-s}{2}\right)\right],\notag\\
        b_{ij}(s) &= k_i\left[\cosh\left(\frac{f_+s}{2}\right) - \cosh\left(\frac{f_-s}{2}\right)\right]\notag\\
        &\qquad\qquad \times\left[f_-(k_1+k_2)\sinh\left(\frac{f_+s}{2}\right)- f_+(k_j-k_i)\sinh\left(\frac{f_-s}{2}\right)\right],\notag\\
        g_\pm(s) &= (k_1^2-k_2^2)\left[f_+^2(\cosh(f_-s)-1)-f_-^2(\cosh(f_+s)-1)\right]\notag\\
        &\qquad\qquad \pm 8f_+f_-k_1k_2\sinh\left(\frac{f_+s}{2}\right)\sinh\left(\frac{f_-s}{2}\right).
    \end{align}
    Let us factor out the leading order behavior in small $s$:
    \begin{align} \label{eq:small_s_factors}
        a_{ij}(s) &= \frac{f_+f_-k_is}{2}\left[\cosh\left(\frac{f_+s}{2}\right)+ \cosh\left(\frac{f_-s}{2}\right)\right]\notag\\
        &\qquad\qquad \times\left[(k_1+k_2)\frac{\sinh(f_+s/2)}{f_+s/2}+(k_j-k_i)\frac{\sinh(f_-s/2)}{f_-s/2}\right]=s\cdot \hat{a}_{ij}(s),\notag\\
        b_{ij}(s) &= \frac{f_+f_-k_is^3}{2}\bigg[\frac{\cosh(f_+s/2)-\cosh(f_-s/2)}{s^2}\bigg]\notag\\
        &\qquad\qquad \times\bigg[(k_1+k_2)\frac{\sinh(f_+s/2)}{f_+s/2}+(k_j-k_i)\frac{\sinh(f_-s/2)}{f_-s/2}\bigg]=s^3\cdot \hat{b}_{ij}(s),\notag\\
        g_\pm(s) &= f_+^2f_-^2(k_1^2-k_2^2) s^2 \bigg[\frac{\cosh(f_-s)-1}{f_-^2s^2} - \frac{\cosh(f_+s)-1}{f_+^2s^2}\notag \\
        &\qquad\qquad \pm 2k_1k_2\frac{\sinh(f_+s/2)}{f_+s/2}\frac{\sinh(f_-s/2)}{f_-s/2}\bigg]=s^2\cdot \hat{g}_\pm(s),\notag\\
        K(s) &= f_+^2f_-^2s^2\left[(k_1+k_2)^2\left(\frac{\cosh(f_+s)-1}{f_+^2s^2}\right)-(k_1-k_2)^2\left(\frac{\cosh(f_-s)-1}{f_-^2s^2}\right)\right]= s^2\cdot \hat{K}(s),
    \end{align}
    and the real part of $\phi$ is
    \begin{align}
        \Re \phi(w,s) &= \frac{\Re\lambda f_+f_-}{2}\left[\frac{a_{12}(s)}{K(s)}w_1^2 + \frac{a_{21}(s)}{K(s)}w_2^2 + \frac{b_{12}(s)}{K(s)}w_3^2 + \frac{b_{21}(s)}{K(s)}w_4^2\right]\notag\\
        &+\frac{\Im\lambda B}{2}\left[\frac{g_-(s)}{K(s)}w_1w_4 + \frac{g_+(s)}{K(s)}w_2w_3\right].
    \end{align}

    First note that $f_\pm$ are positive constants and $f_+ > f_-$ by definition. Also, $\hat{a}_{ij}(s), \hat{b}_{ij}(s)$ and $\hat{K}(s)$ are positive and finite for small $s>0$. Indeed, $\cosh(f_+s/2)\pm\cosh(f_-s/2)\geq0$, and 
    \begin{align}
        &\frac{\sinh(f_\pm s/2)}{f_\pm s/2},\  \frac{\cosh(f_\pm s/2)-1}{f_\pm^2 s^2},\ \text{and } \frac{\cosh(f_+s/2)-\cosh(f_-s/2)}{s^2}\ \text{are finite and positive},\notag\\
        &\frac{\sinh(f_+ s/2)}{f_+ s/2}\geq \frac{\sinh(f_- s/2)}{f_- s/2},\qquad \frac{\cosh(f_+ s)-1}{f_+^2 s^2}\geq \frac{\cosh(f_- s)-1}{f_-^2 s^2},\ \text{and }\notag\\
        &k_1+k_2 > k_i-k_j\ (i,j=1,2),
    \end{align}
    for all $s\geq0$, where we define $\sinh(x)/x=1$ and $\frac{\cosh(x)-1}{x^2}=1/2$ for $x=0$.

    Therefore, $\hat{a}_{ij}(s), \hat{b}_{ij}(s)$ and $\hat{K}(s)$ are uniformly continuous in $(s,B)$ and positive on the compact set $\{s\in[0,c_1]\}\times \{|B|\leq c_B\}$ and thus attain a positive minimum, so that $a_{ij}(s)/K(s)\geq c/s$ and $b_{ij}(s)/K(s)\geq cs$. Similarly, $\hat{g}_\pm(s)$ is finite for small $s>0$ and uniformly continuous in $(s,B)$ on $\{s\in[0,c_1]\}\times \{|B|\leq c_B\}$, so that $|g_\pm(s)/K(s)| \leq \hat{C}$.

    Therefore, we have
    \begin{align}
        &\Re \phi \geq \Re\lambda\left[\frac{c}{s}(w_1^2+w_2^2)+cs(w_3^2+w_4^2)\right]-\frac{|\Im\lambda||B|}{2}\Big[C'|w_1w_4| + C'|w_2w_3|\Big]\notag\\
        &\geq \tilde{c}\Re\lambda\left[\frac{w_1^2+w_2^2}{s}+s(w_3^2+w_4^2)\right]-\hat{C}|\Im\lambda||B|\Big[|w_1w_4| + |w_2w_3|\Big].
    \end{align}
    Now by Cauchy-Schwarz, 
    \begin{equation}
        |w_1w_4| + |w_2w_3| \leq \sqrt{(w_1^2+w_2^2)(w_3^2+w_4^2)}\leq 2\sqrt{(w_1^2+w_2^2)(w_3^2+w_4^2)},
    \end{equation}
    and for any two positive numbers $N,M$, we have $2\sqrt{NM}\leq N+M$. Thus, for $N=\frac{w_1^2+w_2^2}{s}$ and $M=s(w_3^2+w_4^2)$, we have
    \begin{equation} \label{eq:cauchy_schwarz_application}
        |w_1w_4| + |w_2w_3| \leq \frac{w_1^2+w_2^2}{s} + s(w_3^2+w_4^2).
    \end{equation}
    Putting everything together, we have
    \begin{align}
        \Re \phi &\geq \left(\tilde{c}\Re\lambda - \hat{C}|\Im\lambda||B|\right)\left[\frac{w_1^2+w_2^2}{s} + s(w_3^2+w_4^2)\right]\notag\\
        &\geq c''|\lambda|\left[\frac{w_1^2+w_2^2}{s} + s (w_3^2+w_4^2)\right] = c''|\lambda|\left[\frac{|x-y|^2}{s} + s|x+y|^2\right]\notag\\
        &= c''|\lambda|\left[2s(|x|^2+|y|^2)+\left(\frac{1}{s}-s\right)|x-y|^2\right] \geq c|\lambda|\left(s(|x|^2+|y|^2)+\frac{|x-y|^2}{s}  \right), 
    \end{align}
    since $\frac{1}{s}-s\geq\frac{1}{2s}$ for $0<s<1/2$, so $c_1$ is chosen small enough. This holds provided $|B|<c_B$, for a small enough controlled constant $c_B$.

    Moving to the prefactor, we saw above that $\hat{K}(s)$ attains a positive minimum on $\{s\in[0,c_1]\}\times \{|B|\leq c_B\}$, so that $|K(s)|\geq c''s^2$. Remembering $P(s)=\frac{f_+f_-}{2\pi}\sqrt{\frac{2k_1k_2}{K(s)}}$, we get $|P(s)|\leq C'/s$. The proof of the \replaced{lem}{lemma} is thus complete.
\end{proof}

\paragraph{Region 1}
\cref{lem:regions_1_2_3} applies in this region, and moreover
\begin{equation}
    c_1\geq s\geq \max\left(\frac{|x-y|}{|x|+|y|},\frac{1}{|\lambda|(|x|^2+|y|^2)}\right).
\end{equation}
Therefore, in Region 1 we have
\begin{equation}
    \bigg|\frac{e^{\mu s}}{\lambda}q(x,y,s,\lambda)\bigg| \leq \frac{C}{s}\exp\left(-cs|\lambda| (|x|^2+|y|^2)\right),
\end{equation}
and also by \cref{prop:ground_state_estimate},
\begin{equation} \label{eq:counterterm_region_1}
    \bigg|\frac{e^{\mu s}}{\lambda}\big[e^{-\frac{f_+}{2}s}
    \psi_0(x,\lambda)\psi_0^*(y,\lambda^*) \big]\bigg| \leq C \exp\left(-c|\lambda|(|x|^2+|y|^2)\right).
\end{equation}
These estimates and \cref{prop:integral_estimates_1} tell us that
\begin{align}
    &\int_\text{Region 1}\Big|\frac{e^{\mu s}}{\lambda}\big[q(x,y,s,\lambda) - e^{-\frac{f_+}{2}s}
    \psi_0(x,\lambda)\psi_0^*(y,\lambda^*)\big]\Big| ds \notag\\
    &\leq \int_{\max\left(\frac{|x-y|}{|x|+|y|},\frac{1}{|\lambda|(|x|^2+|y|^2)}\right)}^\infty \frac{C}{s}\exp\left(-cs|\lambda| (|x|^2+|y|^2)\right) ds \notag\\
    &\leq C\exp\left(-c|\lambda| (|x|^2+|y|^2)\cdot \max\left(\frac{|x-y|}{|x|+|y|},\frac{1}{|\lambda|(|x|^2+|y|^2)}\right)\right)\notag \\
    &\leq C \exp\left(-c|\lambda|(|x|^2+|y|^2)\cdot \frac{|x-y|}{|x|+|y|}\right) \notag \\
    &\leq C' \exp\left(-c'|\lambda||x-y|\cdot (|x|+|y|)\right) \leq C'' D\left(c|\lambda||x-y|\cdot (|x|+|y|)\right).
\end{align}
This controls the integral over Region 1.

\paragraph{Region 2}
Again, \cref{lem:regions_1_2_3} applies in this region, and we have
\begin{equation} \label{eq:region_2_bounds_on_s}
    \frac{|x-y|}{|x|+|y|} \leq s \leq \frac{1}{|\lambda|(|x|^2+|y|^2)},
\end{equation}
hence   $\lambda|x-y||x+y|=\lambda\frac{|x-y|}{|x+y|}\ |x+y|^2 \le C\ \lambda\ s\ (|x|^2+|y|^2)\le C$,  and so
\begin{equation}
    D(\lambda||x-y|(|x|+|y|)) = C \log\left(\frac{C}{|\lambda||x-y|(|x|+|y|)}\right).
\end{equation}
Also $s<c_1$ in this region. Then thanks to \cref{lem:regions_1_2_3} and \eqref{eq:region_2_bounds_on_s}, we have
\begin{align} \label{eq:region_2_estimate_on_q}
    &\int_\text{Region 2}\bigg|\frac{e^{\mu s}}{\lambda}q(x,y,s,\lambda) \bigg|ds \leq C \int_{\frac{|x-y|}{|x|+|y|}}^{|\lambda|^{-1}(|x|^2+|y|^2)^{-1}}\frac{1}{s}ds \leq C\log\left(\frac{|x|+|y|}{\lambda|x-y|(|x|^2+|y|^2)}\right)\notag \\
    &\leq C'\log\left(\frac{C}{|\lambda||x-y|(|x|+|y|)}\right) \leq C'' D\left(c|\lambda||x-y|(|x|+|y|)\right).
\end{align}
Also, using \cref{prop:ground_state_estimate}, we have
\begin{align} \label{eq:counterterm_region_2}
    &\int_\text{Region 2} \bigg|\frac{e^{\mu s}}{\lambda}[-e^{-\frac{f_+}{2}s}
    \psi_0(x,\lambda)\psi_0^*(y,\lambda^*)]\bigg| ds \leq C \frac{1}{|\lambda|}\ \int_0^{c_1} \big|\psi_0(x,\lambda)\psi_0^*(y,\lambda^*)\big|ds \notag \\
    &\leq C^\sharp \exp\left(-c|\lambda|(|x|^2+|y|^2)\right)\leq C''D\left(c|\lambda||x-y|(|x|+|y|)\right).
\end{align}
From \eqref{eq:region_2_estimate_on_q} and \eqref{eq:counterterm_region_2}, we see that
\begin{equation}
    \int_\text{Region 2}\bigg|\frac{e^{\mu s}}{\lambda}\big[q(x,y,s,\lambda) - e^{-\frac{f_+}{2}s}
    \psi_0(x,\lambda)\psi_0^*(y,\lambda^*)\big]\bigg|ds \leq CD\left(c|\lambda||x-y|(|x|+|y|)\right).
\end{equation}
This controls the relevant integral over Region 2.

\paragraph{Region 3}
In this region \cref{lem:regions_1_2_3} applies, and $0<s<\frac{|x-y|}{|x|+|y|}$; also $0<s<c_1$. So 
\begin{align}
    &\int_\text{Region 3}\bigg|\frac{e^{\mu s}}{\lambda} q(x,y,s,\lambda) \bigg|ds \leq \int_0^{\frac{|x-y|}{|x|+|y|}} \frac{C}{s}\exp\left(-\frac{c}{s}|\lambda| |x-y|^2\right) ds \notag\\
    &\leq C D\left(\frac{c|\lambda||x-y|^2}{\frac{|x-y|}{|x|+|y|}}\right) = CD\left(c|\lambda||x-y|(|x|+|y|)\right).
\end{align}
Also, using \cref{prop:ground_state_estimate}, we have
\begin{align} \label{eq:counterterm_region_3}
    &\int_\text{Region 3} \bigg|\frac{e^{\mu s}}{\lambda}[-e^{-\frac{f_+}{2}s}
    \psi_0(x,\lambda )\psi_0^*(y,\lambda^*)]\bigg| ds \leq C\ \frac{1}{|\lambda|}\ \big|\psi_0(x,\lambda)\psi_0^*(y,\lambda^*)\big|\notag\\
    &\leq C\exp\left(-c|\lambda|(|x|^2+|y|^2)\right)
    \leq C'D\left(c|\lambda||x-y|(|x|+|y|)\right). 
\end{align}
From these two estimates, we see that
\begin{equation}
    \int_\text{Region 3}\bigg|\frac{e^{\mu s}}{\lambda}\big[q(x,y,s,\lambda) - e^{-\frac{f_+}{2}s}
    \psi_0(x,\lambda)\psi_0^*(y,\lambda^*)\big]\bigg|ds \leq C''D\left(c|\lambda||x-y|(|x|+|y|)\right).
\end{equation}
This controls the relevant integral over Region 3.

\subsubsection{Region 4} \label{subsec:region_4}
In this region we have $c_1<s<C_1$. The following \replaced{lem}{lemma} will be useful for the estimates in this region:
\begin{lem} \label{lem:region_4}
    For $c_1\leq s\leq C_1$, and $|B|<c_B$ small enough, we have
     \deleted{CHECK}
    \begin{equation}
        \text{Re}(\phi(x,y,s,\lambda))\geq c|\lambda| \left(|x|^2+|y|^2\right),\quad \text{and}\quad |P(s)| \leq C.
    \end{equation}
\end{lem}

\begin{proof}
    We first set $B=0$. In this case, we have
    \begin{align}
        \phi_{B=0}(w,s) &= \frac{\lambda k_1}{4}\coth\left(\frac{k_1s}{2}\right)w_1^2 + \frac{\lambda k_2}{4}\coth\left(\frac{k_2s}{2}\right)w_2^2 \notag\\
        &+ \frac{\lambda k_1}{4}\tanh\left(\frac{k_1s}{2}\right)w_3^2 + \frac{\lambda k_2}{4}\tanh\left(\frac{k_2s}{2}\right)w_4^2.
    \end{align}
    For $c_1<s<C_1$, $\coth$ and $\tanh$ are continuous and bounded from below by a positive constant, so that
    \begin{equation}
        \text{Re}(\phi_{B=0}(x,y,s))\geq c'|\lambda|\left(w_1^2+w_2^2+w_3^2+w_4^2 \right) = 2c'|\lambda| \left(|x|^2+|y|^2\right).
    \end{equation}
    Now we move to $|B|$ small. We remind the reader that $f_\pm=\sqrt{(k_1\pm k_2)^2+4B^2}\to|k_1\pm k_2|$ as $B\to0$, and $a_{ij},b_{ij},g_\pm,K$ depend on $B$ only through $f_\pm$. Specifically,
    \begin{align}
        K(s)|_{B=0} &= (k_1^2-k_2^2)^2\big[\cosh\big((k_1+k_2)s\big)-\cosh\big((k_1-k_2)s\big)\big] \notag\\
        &= 2(k_1^2-k_2^2)^2 \sinh(k_1s)\sinh(k_2s) > 0,
    \end{align}
    for $c_1\leq s\leq C_1$. By uniform continuity in $(s,B)$ on the compact set $[c_1,C_1]\times\{|B|\leq c_B\}$, there exists $\delta>0$ such that whenever $|B|<\delta$, we have, for some $c$,
    \begin{equation} \label{eq:uniform_bound_on_K}
        K(s) \geq c > 0,
    \end{equation}
    uniformly for $c_1\leq s\leq C_1$. We now reduce $c_B$ so that $c_B <\delta$. Hence \eqref{eq:uniform_bound_on_K} holds for $(s,B)\in[c_1,C_1]\times\{|B|\leq c_B\}$, so no singularities arise in the denominators. Therefore, we have uniform continuity in $(s,B)$ on the compact set $[c_1,C_1]\times\{|B|\leq C_B\}$ for the following quantities:
    \begin{equation}
        \frac{f_+f_-a_{ij}(s)}{2K(s)},\ \frac{f_+f_-b_{ij}(s)}{2K(s)},\ \frac{Bg_\pm(s)}{2K(s)},
    \end{equation}
    and in particular, they are all uniformly bounded.

    Thus, there exists $\tilde{c}>0$ such that for $|B|<\tilde{c}$, we have
    \begin{equation}
        |\phi(w,s) - \phi_{B=0}(w,s)| \leq \frac{c'}{2}|\lambda| (w_1^2+w_2^2+w_3^2+w_4^2),
    \end{equation}
    so that
    \begin{align}
        \text{Re}(\phi(w,s)) &\geq \text{Re}(\phi_{B=0}(w,s)) - |\phi(w,s) - \phi_{B=0}(w,s)| \notag\\
        &\geq \frac{c'}{2}|\lambda| (w_1^2+w_2^2+w_3^2+w_4^2) = c'|\lambda|(|x|^2+|y|^2).
    \end{align}
    We again reduce $c_B$ so that $c_B < \tilde{c}$.

    We now move to the prefactor $P(s)$. Since $K(s)$ is bounded away from zero for $c_1\leq s\leq C_1$ and $|B|<c_B$, we have $|P(s)|\leq C$ for some constant $C>0$. This completes the proof of the \replaced{lem}{lemma}.
\end{proof}

We now apply \cref{lem:region_4}, so
\begin{align}\label{eq:region_4_estimate_on_q}
    &\int_\text{Region 4}\bigg|\frac{e^{\mu s}}{\lambda} q(x,y,s,\lambda) \bigg|ds \leq C\int_{c_1}^{C_1}|e^{\mu s}|\exp(-c|\lambda|(|x|^2+|y|^2)) ds \notag\\
    &\leq C' \exp(-c|\lambda|(|x|^2+|y|^2)) \leq C'\exp(-c''|\lambda||x-y|(|x|+|y|)) \notag\\
    &\leq C'' D(c|\lambda||x-y|(|x|+|y|)).
\end{align}
Moreover by \cref{prop:ground_state_estimate},
\begin{align}\label{eq:region_4_estimate_on_ground_state}
    &\int_\text{Region 4} \bigg|\frac{e^{\mu s}}{\lambda}[-e^{-\frac{f_+}{2}s}
    \psi_0(x,\lambda )\psi_0^*(y,\lambda^*)]\bigg| ds \leq C\bigg|\frac{1}{\lambda}\psi_0(x,\lambda)\psi_0^*(y,\lambda^*)\bigg|\notag\\
    &\leq C\exp\left(-c|\lambda|(|x|^2+|y|^2)\right)
    \leq C'D\left(c|\lambda||x-y|(|x|+|y|)\right). 
\end{align}
From \eqref{eq:region_4_estimate_on_q} and \eqref{eq:region_4_estimate_on_ground_state}, we see that
\begin{equation}
    \int_\text{Region 4}\bigg|\frac{e^{\mu s}}{\lambda}\big[q(x,y,s,\lambda) - e^{-\frac{f_+}{2}s}
    \psi_0(x,\lambda)\psi_0^*(y,\lambda^*)\big]\bigg|ds \leq C''D\left(c|\lambda||x-y|(|x|+|y|)\right).
\end{equation}
This controls the relevant integral over Region 4.

\subsubsection{Region 5} \label{subsec:region_5}
In this region we have $s>C_1$ and $|\lambda|\left(|x|^2+|y|^2\right)\leq\exp(c^\sharp s)$, where $c^\sharp$ will be picked below. The following \replaced{lem}{lemma} will be useful for the estimates in this region:
\begin{lem}\label{lem:region_5}
    For $\lambda\in\Omega,\ |B|<c_B$ and $s>C_1$, we write 
    \begin{align}
        &\phi(x,y,s,\lambda) = \lim_{s\to\infty}\phi(x,y,s,\lambda) + \text{ERROR}_1(x,y,s,\lambda)\notag\\
        &P(s) = C_P\exp\Big(-\frac{f_+}{2}s\Big)\big(1+\text{ERROR}_2(s)\big),
    \end{align}
    where $C_P$ is the same prefactor that appears in the counter-term $e^{-\frac{f_+}{2}s}
    \psi_0(x,\lambda)\psi_0^*(y,\lambda^*)$, and
    \begin{align}
        |\text{ERROR}_1(x,y,s,\lambda)| &\leq C \exp(-c s)|\lambda|\left(|x|^2+|y|^2\right),\ \text{and}\notag\\
        |\text{ERROR}_2(s)| &\leq C \exp(-c s),
    \end{align}
    for $s\geq C_1$ large enough, and some constants $C,c>0$. The reader should note that \[ \phi_\infty(x,y,\lambda):=\lim_{s\to\infty}\phi(x,y,s,\lambda)\] is the exponent in the ground state $\psi_0(x,\lambda)\psi_0^*(y,\lambda^*)$.
\end{lem}

\begin{proof}
    For the first estimate, we write the explicit expression for the limit of $\phi(x,y,s,\lambda)$ as $s\to\infty$:
    \begin{equation}
        \lim_{s\to\infty}\phi(x,y,s,\lambda) = \frac{\lambda f_+k_1}{2(k_1+k_2)}(x_1^2+y_1^2)
        + \frac{\lambda f_+k_2}{2(k_1+k_2)}(x_2^2+y_2^2)
        - \frac{i\lambda B (k_1-k_2)}{k_1+k_2}(x_1x_2-y_1y_2).
    \end{equation}
    Then we have
    \begin{equation}
    \begin{aligned}
        \text{ERROR}_1=\lambda\Big[&M_1(x_1^2+y_1^2) + M_2x_1y_1
        + M_3(x_2^2+y_2^2) + M_4x_2y_2\\
        &+M_5(x_1x_2-y_1y_2) + M_6(x_1y_2-x_2y_1)\Big]
    \end{aligned}
    \end{equation}
    for $M_i=M_i(B,k_1,k_2)$.
    
    We want to show that
    \begin{equation}\label{eq:matrix_size_region_5}
        |M_i|\leq C \exp(-c s),
    \end{equation}
    for some constants $C,c>0$ and large enough $s$. Indeed, if so, then
    \begin{equation}
        |\text{ERROR}_1| \leq |\lambda|C\exp(-c s)(|x|^2+|y|^2 + |x_1y_1| + |x_2y_2| + |x_1x_2| + |y_1y_2| + |x_1y_2| + |x_2y_1|).
    \end{equation}
    By Cauchy-Schwarz, we have
    \begin{align}
        |x_1y_1| + |x_2y_2| &\leq 2\sqrt{|x|^2+|y|^2} \leq |x|^2+|y|^2,\notag\\
        |x_1y_2| + |x_2y_1| &\leq 2\sqrt{|x|^2+|y|^2} \leq |x|^2+|y|^2, \notag\\
        |x_1x_2| + |y_1y_2| &\leq 2\sqrt{(x_1^2+y_1^2)(x_2^2+y_2^2)} \leq |x|^2+|y|^2,
    \end{align}
    and we get the desired result
    \begin{equation}
        |\text{ERROR}_1| \leq C'|\lambda|\exp(-c s)(|x|^2+|y|^2).
    \end{equation}

    By definition we have $f_+>f_->0$. Then 
    \begin{align}
        \sinh(f_+s) &= \frac{1}{2}e^{f_+s}\left(1-e^{-2f_+s}\right),\notag\\
        |\sinh(f_-s)| &\leq e^{f_-s} = \frac{1}{2}e^{f_+s}\left(2e^{-(f_+-f_-)s}\right).
    \end{align}
    Therefore, for large enough $s$, we have
    \begin{equation}
        \alpha_{ij}(s) = k_if_-(k_1+k_2)\frac{1}{2}e^{f_+s}\left(1+\mathcal{O}(e^{-(f_+-f_-)s})\right).
    \end{equation}
    Similarly, for large enough $s$, we have
    \begin{align}
        \sinh\left(\frac{f_-+f_+}{2}s\right) &= \frac{1}{2}e^{f_+s}\left[e^{-(f_+-f_-)s/2}-e^{-(3f_++f_-)s/2}\right] = e^{f_+s}\mathcal{O}\left(e^{-(f_+-f_-)s/2}\right),\notag\\
        \sinh\left(\frac{f_--f_+}{2}s\right) &= \frac{1}{2}e^{f_+s}\left[e^{-(3f_+-f_-)s/2}-e^{-(f_++f_-)s/2}\right] = e^{f_+s}\mathcal{O}\left(e^{-(f_++f_-)s/2}\right).
    \end{align}
    Then,
    \begin{equation}
        \beta_{ij}(s) = e^{f_+s}\mathcal{O}(e^{-(f_+-f_-)s/2}).
    \end{equation}
    Also,
    \begin{align}
        \cosh(f_+s)-1 &= \frac{1}{2}e^{f_+s}\left(1-e^{-f_+s}\right)^2,\notag\\
        \cosh(f_-s)-1 &= e^{f_+s}\mathcal{O}(e^{-(f_+-f_-)s}).
    \end{align}
    Then,
    \begin{align}
        \gamma_1(s) &= (k_1^2-k_2^2)\left[f_+^2e^{f_+s}\mathcal{O}\left(e^{-(f_+-f_-)s}\right) - f_-^2\frac{1}{2}e^{f_+s}\left(1-e^{-f_+s}\right)^2\right] ,\notag\\
        &= e^{f_+s}\left[-\frac{f_-^2(k_1^2-k_2^2)}{2}+\mathcal{O}\left(e^{-(f_+-f_-)s}\right)\right],\notag\\
        K(s) &= f_-^2(k_1+k_2)^2\frac{1}{2}e^{f_+s}\left(1+\mathcal{O}(e^{-(f_+-f_-)s})\right) -f_+^2(k_1-k_2)^2e^{f_+s}\mathcal{O}(e^{-(f_+-f_-)s})\notag\\
        &= e^{f_+s}\left[\frac{f_-^2(k_1+k_2)^2}{2} + \mathcal{O}(e^{-(f_+-f_-)s})\right].
    \end{align}
    And finally, 
    \begin{align}
        \sinh\left(\frac{f_+s}{2}\right)\sinh\left(\frac{f_-s}{2}\right)
        &= \frac{1}{4}e^{f_+s}\Big[e^{-(f_+-f_-)s/2}-e^{-(f_++f_-)s/2}\notag\\
        &\qquad\qquad{}-e^{-(3f_+-f_-)s/2}+e^{-(3f_++f_-)s/2}\Big]\notag\\
        &= e^{f_+s}\mathcal{O}\left(e^{-(f_+-f_-)s/2}\right),
    \end{align}
    so that we have
    \begin{equation}
        \gamma_2(s) = e^{f_+s}\mathcal{O}\left(e^{-(f_+-f_-)s/2}\right).
    \end{equation}

    For the $x_1^2+y_1^2$ term, we have
    \begin{align}
        |M_1| &= \bigg|\frac{f_+f_- \alpha_{12}(s)}{2K(s)} - \frac{f_+k_1}{2(k_1+k_2)}\bigg|\\
        &= \bigg|\frac{f_+f_-}{2} \frac{k_1f_-(k_1+k_2)}{f_-^2(k_1+k_2)^2}\frac{1+\mathcal{O}(e^{-(f_+-f_-)s})}{1 + \mathcal{O}(e^{-(f_+-f_-)s})} - \frac{f_+k_1}{2(k_1+k_2)}\bigg|\\
        &= \bigg|\frac{f_+k_1}{2(k_1+k_2)}\left(1 + \mathcal{O}(e^{-(f_+-f_-)s})\right) - \frac{f_+k_1}{2(k_1+k_2)}\bigg| \notag\\
        &= |\mathcal{O}(e^{-(f_+-f_-)s})|\leq C\exp(-c s).
    \end{align}
    The same result applies for $M_3$ (the $x_1y_1$ term), since they are defined in the same way as $M_1$, but with $k_1$ and $k_2$ interchanged.

    For the $x_2^2+y_2^2$ term, we have
    \begin{align}
        |M_2| &= \bigg|\frac{f_+f_- \beta_{12}(s)}{K(s)}\bigg| = \Bigg|\frac{f_+f_- e^{f_+s}\mathcal{O}(e^{-(f_+-f_-)s})}{e^{f_+s}\left[\frac{f_-^2(k_1+k_2)^2}{2} + \mathcal{O}(e^{-(f_+-f_-)s})\right]}\Bigg| \notag\\
        &= \bigg|\mathcal{O}(e^{-(f_+-f_-)s}) \bigg| \leq C\exp(-c s).
    \end{align}
    The same result applies for $M_4$ (the $x_2y_2$ term), since they are defined in the same way as $M_2$, but with $k_1$ and $k_2$ interchanged.

    Next is the cross term $x_1x_2-y_1y_2$. For this case we have
    \begin{align}
        |M_5| &= \bigg|\frac{iB \gamma_1(s)}{K(s)} + \frac{iB (k_1-k_2)}{k_1+k_2}\bigg| = |B|\Bigg| \frac{e^{f_+s}\left[-\frac{f_-^2(k_1^2-k_2^2)}{2}+\mathcal{O}\left(e^{-(f_+-f_-)s}\right)\right]}{e^{f_+s}\left[\frac{f_-^2(k_1+k_2)^2}{2} + \mathcal{O}(e^{-(f_+-f_-)s})\right]} + \frac{k_1-k_2}{k_1+k_2}\Bigg|\notag\\
        &= |B| \Bigg| \frac{-f_-^2(k_1^2-k_2^2)}{f_-^2(k_1+k_2)^2} \frac{1+\mathcal{O}(e^{-(f_+-f_-)s})}{1+\mathcal{O}(e^{-(f_+-f_-)s})} + \frac{k_1-k_2}{k_1+k_2} \Bigg|\notag\\
        &= |B| \Bigg|-\frac{k_1-k_2}{k_1+k_2}\left(1 + \mathcal{O}(e^{-(f_+-f_-)s})\right) + \frac{k_1-k_2}{k_1+k_2}\Bigg|\notag\\
        &= |B||\mathcal{O}(e^{-(f_+-f_-)s})| \leq C\exp(-c s),
    \end{align}
    provided $|B|$ is bounded.

    Finally, we have the last cross term $x_1y_2-x_2y_1$. For this case we have
    \begin{align}
        |M_6| &= \bigg|-\frac{iB \gamma_2(s)}{K(s)}\bigg| = |B| \Bigg| \frac{e^{f_+s}\mathcal{O}\left(e^{-(f_+-f_-)s/2}\right)}{e^{f_+s}\left[\frac{f_-^2(k_1+k_2)^2}{2} + \mathcal{O}(e^{-(f_+-f_-)s})\right]} \Bigg|\notag\\
        &= |B| \big|\mathcal{O}\left(e^{-(f_+-f_-)s/2}\right) \big| \leq C\exp(-c s),
    \end{align}
    again provided $|B|$ is bounded. This completes the proof of \eqref{eq:matrix_size_region_5}.

    For the prefactor error estimate, we use the above asymptotics for $K(s)$ to see that
    \begin{align}
        P(s) &= \frac{f_+f_-}{2\pi}\sqrt{\frac{2k_1k_2}{K(s)}} = \frac{f_+f_-}{2\pi}\sqrt{\frac{2k_1k_2}{e^{f_+s}\left[\frac{f_-^2(k_1+k_2)^2}{2} + \mathcal{O}(e^{-(f_+-f_-)s})\right]}}\notag\\
        &= \frac{f_+f_-}{2\pi}\sqrt{\frac{4k_1k_2}{f_-^2(k_1+k_2)^2}}e^{-f_+s/2}\left(1+\mathcal{O}(e^{-(f_+-f_-)s})\right)\notag\\
        &= C_P\exp\Big(-\frac{f_+}{2}s\Big)\big(1+\text{ERROR}_2(s)\big),
    \end{align}
    where one checks \added{that} the $C_P$ is the same prefactor that appears in the counter-term $$e^{-\frac{f_+}{2}s}
    \psi_0(x,\lambda)\psi_0^*(y,\lambda^*).$$ This completes the proof of the \replaced{lem}{lemma}.
\end{proof}

Remember that in region 5 we have $|\lambda|(|x|^2+|y|^2)\leq\exp(c^\sharp s)$. Therefore, by \cref{lem:region_5} and picking $c^\sharp$ small enough, we have
\begin{equation}
    |\text{ERROR}_1(x,y,s,\lambda)| \leq Ce^{-cs} \exp(c^\sharp s)\leq C\exp(-c's).
\end{equation}
Then
\begin{align}
    &P(s)e^{-\phi(x,y,s,\lambda)} = C_Pe^{-\frac{f_+}{2}s}
    [1+\text{ERROR}_2(s)]e^{-\phi_\infty(x,y,s)-\text{ERROR}_1(x,y,s,\lambda)}\notag\\
    &= C_P e^{-\frac{f_+}{2}s}
    [1+\text{ERROR}_2(s)]e^{-\phi_\infty(x,y,s)}e^{-\text{ERROR}_1(x,y,s,\lambda)}\notag\\
    &= C_P e^{-\frac{f_+}{2}s}
    e^{-\phi_\infty(x,y,s)}\big[1+\text{ERROR}_3\big],
\end{align}
with $|\text{ERROR}_3|\leq C\exp(-c's)$. On the other hand, we have
\begin{equation}
    \frac{1}{\lambda} e^{-\frac{f_+}{2}s} 
    \psi_0(x,\lambda)\psi_0^*(y,\lambda^*) = C_P e^{-\frac{f_+}{2}s} e^{-\phi_\infty(x,y,s)}.
\end{equation}

We now apply \cref{lem:region_5}. We have
\begin{align}
    &\int_\text{Region 5}\bigg|\frac{e^{\mu s}}{\lambda}\big[q(x,y,s,\lambda) - e^{-\frac{f_+}{2}s}
    \psi_0(x,\lambda)\psi_0^*(y,\lambda^*)\big]\bigg|ds \notag \\
    &\qquad\qquad\leq \int_{C_1}^\infty C\big|e^{(\mu-\frac{f_+}{2}-c')s} \big| \bigg|\frac{1}{\lambda}\psi_0(x,\lambda)\psi_0^*(y,\lambda^*)\bigg| ds \notag\\
    &\qquad\qquad\leq C' \bigg|\frac{1}{\lambda}\psi_0(x,\lambda)\psi_0^*(y,\lambda^*)\bigg|,
\end{align}
provided we pick $\replacedm{c_\mu}{c_\star}$ small enough in our hypothesis $\Re\mu<\frac{f_+}{2}+c_\star$.\\
Continuing with the estimates, and applying \cref{prop:ground_state_estimate}, we get
\begin{align}
    &\int_\text{Region 5}\Big|e^{\mu s}\big[q(x,y,s,\lambda) -e^{-\frac{f_+}{2}s}
    \psi_0(x,\lambda)\psi_0^*(y,\lambda^*)\big]\Big|ds \leq C' \exp(-c|\lambda|(|x|^2+|y|^2))\notag\\
    &\leq C'' D(c|\lambda||x-y|(|x|+|y|)).
\end{align}
This controls the relevant integral over Region 5.

\subsubsection{Region 6} \label{subsec:region_6}
In this region we have $s>C_1$ and $|\lambda|\left(|x|^2+|y|^2\right)>\exp(c^\sharp s)$. The following \replaced{lem}{lemma} will be useful for the estimates in this region:
\begin{lem} \label{lem:region_6}
    Let $s>C_1$ (large enough) and $|\lambda|\left(|x|^2+|y|^2\right)>\exp(c^\sharp s)$, then 
    \begin{align}
        \text{Re}\big(\phi(x,y,s,\lambda)\big) &\geq c|\lambda|(|x|^2+|y|^2) \geq \frac{1}{2}c|\lambda|(|x|^2+|y|^2) + \frac{1}{2}c\exp(c^\sharp s), \notag\\
        |P(s)| &\leq C e^{-\frac{f_+}{2}s} .
    \end{align}
\end{lem}
\begin{proof}
    First note that $\Re\phi_\infty(x,y,s) \geq c|\lambda|(|x|^2+|y|^2)$. Indeed, 
    \begin{align}
        &\Re{\phi_\infty}(x,y,\lambda) = \Re\lambda\big[\eta(x_1^2+y_1^2)+\zeta(x_2^2+y_2^2)\big]-\Im\lambda \xi(x_1x_2-y_1y_2)\notag\\
        &\geq c''\Re\lambda(|x|^2+|y|^2) - |\Im\lambda||B|C'(|x_1x_2|+|y_1y_2|)\notag\\
        &= c^\sharp(\Re\lambda - C|\Im\lambda||B|)(|x|^2+|y|^2) \geq c|\lambda|(|x|^2+|y|^2),
    \end{align}
    provided $|\Im\lambda|<C'|\Re\lambda|$ and $|B|$ is small enough (remember $\eta,\zeta>0$). Now, by \cref{lem:region_5}, we have
    \begin{align}
        &\Re{\phi}(x,y,s) = \Re{\phi_\infty}(x,y,\lambda) + \Re{\textrm{ERROR}_1}(x,y,s,\lambda) \notag\\
        &\geq c|\lambda|(|x|^2+|y|^2) - |\textrm{ERROR}_1(x,y,s,\lambda)| \notag\\
        &\geq c|\lambda|(|x|^2+|y|^2) \geq \frac{1}{2}c|\lambda|(|x|^2+|y|^2) + \frac{1}{2}c\exp(c^\sharp s),
    \end{align}

    where we used the region 6 hypothesis $|\lambda|(|x|^2+|y|^2)>\exp(c^\sharp s)$. Picking $c^\sharp<c$, we have the desired result. 
    
    For the prefactor estimate, we have shown in \cref{lem:region_5} that 
    \begin{equation}
        P(s) = C_P\exp\big(-\frac{f_+}{2} s\big)\big(1+\text{\rm ERROR}_2(s)\big),\ \text{with}\ |\textrm{ERROR}_2(s)|\leq C''\exp(-c s),
    \end{equation}
    and it follows that $|P(s)|\leq C e^{-\frac{f_+}{2}s}$. The proof of the \replaced{lem}{lemma} is complete.
\end{proof}

We now apply the above \cref{lem:region_6}:
\begin{align} \label{eq:region_6_estimate_on_q}
    &\int_\text{Region 6}\bigg|\frac{e^{\mu s}}{\lambda} q(x,y,s,\lambda) \bigg|ds \leq \int_{C_1}^\infty |e^{(\mu -\frac{f_+}{2})s}|e^{-\frac{c|\lambda|}{2}(|x|^2+|y|^2)}e^{-\frac{c}{2}\exp(c^\sharp s)} ds \notag\\
    &\leq C\exp\left(-\frac{c|\lambda|}{2}(|x|^2+|y|^2)\right) \leq C'D\big(c|\lambda||x-y|(|x|+|y|)\big).
\end{align}
Moreover, by \cref{prop:ground_state_estimate}, we have in region 6 that\\
\begin{align}
    \bigg|\frac{e^{\mu s}}{\lambda}[-e^{-\frac{f_+}{2}s}
    \psi_0(x,\lambda )\psi_0^*(y,\lambda^*)]\bigg| &= \big|
    e^{(\mu-\frac{f_+}{2})s} \big|\bigg|\frac{1}{\lambda}\psi_0(x,\lambda)\psi_0^*(y,\lambda^*)\bigg|\notag\\
    &\leq c\big|e^{(\mu-\frac{f_+}{2})s}\big|\exp\left(-c|\lambda|(|x|^2+|y|^2)\right)\notag\\
    &\leq c\big|e^{(\mu-\frac{f_+}{2})s}\big| \exp\left(-\frac{c}{2}|\lambda|(|x|^2+|y|^2) - \frac{c}{2}\exp(c^\sharp s)\right),\notag
\end{align}
so that
\begin{align} \label{eq:region_6_estimate_on_ground_state}
    &\int_\text{Region 6}\Big|e^{\mu s}\big[- \frac{e^{-\frac{f_+}{2} s}}{\lambda}\psi_0(x,\lambda)\psi_0^*(y,\lambda^*)\big]\Big|ds\notag\\
    &\leq C \int_{C_1}^\infty \big|e^{(\mu -\frac{f_+}{2})s}\big| \exp\left(-\frac{c}{2}\exp(c^\sharp s) \right)ds \cdot \exp\left(-\frac{c|\lambda|}{2}(|x|^2+|y|^2) \right)\notag\\
    &= C' \exp\left(-\frac{c|\lambda|}{2}(|x|^2+|y|^2)\right) \leq C''D\big(c|\lambda||x-y|(|x|+|y|)\big).
\end{align}
From \eqref{eq:region_6_estimate_on_q} and \eqref{eq:region_6_estimate_on_ground_state}, we see that
\begin{equation}
    \int_\text{Region 6}\bigg|\frac{e^{\mu s}}{\lambda}\big[q(x,y,s,\lambda) - e^{-\frac{f_+}{2} s}
    \psi_0(x,\lambda)\psi_0^*(y,\lambda^*)\big]\bigg|ds \leq C''D\big(c|\lambda||x-y|(|x|+|y|)\big).
\end{equation}
This controls the relevant integral over Region 6.

We have controlled the integral over all the regions $1,\cdots,6$ of the integral in \eqref{eq:first_main_thm}. All those integrals have been shown to be at most $C D(c|\lambda||x-y|(|x|+|y|))$. Then \eqref{eq:first_main_thm} holds. The proof of the first main theorem is now complete.

\subsection{Second Main theorem} \label{sec:second_main_thm}

 For the second main theorem, the basic constants are simply $k_1,k_2$ and $C_B$ which will play a role in a moment. As before, controlled constants, $c,c'$ etc., are constants that depend only on the basic constants. We make the additional assumption that 
$k_1\ne k_2$. We expect that this assumption can be relaxed. Here we set $\lambda\in\Omega'$, where $\Omega'=\{\lambda\in\mathbb{C}:\Re \lambda \geq C_\lambda, |\Im \lambda|\leq c_\lambda\Re \lambda\}$, where $C_\lambda$ is \deleted{large enough controlled constant} \added{a sufficiently large controlled constant} and $c_\lambda$ is a small enough controlled constant. We now suppose that $B\in\Lambda$, where $\Lambda=\{B\in\mathbb{C}:|\Re B|\leq C_B, |\Im B|\leq c_B\}$, for a small enough controlled constant $c_B$. Note that since $B\in\CC$ is allowed to be complex, the quantity $f_+$, defined in \eqref{eq:heat_kernel_parameters}, takes on complex values. 

\begin{thm}\label{thm:C11}
    For $(x,y,\lambda,B,\mu)\in\mathbb{R}^2\times\mathbb{R}^2\times\Omega'\times\Lambda\times\{
        \mu\in\mathbb{C}:\Re \mu<\frac{\Re{f_+}}{2}+c_\star\}$ with $x\neq y$, we have
    \begin{equation} \label{eq:second_main_thm}
        |\tilde{G}(x,y,\lambda,\mu)| = \bigg|\int_0^\infty \frac{e^{\mu s}}{\lambda}\big[q(x,y,s,\lambda) - e^{-\frac{f_+}{2} s}\psi_0(x)\psi_0^*(y)\big] ds\bigg| \leq CD(c|\lambda|\ |x-y|(|x|+|y|)),
    \end{equation}
    For $\tilde{G}(x,y,\lambda,\mu)$ defined in \eqref{eq:resolvent_definition}, and $D(\cdot)$ defined in \eqref{eq:D_function_definition}. Here $C,c,\hat{c}$ are controlled constants.
\end{thm}
\added{The proof of \cref{thm:C11} is not used in the present paper and is omitted for brevity; see the arXiv version of this manuscript \cite{fefferman2025magneticdoublewellslowerbounds} for full details.}
\bigskip

\subsection{Analyticity of the Resolvent} \label{sec:analyticity_of_resolvent}
In this section we use the first and second main theorems to prove that the resolvent $\tilde{G}(x,y,\lambda,\mu)$ is analytic in $\lambda$ and $B$ in the regions $\Omega, \Omega'$ and $\Lambda$ defined above.

\begin{thm}\label{thm:analyticity of the MHO resolvent}
    Let $(x,y,\lambda,\mu)\in\mho$,\\ where $\mho=\{(x,y,\lambda,\mu)\in\mathbb{R}^2\times\mathbb{R}^2\times\Omega\times\mathbb{C}:x\neq y, \Re\mu<\frac{\Re{f_+}}{2}+c_\star\}$, where  $c_\star>0$ is taken sufficiently small. Then the integral
    \begin{equation}
        \tilde{G}(x,y,\lambda,\mu) = \int_0^\infty \frac{e^{\mu s}}{\lambda}\big[q(x,y,s,\lambda) - e^{-\frac{f_+}{2} s}\psi_0(x,\lambda)\psi_0^*(y,\lambda^*)\big] ds
    \end{equation}
    converges for $(x,y,\lambda,\mu)\in\mho$; the function $\tilde{G}(x,y,\lambda,\mu)$ is continuous on $\mho$ and analytic in $(\lambda,\mu)$ for fixed $(x,y)$. For $\lambda\in\Omega', B\in\Lambda$ and with $\Re\mu< \frac{\Re{f_+}}{2}+c_\star$, the theorem also holds and, moreover, $\tilde{G}(x,y,\lambda,\mu)$ is analytic in $B$.
\end{thm}
\begin{proof} 
    We first define
    \begin{equation}
        \tilde G_N(x,y,\lambda,\mu) = \int_{1/N}^N \frac{e^{\mu s}}{\lambda}\left[q(x,y,s,\lambda) - e^{-\frac{f_+}{2} s}\psi_0(x, \lambda)\psi_0^*(y, \lambda^*)\right]ds,
    \end{equation}
    and will show that $\lim_{N\to\infty}\tilde G_N(x,y,\lambda,\mu) = \tilde G(x,y,\lambda,\mu)$ uniformly in $(x,y,\lambda,\mu)$ varying in any compact subset of $\mho$. This will give the desired result, since $\tilde G_N(x,y,\lambda,\mu)$ is analytic in $(\lambda,\mu)$ for fixed $(x,y)$ and continuous in $(x,y,\lambda,\mu)$.

    To show the uniform convergence, we need to show that for any $\epsilon>0$ and any compact subset $K\subset\mho$, there exists an $N_0>0$ such that for all $N>N_0$ and all $(x,y,\lambda,\mu)\in K$ we have
    \begin{align} \label{eq:analyticity_uniform_convergence}
      &  |\tilde G(x,y,\lambda,\mu) - \tilde G_N(x,y,\lambda,\mu)| 
     \leq \int_0^{1/N}\Big|e^{\mu s}\left[q(x,y,s,\lambda) - e^{-\frac{f_+}{2} s}\psi_0(x, \lambda)\psi_0^*(y, \lambda^*)\right]\Big| \notag\\
      &  + \int_N^\infty \Big|e^{\mu s}\left[q(x,y,s,\lambda) - e^{-\frac{f_+}{2} s}\psi_0(x, \lambda)\psi_0^*(y, \lambda^*)\right]\Big| < \epsilon.
    \end{align}
    To prove \eqref{eq:analyticity_uniform_convergence}, we fix $x_0,y_0,\lambda_0,\mu_0\in\mho$. Then for small $\delta>0$ depending on $x_0,y_0,\lambda_0,\mu_0$, we prove both statements:
    \begin{align}
        &\lim_{\tau\to0}\int_0^\tau \Big|e^{\mu s}\left[q(x,y,s,\lambda) - e^{-\frac{f_+}{2} s}\psi_0(x, \lambda)\psi_0^*(y, \lambda^*)\right]\Big| ds = 0,\quad \text{and}\notag\\
        &\lim_{T\to\infty}\int_T^\infty \Big|e^{\mu s}\left[q(x,y,s,\lambda) - e^{-\frac{f_+}{2} s}\psi_0(x, \lambda)\psi_0^*(y, \lambda^*)\right]\Big| ds = 0,
    \end{align}
    uniformly for $(x,y,\lambda,\mu)$ satisfying
    \begin{equation}\label{eq:uniform_convergence_conditions}
        |x-x_0|,|y-y_0|,|\lambda-\lambda_0|,|\mu-\mu_0|<\delta.
    \end{equation}
    This will establish and complete the proof of the theorem. 

    Since $x_0\neq y_0$ for $(x_0,y_0,\lambda_0,\mu_0)\in\mho$, \eqref{eq:uniform_convergence_conditions} implies that 
    \begin{equation}
        |x-y| > \frac{|x_0-y_0|}{2} > 0, \quad \text{and}\quad |x^2|+|y^2| < 2(|x_0|^2+|y_0|^2),
    \end{equation}
    provided $\delta$ is small enough in \eqref{eq:uniform_convergence_conditions}. Therefore, we can use the first main theorem, specifically \cref{lem:regions_1_2_3}, and estimates \eqref{eq:counterterm_region_1}, \eqref{eq:counterterm_region_2}, \eqref{eq:counterterm_region_3} to estimate the first integral:
    \begin{align}
        &\int_0^\tau \Big|e^{\mu s}\left[q(x,y,s,\lambda) - e^{-\frac{f_+}{2} s}\psi_0(x, \lambda)\psi_0^*(y, \lambda^*)\right]\Big| ds \notag\\
        &\leq  \int_0^\tau \frac{C}{s}\exp\left(-c|\lambda|\left[s(|x|^2+|y|^2)+\frac{|x-y|^2}{s}\right]\right)ds + \int_0^\tau C\exp\left(-c|\lambda|(|x|^2+|y|^2)\right) ds \notag\\
        &\leq \int_0^\tau \frac{C}{s}\exp\left(-c|\lambda|\frac{|x_0-y_0|^2}{s}\right)ds + \int_0^\tau C\exp\left(-c|\lambda|(|x_0|^2+|y_0|^2)\right) ds \to 0 \quad \text{as } \tau\to0,
    \end{align}
    uniformly for $(x,y,\lambda,\mu)$ satisfying \eqref{eq:uniform_convergence_conditions}.

    For the second integral, we use the analysis of regions 5 and 6 in the proof of the first main theorem, \cref{lem:region_5,lem:region_6}, to prove that
    \begin{equation}
        \int_T^\infty \Big|e^{\mu s}\left[q(x,y,s,\lambda) - e^{-\frac{f_+}{2} s}\psi_0(x, \lambda)\psi_0^*(y, \lambda^*)\right]\Big| ds \to 0 \quad \text{as } T\to\infty,
    \end{equation}
    uniformly for $(x,y,\lambda,\mu)$ satisfying \eqref{eq:uniform_convergence_conditions}. This completes the proof of the theorem.
\end{proof}

\subsection{Proof of \cref{prop:MHO analyticity and bounds}}\label{subsec:proof_MHO_prop}
\deleted{The previous placeholder subsection on derivative bounds has been removed.}
\begingroup
\ifshowchanges\color{ct_green}\fi
\begin{proof}
Set
\[
\widehat\Omega=\left\{(\lambda,\mu)\in\mathbb C^2:\lambda\in\Omega,\ \Re\mu<\frac{f_+}{2}+c_\star,\ \mu\ne\frac{f_+}{2}\right\}
\]
and
\[
\widehat\Omega_0=\left\{(\lambda,\mu)\in\mathbb R\times\mathbb C:\lambda>\Lambda,\ \Re\mu<\frac{f_+}{2}\right\}\subset\widehat\Omega.
\]
From \cref{eq:LT_of_heat,eq:resolvent_definition}, for $(\lambda,\mu)\in\widehat\Omega_0$ and $x\ne y$,
\[
(H_\lambda-\lambda\mu)^{-1}(x,y)=\widetilde G(x,y,\lambda,\mu)+\frac1\lambda\left(\mu-\frac{f_+}{2}\right)^{-1}\psi_0(x,\lambda)\psi_0^*(y,\lambda^*).
\]
Define operators $T_1(\lambda,\mu)$, $T_2(\lambda,\mu)$ and $T(\lambda,\mu):L^2(\mathbb R^2)\to L^2(\mathbb R^2)$ by
\[
T_1(\lambda,\mu)f(x)=\int_{\mathbb R^2}\widetilde G(x,y,\lambda,\mu)f(y)\,dy,
\]
\[
T_2(\lambda,\mu)f(x)=\frac1\lambda\left(\mu-\frac{f_+}{2}\right)^{-1}\psi_0(x,\lambda)\int_{\mathbb R^2}\psi_0^*(y,\lambda^*)f(y)\,dy,
\]
and $T(\lambda,\mu)=T_1(\lambda,\mu)+T_2(\lambda,\mu)$.
By \cref{thm:first main thm,thm:analyticity of the MHO resolvent}, the map
\[
\widehat\Omega\ni(\lambda,\mu)\longmapsto T_1(\lambda,\mu)\in\calB(L^2(\mathbb R^2))
\]
is analytic and satisfies $\|T_1(\lambda,\mu)\|\le C/|\lambda|$. Since $|B|$ is smaller than a sufficiently small constant, $\|\psi_0(\cdot,\lambda)\|_{L^2}\le C$ for $\lambda\in\Omega$; hence $T_2$ is analytic on $\widehat\Omega$ and
\[
\|T_2(\lambda,\mu)\|\le \frac{C}{|\lambda|}\left|\mu-\frac{f_+}{2}\right|^{-1}.
\]
Consequently, $T$ is analytic on $\widehat\Omega$ and
\[
\|T(\lambda,\mu)\|\le \frac{C}{|\lambda|}\left(1+\left|\mu-\frac{f_+}{2}\right|^{-1}\right).
\]
For $(\lambda,\mu)\in\widehat\Omega_0$, the displayed resolvent identity shows
\[
(H_\lambda-\lambda\mu)T(\lambda,\mu)=\Id.
\]
By analytic continuation, the same identity holds for all $(\lambda,\mu)\in\widehat\Omega$. In particular,
\[
T(\lambda,\mu)=(H_\lambda-\lambda\mu)^{-1}
\]
for $\lambda\in(\Lambda,\infty)$, $\Re\mu<f_+/2+c_\star$, and $\mu\ne f_+/2$.
For the contour used in \cref{eq:z_lam-def}, the relevant parameter is
\[
\mu_\lambda(\xi)=\frac{z_\lambda(\xi)+\lambda^2}{2\lambda}=\frac{e_0^{\rm MHO}}{2}+\frac{c_{\mathrm{ctr}}}{2}(e_1^{\rm MHO}-e_0^{\rm MHO})\xi,
\]
not $z_\lambda(\xi)/\lambda$. With $c_{\mathrm{ctr}}>0$ chosen sufficiently small, this satisfies $\Re\mu_\lambda(\xi)<f_+/2+c_\star$ and $|\mu_\lambda(\xi)-f_+/2|>c$ uniformly for $\xi\in\mathbb S^1$. This proves Part 1 of \cref{prop:MHO analyticity and bounds}.

It remains to prove Part 2. Let
\eq{
u = r_\lambda^{\mathrm{MHO}}(\lambda^2+z_\lambda)f
} where $f\in L^2(\RR^2)$ with $\supp(f)\subseteq T$.

Applying \cref{thm:first main thm} we find
\eql{\label{eq:main conclusion within Combes-Thomas for MHO}
|u(x)|
\le
C \exp\br{-c\abs{\lambda}\br{\dist(x,T)}^2}\,\norm{f}_{L^2},
\qquad
\text{for } \mathrm{dist}(x,T) > \abs{\lambda}^{-1/2}.
}

Fix $x_0 \in S$ and let $B$ denote the disc $B_r(x_0)$, with
\eq{
r := \frac{1}{2\abs{\lambda}\br{1+\sup_{x\in S}\norm{x}}}\,.
} %The point of the first constraint is that in particular this implies that for $x\in B_r(x_0)$\eq{
%\dist(x,T) > \abs{\lambda}^{-1/2}\,.
%} Moreover, the second constraint ensures the uniform boundedness of the elliptic constants appearing below.

Since $f=0$ on $B$, we have
\eql{\label{eq:Hamiltonian acting on state}
\Bigl[
\Bigl(P-\frac{\lambda b}{2}X^\perp\Bigr)^2  -\lambda^2\Id + \frac{1}{2}\lambda^{2}\left\langle X,\Hessian{v}(0) X\right\rangle - z_\lambda\Id
\Bigr]u = 0
\quad \text{on } B\,.
}
% The handwritten value after "E_n =" is slightly unclear, but it appears to be n\lambda.

We rescale from $B$ to the unit disc $B_0$ by setting
\[
\tilde{u}(y) = u(x_0+ry),
\qquad y \in B_0.
\]

\cref{eq:Hamiltonian acting on state} then takes the form
\[
\bigl[-\Delta_y + \vec{\beta}(y)\cdot \nabla_y + \gamma(y)\bigr]\tilde{u}(y)=0
\qquad \text{on } B_0,
\]
for 
\eq{
\beta(y) := \ii b \lambda r (x_0+r y)^\perp
} and 
\eq{
\gamma(y) := r^2 \br{\frac{\lambda^2b^2}{4}\abs{x_0+ry}^2+\frac12\lambda^2\ip{x_0+ry}{\Hessian{v}(0) \br{{x_0+ry}}}-\lambda^2 -z_\lambda}
}

Standard elliptic theory (see e.g. \cite[Theorem 6.2]{GilbargTrudinger2001}) tells us that
\[
\bigl|\partial_y^\alpha \tilde{u}(0)\bigr|
\le
C_\alpha\br{1+\norm{\beta}_{C^1(B_0)}+\norm{\gamma}_{L^\infty(B_0)}}^{2} \sup_{B_0} |\tilde{u}|
\qquad \text{for } |\alpha|\le 2\,.
\]

Returning from $\tilde{u}$ to $u$, we find that
\[
\bigl|\partial_x^\alpha u(x_0)\bigr|
\le
C r^{-|\alpha|}\sup\bigl\{|u(x)|:x\in B(x_0,r)\bigr\}.
\]

Together with \cref{eq:main conclusion within Combes-Thomas for MHO}, this completes the proof of \cref{eq:Combes-Thomas type estimate for MHO resolvent}, since we may absorb the polynomial factor $r^{-\abs{\alpha}}$ using the Gaussian decay.
\end{proof}
\endgroup
\section{Approximations with the MHO}

\added{The estimates below are used only with $b$ in a fixed sufficiently small interval and $\lambda\in\Omega_{\Lambda,\ve}$.}
\bigskip

Consider the Landau Hamiltonian
\eql{
H^{\mathrm{Landau}}_\lambda := \left(P-\frac{1}{2}b\lambda X^{\perp}\right)^{2}
} with $b,\lambda>0$ and on it the magnetic harmonic oscillator
\eql{h_{\lambda}^{\rm MHO }:=H^{\mathrm{Landau}}_\lambda+\frac{1}{2}\lambda^{2}\left\langle x,\Hessian{v}\left(0\right)x\right\rangle\,.} Here $\Hessian{v}\left(0\right)>0$ is some $2\times 2$ matrix.

Note that by scaling, \eql{\mathfrak{U}_{\sqrt{\lambda}}^{\ast}h_{\lambda}^{\rm MHO }\mathfrak{U}_{\sqrt{\lambda}}=\lambda h^{\rm MHO }}
for a dilation operator $\mathfrak{U}_{\sqrt{\lambda}}$ (only unitary if $\lambda$ is real) defined in \cref{eq:dilation operators} below, and with $h^{\rm MHO } := h^{\rm MHO }_1$. Let us denote the eigenvalues of $h^{\rm MHO }$ as $e_j^{\rm MHO }$, $j=0,1,\cdots$ in ascending order. 

We also have our one-well Hamiltonian
\eql{
h_\lambda = H^{\mathrm{Landau}}_\lambda + \lambda^2 v(X)
} with $v$ having a unique non-degenerate minimum at zero. In particular, \deleted{In particular,} \eq{v\left(x\right)\stackrel{x\to0}{=}-1+\frac{1}{2}\left\langle x,\Hessian{v}\left(0\right)x\right\rangle +\mathcal{O}\left(\norm{x}^{3}\right)}
and where $\Hessian{v}$ is the Hessian of $v$ and we assume $\Hessian{v}\left(0\right)>0$.

Then the analysis in \cite{Matsumoto_1994} implies that for large \emph{real} $\lambda$, if $e_{\lambda,j}$ are the eigenvalues of $h_\lambda$, \eql{\label{eq:MHO one-well eigenvalues}e_{\lambda,j}=-\lambda^{2}+e_{j}^{\rm MHO }\lambda+O\left(\lambda^{\frac{1}{2}}\right)\,.}

Concerning the eigenvectors,  \cite{Matsumoto_1994} does not provide an asymptotic expansion. However  the ideas presented in the non-magnetic \cite{Simon_1983_AIHPA_1983__38_3_295_0} can most likely be generalized to the magnetic case. We don't need such precision here and present a self-contained proof for the following

\begin{lem}\label{lem:MHO ground state approximates one-well ground state}
Denote by  $(e_\lambda,\vf_\lambda)$ and $(e_\lambda^{\rm MHO}, \vf_\lambda^{\rm MHO} )$, respectively, the normalized ground state eigenpairs of the single well magnetic Hamiltonian $h_\lambda$ (see \cref{eq:single-well Hamiltonian}) and the magnetic harmonic oscillator Hamiltonian $h_\lambda^{\rm MHO}$  (see \cref{eq:h-MHO}). Let $P_\lambda=\vf_\lambda\otimes\vf^*_\lambda$ denote the orthogonal projector onto the ground state eigenspace of $h_\lambda$, and $P_\lambda^\perp=\Id - P_\lambda$.

Then,  for all $\lambda$ real and sufficiently large, 
 \begin{equation}\big\| P_\lambda^\perp \vf_\lambda^{\rm MHO}\big\| \lesssim \lambda^{-1/2}\,.
    \label{eq:Pperp-psi-bound}
    \end{equation}
\end{lem}
\begin{proof}
To bound the norm of $P_\lambda^\perp \vf_\lambda^{\rm MHO}$, we observe: If $P^\perp_\lambda\psi=\psi$, i.e. $\psi\perp \vf_\lambda$, then
\begin{equation}\label{eq:Pperp-psi}
\|P^\perp_\lambda\psi\|\ \le\ \frac{1}{c_{\rm gap}(\lambda)}\ \big\|(h_\lambda-e_\lambda) P^\perp_\lambda\psi\big\| ,
\end{equation}
where $c_{\rm gap}$, the distance between the first and second eigenvalues of $h_\lambda$ is given, using \cref{eq:MHO one-well eigenvalues}, by
\[
c_{\rm gap}(\lambda) = \br{e_1^{\rm MHO} - e_0^{\rm MHO}}\lambda+O(\lambda^{1/2})\ .
\]
The bound \cref{eq:Pperp-psi} is a consequence of the bound 
\eq{
\norm{\br{h_\lambda - e_\lambda\Id}\psi} \geq \frac{1}{\norm{\psi}}\ip{\psi}{\br{h_\lambda-e_\lambda\Id}\psi}\ge c_{\rm gap}\ \|\psi\|\ ;
} 
the first inequality follows by the Cauchy-Schwarz inequality, and second from the spectral theorem.  
Applying \cref{eq:Pperp-psi} with $\psi=\vf_\lambda^{\rm MHO}$, we have 
\begin{equation}\label{eq:Pperp-phi-MHO}
\|P^\perp_\lambda \vf_\lambda^{\rm MHO}\|\ \le\    \frac{1}{c_{\rm gap}(\lambda)}\ \big\|(h_\lambda-e_\lambda\Id) \vf_\lambda^{\rm MHO}\big\| ,
\end{equation}
To bound the right hand side of \cref{eq:Pperp-phi-MHO}, note first that
\begin{equation} h_\lambda-e_\lambda\Id = \big(h_\lambda^{\rm MHO}-e_\lambda^{\rm MHO}\Id\big) + \mathcal{O}(\lambda^{+1/2})\ \Id  +\lambda^2 v_{\rm err}(x), \label{eq:h_lambda-expand}
\end{equation}
where 
$v_{\rm err}(x) := v(x) - \br{-1+\frac{1}{2}\left\langle x,\Hessian{v}\left(0\right)x\right\rangle}$ and 
$\abs{v_{\rm err}(x)}=\mathcal{O}(\norm{x}^3)$ as $x\to0$. 
Since $h_\lambda^{\rm MHO}\vf_\lambda^{\rm MHO} = e_\lambda^{\rm MHO}\vf_\lambda^{\rm MHO}$, applying the expansion \cref{eq:h_lambda-expand} to $\vf_\lambda^{\rm MHO}$, and then taking norms, yields
\begin{equation}
\big\| \big(h_\lambda-e_\lambda\Id \big)\vf_\lambda^{\rm MHO} \big\|\ \lesssim\  \lambda^{+1/2}\ +\  \lambda^2 \norm{v_{\rm err}(X)\vf_\lambda^{\rm MHO}}.
\label{eq:residual-bound}
\end{equation}
  To bound the latter term, recall the form of the exact ground state:
\eql{\label{eq:exact ground state of MHO}
\vf_\lambda^{\rm MHO}(x) := C \sqrt{\lambda}\exp\br{-\lambda \ip{x}{ S x}}\qquad(x\in\RR^2)
} where  $S$ is  positive definite  and $C<\infty$, both depend on $D^2v(0)$. Using \cref{eq:exact ground state of MHO}, and the behavior of $v_{\rm err}(x)$ near $x=0$, we have the estimate
\eql{\label{eq:potential remainder on MHO ground state}
\lambda^2 \norm{v_{\rm err}(X)\vf_\lambda^{\rm MHO}} \lesssim \lambda^{1/2} + \ee^{-c\lambda }\lesssim \lambda^{1/2}\,.
}
Substitution of \cref{eq:potential remainder on MHO ground state} into
\cref{eq:residual-bound}, and the resulting bound into \cref{eq:Pperp-phi-MHO} yields 
 \cref{eq:Pperp-psi-bound}. This completes the proof of \cref{lem:MHO ground state approximates one-well ground state}.
\end{proof}

Next, we deal with the double-well Hamiltonian,
\eql{
H_\lambda \equiv H^{\mathrm{Landau}}_\lambda+\lambda^2 v(X+d)+\lambda^2 v(X-d)\,.
} By the same assumptions on $v$ as above and \cite{Matsumoto_1994}, if the eigenstates of $H_\lambda$ are given by $E_{j,\lambda}$, then, a-priori,
\eql{
E_{0,\lambda} \approx E_{1,\lambda} \approx -\lambda^2 + e_0^{\rm MHO}\lambda + O(\lambda^{1/2})
} and the next eigenvalue is at distance $c_{\rm gap}$ (as above) away. Let $\calV_\lambda$ be the eigenspace of $H_\lambda$ associated to both eigenvalues $E_{0,\lambda}, E_{1,\lambda}$ and $\Pi_\lambda$ the associated self-adjoint projection.

Recall the magnetic translation operators given in \cref{eq:magnetic translations}. These operators are unitary if $\lambda\in\RR$ and otherwise generally not even bounded; in factorized form the two exponentials commute. $\widehat{R}^z$ commutes with the magnetic kinetic energy $H^{\rm Landau}_{\lambda}$ and obeys \cref{eq:magnetic translations shift potential} as well as \cref{eq:magnetic translations shift Hamiltonian}.

Then we also have
\begin{lem}\label{lem:translates of MHO ground state approximates double-well ground state}
    For all $\lambda$ real and sufficiently large, \eql{
        \norm{\Pi_\lambda^\perp \widehat{R}^{\pm d}\vf_\lambda^{\rm MHO}}\lesssim \lambda^{-1/2}\,.
    } 
\end{lem}
\begin{proof}
    We have
    \eq{
        \br{H_\lambda - \br{-\lambda^2+e_\lambda^{\rm MHO}}\Id}\widehat{R}^{ d}\vf_\lambda^{\rm MHO}
        &= \widehat{R}^{ d} \br{h_\lambda - \br{-\lambda^2+e_\lambda^{\rm MHO}}\Id}\vf_\lambda^{\rm MHO}\notag\\
        &\qquad{}+\widehat{R}^{ d}\lambda^2v(X-2d)\vf_\lambda^{\rm MHO}
    } so by the proof of \cref{eq:potential remainder on MHO ground state}, we have
    \eq{
    \norm{\br{H_\lambda - \br{-\lambda^2+e_\lambda^{\rm MHO}}\Id}\widehat{R}^{ d}\vf_\lambda^{\rm MHO}}
    \lesssim \lambda^{1/2} + \exp\br{\addedm{-2cd\lambda}} \lesssim \lambda^{1/2}\,.
    } 

    Otherwise the proof proceeds as in \cref{eq:potential remainder on MHO ground state} and we get the same result.
\end{proof}

		\pagebreak
		\begingroup
		\let\itshape\upshape
		\printbibliography
		\endgroup

    \end{document}